\documentclass[a4paper, 11pt]{article}

\usepackage{standalone, xcolor, enumitem}
\usepackage{amsmath, amssymb, amsthm, physics, bbm, mathtools, stmaryrd, leftidx}
\usepackage{thmtools}
\usepackage{longtable}
\usepackage{hyperref}

\usepackage{lmodern}

\usepackage[normalem]{ulem}
\usepackage{accents}

\usepackage{pgfplots}
\pgfplotsset{compat=1.16}

\usetikzlibrary{arrows, decorations.pathreplacing, cd, patterns, calc, decorations.markings, backgrounds, tikzmark}

\usepackage{quiver}

\usepackage[new]{old-arrows}
\usepackage{changepage}

\usepackage[bottom]{footmisc}

\usepackage{glossaries}
\glsdisablehyper

\theoremstyle{plain}
\newtheorem{theorem}{Theorem}[section]
\newtheorem{proposition}[theorem]{Proposition}
\newtheorem{lemma}[theorem]{Lemma}
\newtheorem{corollary}[theorem]{Corollary}

\numberwithin{equation}{section}
\theoremstyle{definition}
\newtheorem{definition}[theorem]{Definition}

\newtheorem{remark}[theorem]{Remark}

\numberwithin{figure}{section}

\setlist[enumerate]{topsep = 1ex, leftmargin=.7cm, itemsep= -2pt}

\allowdisplaybreaks

\definecolor{color1}{HTML}{1b9e77}
\definecolor{color2}{HTML}{F1A340}
\definecolor{color3}{HTML}{5e3c99}

\hypersetup{
  colorlinks = true, 
  urlcolor = color3, 
  linkcolor = color1, 
  citecolor = color3 
}

\newacronym{fsmooth}{Fr-smooth}{Fr\"olicher smooth}

\newcommand{\ii}{\mathrm{i}\,}
\newcommand{\R}{\mathbb{R}}
\newcommand{\Rp}{\mathbb{R}_{+}}
\newcommand{\Z}{\mathbb{Z}}
\newcommand{\C}{\mathbb{C}}
\DeclareMathOperator{\id}{{\mathbbm{1}}}
\newcommand{\dz}{\dd z \dd \bar z}
\newcommand{\detz}[1]{{\det}_\zeta \, {#1}}
\newcommand{\dvol}[1]{\dd V_{#1}}
\newcommand{\dboundary}[1]{\dd \ell_{#1}}

\DeclareMathOperator{\Det}{Det}

\DeclareMathOperator{\vspan}{span}

\newcommand{\A}{\mathbb{A}}
\newcommand{\F}{\mathbb{F}}
\newcommand{\boldone}{\mathbbm{1}}

\newcommand{\disk}{\mathbb{D}} 
\newcommand{\odisk}{\setsuchthat{z \in \C}{|z| < 1}} 
\newcommand{\cdisk}{\setsuchthat{z \in \C}{|z| \leq 1}} 
\newcommand{\jodisk}{\setsuchthat{z \in \C}{|z| > 1}} 
\newcommand{\jcdisk}{\setsuchthat{z \in \C}{|z| \geq 1}} 
\newcommand{\uhp}{\setsuchthat{z \in \C}{\Im z > 0}}

\newcommand{\blank}{\:\cdot\:}
\newcommand{\setsuchthat}[2]{\left\{ {#1} \:\middle|\: {#2} \right\}}
\newcommand{\setsuchthatinline}[2]{\{ {#1} \:|\: {#2} \}}
\newcommand{\restrict}[2]{{#1}\vert_{{#2}}}
\newcommand{\const}{\mathrm{(const.)}}

\DeclareMathOperator{\Diffpan}{{Diff}_+^{\an}(S^1)}
\DeclareMathOperator{\DiffC}{{Def}_{\C}(S^1)}
\DeclareMathOperator{\DefC}{\DiffC}

\DeclareMathOperator{\vect}{Vect}
\newcommand{\VectR}{\vect^{\an}_{\R}(S^1)}
\newcommand{\VectC}{\vect^{\an}_{\C}(S^1)}
\newcommand{\Witt}{\VectC}
\newcommand{\Detrc}{\Det_{\Rp}^\charge}
\newcommand{\DefCnbhd}[1]{U({#1})}
\DeclareMathOperator{\univalent}{Univ}

\DeclareMathOperator{\mob}{PSL(2, \C)}
\DeclareMathOperator{\mobZ}{PSL(2, \Z)}
\DeclareMathOperator{\psl}{PSL(2, \R)}
\newcommand{\mobvect}{\mathfrak{sl}(2, \C)}
\newcommand{\genus}{\mathsf{g}}
\newcommand{\boundaries}{\mathsf{b}}
\newcommand{\moduli}[2]{\mathcal{M}_{{#1}, {#2}}}
\newcommand{\annuli}{\moduli{0}{2}}
\newcommand{\modulimob}[2]{\mathcal{M}_{{#1}, {#2}}^{\textsf{proj}}}
\newcommand{\modulimobf}[2]{\mathcal{M}_{{#1}, {#2}}^{\textsf{f}}} 

\newcommand{\moduligb}{\moduli{\genus}{\boundaries}}
\newcommand{\moduligbmob}{\modulimob{\genus}{\boundaries}}
\newcommand{\moduligbmobf}{\modulimobf{\genus}{\boundaries}}

\newcommand{\disks}{\moduli{0}{1}}
\newcommand{\spheres}{\moduli{0}{0}}
\newcommand{\pants}{\moduli{0}{3}}
\newcommand{\handles}{\moduli{1}{1}}
\newcommand{\tori}{\moduli{1}{0}}
\newcommand{\jmoduli}[3]{\mathcal{M}_{{#1}, {#2}, {#3}}}
\newcommand{\jmoduligb}[1]{\jmoduli{\genus}{\boundaries}{{#1}}}
\newcommand{\projj}[1]{\pi_{{#1}}}
\newcommand{\hmoduli}[2]{\mathcal{M}^{\textsf{hyp}}_{{#1}, {#2}}}
\newcommand{\hmoduligb}{\hmoduli{\genus}{\boundaries}}
\newcommand{\hpants}{\hmoduli{0}{3}}
\newcommand{\univteichm}{T(1)}

\newcommand{\modulus}[1]{\mathsf{mod}({#1})}

\newcommand{\an}{\mathrm{an}}
\newcommand{\charge}{\mathbf{c}}
\newcommand{\sew}[2]{\:\leftidx{{_{#1}}}{{\infty}}{{_{#2}}}\:}
\newcommand{\sewx}[3]{
\mathchoice
	{\:\leftidx{{_{#2}}}{{\overset{\mathclap{{#1}}}{\infty}}}{{_{#3}}}\:}
	{\:\leftidx{{_{#2}}}{{{\infty}}}{{_{#3}^{{#1}}}}\:}
    {\:\leftidx{{_{#2}}}{{{\infty}}}{{_{#3}^{{#1}}}}\:}
    {\:\leftidx{{_{#2}}}{{{\infty}}}{{_{#3}^{{#1}}}}\:}
}
\newcommand{\sewself}[2]{\:{\infty}{{_{#1, #2}}}\:}
\newcommand{\sewxself}[3]{
\mathchoice
	{\:{{\overset{\mathclap{{#1}}}{\infty}}}{{_{#2, #3}}}\:}
	{\:{{{\infty}}}{{_{#2, #3}^{{#1}}}}\:}
	{\:{{{\infty}}}{{_{#2, #3}^{{#1}}}}\:}
	{\:{{{\infty}}}{{_{#2, #3}^{{#1}}}}\:}
}
\newcommand{\sewall}{\:\underline{\infty}\:}
\newcommand{\sewxall}[1]{
\mathchoice
	{\:\underline{\overset{\mathclap{{#1}}}{\infty}}\:}
	{\:\underline{\infty^{{#1}}}\:}
	{\:\underline{\infty^{{#1}}}\:}
	{\:\underline{\infty^{{#1}}}\:}
}

\newcommand{\unravelplus}[1]{\:{{#1}}_+\:}
\newcommand{\unravelminus}[1]{\:{{#1}}_-\:}
\newcommand{\unravel}[1]{\:{{#1}}\:}

\DeclareMathOperator{\lfunct}{S_L^0}

\newcommand{\detmulta}[3]{\operatorname{m}_{#1, #2, #3}}
\newcommand{\natisophi}[3]{\operatorname{I}^{#2}_{#3, #1}}

\newcommand{\gelfandfuks}{\omega_{\mathrm{GF}}}
\newcommand{\bott}{\Omega_{\mathrm{BT}}}

\newcommand{\exchange}{\mathsf{flip}}
\newcommand{\evaluation}{\mathsf{ev}}
\newcommand{\applyat}[3]{#1_{#2, #3}}

\newcommand{\param}[1]{
\mathchoice
    {\big(\: {#1} \:\big)}
    {(#1)}
    {(#1)}
    {(#1)}
}
\newcommand{\parameq}[1]{
\mathchoice
    {\big[\: {#1} \:\big]}
    {[#1]}
    {[#1]}
    {[#1]}
}

\makeatletter
\newcommand{\uset}[3][-0.35ex]{%
  \mathrel{\mathop{#3}\limits_{
    \vbox to#1{\kern-7\ex@
    \hbox{$\scriptscriptstyle#2$}\vss}}}}
\makeatother

\newcommand{\diffActingInline}[1]{\,
    \smash{\uset{{#1}}{*}} 
\, }
\newcommand{\diffActing}[1]{\diffActingInline{#1}}

\DeclareMathOperator{\definvset}{{Inv}_{\C}(S^1)}
\DeclareMathOperator{\defmultset}{{Comp}_{\C}(S^1)}

\newcommand{\rotnumberalg}{\mathrm{rot}}
\newcommand{\rotcocycle}{\Omega_{\mathrm{RCR}}}
\newcommand{\rotcocyclealg}{\omega_{\mathrm{rot}}}

\newcommand{\scaling}[1]{\mathsf{sc}_{{#1}}}
\newcommand{\scalinggroup}{\mathrm{Sc}}

\newcommand{\rotation}{\mathsf{r}}

\newcommand{\curves}[1]{\mathcal{C}({{#1}})}

\newcommand{\functions}[1]{\mathcal{F}({{#1}})}

\newcommand{\compdef}[3]{\mathcal{U}_{#1, #2, #3}}

\newcommand{\diffActingEInline}[1]{\,
\smash{\uset{{#1}}{\overset{\RMFE}{*}}}
\, }
\newcommand{\diffActingE}[1]{\diffActingEInline{#1}}

\newcommand{\inversion}{\mathrm{J}}
\newcommand{\invinv}{\mathsf{t}}
\newcommand{\invinvx}[1]{\mathsf{t}_{#1}}

\newcommand{\schwarzian}[1]{\mathcal{S}[{#1}]}

\newcommand{\conj}[1]{{#1}^*} 
\newcommand{\gcc}{{g_{\mathrm{cc}}}} 

\newcommand{\ellpar}[1]{a^\parallel_{{#1}}}
\newcommand{\ellpari}[1]{b^\parallel_{{#1}}}

\newcommand{\lenergy}[1]{\operatorname{I^L}_{\, {#1}}}

\newcommand{\lpotvect}{F} 

\newcommand{\Sr}{S}
\newcommand{\Sra}{S_1}
\newcommand{\Srb}{S_2}
\newcommand{\Srx}{A}
\newcommand{\Src}{S_3}

\newcommand{\anloops}[1]{\mathcal{L}^{\an}({#1})}
\newcommand{\cloops}[1]{\mathcal{L}({#1})}
\newcommand{\loopzero}{\gamma}
\newcommand{\loopone}{\gamma_1}
\newcommand{\looptwo}{\gamma_2}
\newcommand{\gonlyan}{h}

\DeclareMathOperator{\RMFE}{E}
\DeclareMathOperator{\RMFD}{D}



\DeclareRobustCommand\longtwoheadrightarrow
 {\relbar\joinrel\twoheadrightarrow}

\title{Universality of the conformal anomaly}  
\date{October 14, 2026}

\author{Sid Maibach\thanks{Department of Mathematics, ETH Z\"urich, R\"amistrasse 101, 8092 Z\"urich, Switzerland. \\ 
\protect\url{sid.maibach@math.ethz.ch}} \;
and \,
Eveliina Peltola\thanks{Department of Mathematics and Systems Analysis, Aalto University, Otakaari 1, 02150 Espoo, Finland, and Division of Mathematics, University of Cologne, Weyertal 86-90, 50931 Cologne, Germany. \\ 
\protect\url{eveliina.peltola@aalto.fi}}}

\begin{document}
\maketitle

\begin{center}
\begin{minipage}{0.95\textwidth}
\abstract{
	The sewing operation is the fundamental algebraic structure of the (infinite-dimensional) Segal moduli spaces as the conformal cobordism category. We prove a universal property of local $\Rp$-extensions of the sewing operation, characterizing them by one real parameter: their central charge. Our result is an analogue of ``Mumford's theorem,'' that appeared in Segal's monograph~\cite{Segal:Definition_of_CFT, Segal:Definition_of_CFT_collection}, classifying complex extensions, also known as one-dimensional modular functors. Algebraically, the real extensions are characterized by real-valued cocycles on pairs of surfaces under sewing. We identify the disk-disk cocycle as the universal Liouville action, also known as loop Loewner energy.

	The real determinant line bundle encodes the trace anomaly of euclidean two-dimensional conformal field theories (CFT) and the restriction function of Schramm--Loewner evolution (SLE) loop measures. It serves as the guiding example in our setup and proof strategy, providing the universal $\Rp$-extension. Our results thus give a mathematical explanation for the appearance of the central charge in CFT and SLE, and prove the conjectured uniqueness of Brownian loop measures as the canonical restriction function for SLE.
}

\bigskip{}


\bigskip{}


\end{minipage}
\end{center}

\newpage

\tableofcontents

\newpage

\section{Introduction}
The motivation of the present work is to provide a mathematical explanation for the breaking of conformal symmetry both in euclidean two-dimensional \emph{conformal field theories} (CFTs, Section~\ref{section:cft}) and for \emph{Schramm--Loewner evolution} measures (SLE, Section~\ref{section:sle}), both of which are quantified by a constant $\charge \in \R$ called the \emph{central charge}.
To establish this, we prove a universal property of \emph{local $\Rp$-extensions of Segal moduli spaces} (Theorem~\ref{thm:universality}).

Because CFT and SLE are deeply related, it is not surprising that the same universal property applies to both conformal anomalies.
On the SLE side, our results yield the conjectured uniqueness~\cite{Benoist:Classifying_conformally_invariant_loop_measures} of the restriction function of SLE/MKS (Malliavin--Kontsevich--Suhov)
loop measures~\cite{Kontsevich:CFT_SLE_and_phase_boundaries, Kontsevich-Suhov:On_Malliavin_measures_SLE_and_CFT} on Riemann surfaces (Theorem~\ref{thm:restriction_function}), usually expressed using Brownian loop measure.
On the CFT side, we show that any local covariance under Weyl transformations induces a local $\Rp$-extension, 
which then is isomorphic to the real determinant line of Friedan~\&~Shenker~\cite{Friedan-Shenker:The_analytic_geometry_of_two-dimensional_conformal_field_theory}, defined by the conformal anomaly formula~\eqref{eq:conformal_anomaly} of Polyakov~\&~Alvarez~\cite{Polyakov:Quantum_geometry_of_bosonic_strings,Alvarez:Theory_of_strings_with_boundaries_fluctuations_topology_and_quantum_geometry}.
As an intermediate step of independent interest, we prove that conformal welding relates the disk-disk cocycle of any local $\Rp$-extension to the universal Liouville action~\cite{Takhtajan-Teo:Weil-Petersson_metric_on_the_universal_Teichmuller_space}, also known as loop Loewner energy~\cite{Wang:Equivalent_descriptions_of_Loewner_energy,Rohde-Wang:Loewner_energy_of_loops_and_regularity_of_driving_functions} (Theorem~\ref{thm:lpot}).

\subsection{Universal property of local \texorpdfstring{$\Rp$}{R+}-extensions of Segal moduli spaces}
\label{section:universal_property}

In his seminal manuscript~\cite{Segal:Definition_of_CFT, Segal:Definition_of_CFT_collection}, Segal introduced moduli spaces of compact connected genus~$\genus$ surfaces with~$\boundaries$ enumerated and analytically parametrized boundary components.
We denote these moduli spaces by $\moduligb$, and call them \emph{Segal moduli spaces}.
From the start, the purpose of these moduli spaces has been to use them as morphisms in the domain category for a functorial axiomatization of CFTs.
Given a pair of surfaces, they may be ``composed'' by identifying (gluing, sewing, or welding) one boundary component on the first surface with one boundary component on the second, using the boundary parametrizations.
This defines the \emph{sewing operation}
\begin{equation}
	\blank \sew{j}{k} \blank \colon \moduli{\genus_1}{\boundaries_1} \times \moduli{\genus_2}{\boundaries_2} \to \moduli{\genus_1 + \genus_2}{\boundaries_1 + \boundaries_2 - 2}, \qquad 1 \leq j \leq \boundaries_1, 1 \leq k \leq \boundaries_2,
	\label{eq:sewing_moduligb}
\end{equation}
and a \emph{self-sewing operation} $\sewself{j}{k} \blank \colon \moduligb \to \moduli{\genus + 1}{\boundaries - 2}$ for distinct $1 \leq j \neq k \leq \boundaries$.
We refer to the summary in Section~\ref{section:setup} and to our companion work~\cite{Maibach-Peltola:Complex_deformations_of_the_circle} for detailed definitions (including the smooth Fr\"olicher structure of the Segal moduli spaces). 

Due to the conformal anomaly of CFTs, Segal's axioms allow for projective functors (see also~\cite{Moore-Seiberg:Classical_and_quantum_conformal_field_theory}).
The definition of ``projectiveness'' is analogous to projective representations of a group $G$, which are equivalent to (non-projective) representations of central extensions~$\hat{G}$ of the group by an abelian group $H$, i.e.,~short exact sequences $H \hookrightarrow \hat{G} \twoheadrightarrow G$ such that $H$ maps into the center of $\hat{G}$.
Segal defined the notion of \emph{modular functor} to absorb the projectiveness of a CFT functor (see~\cite{Segal:Two-dimensional_conformal_field_theories_and_modular_functors} for a brief account).
In addition, modular functors have been related to topological quantum field theory (TQFT) by Witten~\cite{Witten:QFT_and_Jones_polynomial} (see also~\cite{Reshetikhin-Turaev:Invariants_of_3-manifolds_via_link_polynomials_and_quantum_groups,Turaev:Modular_categories_and_3-manifold_invariants,Bojko-Kirillov:Lectures_on_tensor_categories_and_modular_functors}), and to vertex operator algebras (VOA) by Huang~\cite{Huang:2D_Conformal_geometry_and_VOAs, Huang:Genus-zero_modular_functors_and_intertwining_operator_algebras}.

Considering the algebraic structure of the totality $(\moduligb)_{\genus, \boundaries \geq 0}$ of the Segal moduli spaces under the sewing operations, a modular functor can be understood as a central extension of the sewing operation by vector spaces (see Section~\ref{section:modular_functor})\footnote{
Since this is in the context of chiral CFTs (see~\cite{Henriques:The_functorial_approach_to_chiral_2D_CFT} for an introduction), Segal's definition considers complex vector bundles, which are assumed to be (anti-)holomorphic.
Finite-dimensional complex vector spaces correspond to rational chiral CFTs. In contrast, we consider full CFTs in the present article.}.
In the case of the sewing operation, the notions of ``center'' and ``exactness'' do not seem meaningful --- yet, a modular functor makes sense of the projection $\hat{G} \twoheadrightarrow G$ as a vector bundle over each~$\moduligb$.

The fundamental example of a modular functor is the complex determinant line bundle~\cite{
	Quillen:Determinants_of_Cauchy-Riemann_operators_over_Riemann_surface,
	ACKP:Moduli_spaces_of_curves_and_representation_theory,
	Beilinson-Schechtman:Determinant_bundles_and_Virasoro_algebra,
	Segal:Definition_of_CFT,
	Segal:Two-dimensional_conformal_field_theories_and_modular_functors,
	Segal:Definition_of_CFT_collection,
	Huang:2D_Conformal_geometry_and_VOAs, 
	Bojko-Kirillov:Lectures_on_tensor_categories_and_modular_functors,
	Kriz:On_spin_and_modularity_in_conformal_field_theory,
	Radnell:PhD, 
	Frenkel-Ben-Zvi:Vertex_Algebras_and_Algebraic_Curves, 
	RSS:Quasiconformal_Teichmuller_theory_as_analytical_foundation_for_CFT,
	RSSS:Schiffer_operators_and_calculation_of_a_determinant_line_in_conformal_field_theory}, which is a \emph{complex} one-dimensional modular functor.
After a correspondence with Deligne, Segal presented convincing arguments for a result, which is inspired by a more abstract theorem by Mumford (see~\cite{Knudsen-Mumford:The_projectivity_of_the_moduli_space_of_stable_curves_1} and~\cite[Section~5]{Segal:Definition_of_CFT_collection}): holomorphic one-dimensional modular functors are fully characterized by their central charge.
The subsequent works~\cite[Appendix~D]{Huang:2D_Conformal_geometry_and_VOAs} and~\cite{Kriz:On_spin_and_modularity_in_conformal_field_theory} added more details to the arguments.

Since our work concerns the conformal anomaly of full (non-chiral) CFTs and of 
SLE measures, we define the notion of an $\Rp$-extension of the Segal moduli spaces, which may be interpreted as a \emph{real one-dimensional modular functor}.
Analogously to the complex case, it may be realized by principal $\Rp$-bundles\footnote{
	$\Rp$-extensions of Segal moduli spaces
	induce actual central extensions of complex deformations of the unit circle, as explained in Section~\ref{section:central_extension}.
}
\begin{equation}
	\RMFE(\moduligb) \longrightarrow \moduligb, \qquad \genus, \boundaries \geq 0,
\end{equation}
and \emph{sewing isomorphisms} covering the sewing operation~\eqref{eq:sewing_moduligb},
\begin{equation}
	\blank \sewx{\RMFE}{j}{k\phantom{,j}} \blank \colon \RMFE(\moduli{\genus_1}{\boundaries_2}) \boxtimes \RMFE(\moduli{\genus_2}{\boundaries_2}) \to \RMFE(\moduli{\genus_1 + \genus_2}{\boundaries_1 + \boundaries_2 - 2}).
\end{equation}
See Definition~\ref{def:modular_functor} for details (Fr\"olicher smoothness, self-sewing isomorphisms, associativity, and symmetry), and Section~\ref{section:modular_functor} for properties of $\Rp$-extensions.

It is a natural question to ask whether a version of the above ``Mumford--Segal theorem'' also holds in the real case.
(See also the remarks in our earlier work~\cite[Section~1.3]{Maibach-Peltola:From_the_conformal_anomaly_to_the_Virasoro_algebra}.)
Indeed, we establish this in Theorem~\ref{thm:universality}, assuming the $\Rp$-extensions to be \emph{local}, meaning that the sewing isomorphisms depend only on the parametrizations of the boundary components involved in the sewing.
In Proposition~\ref{prop:local_equiv} in Section~\ref{section:local_equiv}, we present a number of equivalent ways to define the locality.
(Interestingly enough, Remark~\ref{remark:transnumber} suggests that locality might be necessary.)

\begin{theorem}[Universality of local $\Rp$-extensions]
	All local Fr\"olicher smooth $\Rp$-exten\-sions with equal central charges are isomorphic (in the sense of Definitions~\ref{def:modular_functor} and~\ref{def:modular_functor_isomorphism}).
	\label{thm:universality}
\end{theorem}

We prove Theorem~\ref{thm:universality} by combining Theorem~\ref{thm:lpot} below with cohomological computations from earlier work~\cite{Maibach-Peltola:From_the_conformal_anomaly_to_the_Virasoro_algebra, Maibach-Peltola:Complex_deformations_of_the_circle} and with modular and crossing invariance properties (Proposition~\ref{prop:E_prop}).
The proof is completed inductively in $\boundaries \geq 0$ and $\genus \geq 0$ in Section~\ref{section:iso}. 
Note that our notion of $\Rp$-extension does not inherently involve representations of mapping class groups, and compared to Segal's arguments for the complex case, we do not require a connection on the $\Rp$-extension in our proof.

\subsection{Universal Liouville action as a cocycle}

The base cases $\spheres$ and $\disks$ in the inductive proof of Theorem~\ref{thm:universality} involve a single sewing isomorphism, where two disks combine into a sphere.
While the Riemann sphere $\hat{\C} = \C \cup \{\infty\}$ represents the single element of $\spheres$, the Segal moduli space $\disks$ of disks is infinite-dimensional.
The universal cover of $\disks$ is a subset\footnote{
	The boundary parametrizations in Segal moduli space are analytic, whereas universal Teichm\"uller space allows for quasisymmetric (or Weil-Petersson quasisymmetric) regularities.
} of the universal Teichm\"uller space $\univteichm$, which may be interpreted as the Teichm\"uller space of a disk with quasisymmetric boundary parametrization~\cite{RSS:Quasiconformal_Teichmuller_theory_as_analytical_foundation_for_CFT}.
Equivalently, $\univteichm$ can be represented as a space of normalized univalent functions~\cite{Lehto:Univalent_functions_and_Teichmuller_spaces}.
For an element of $\disks$, we first find a representative $\disk \diffActing{1} \phi$ comprising the unit disk with boundary parametrization given by an orientation-preserving analytic diffeomorphism $\phi \in \Diffpan$ of the unit circle (Lemma~\ref{lemma:complex_deformations}).
Then, the sewn sphere $(\disk \diffActing{1} \phi) \sew{1}{1} \disk$ may be uniformized to $\hat{\C}$, and the seam in $\hat \C$ is a simple analytic loop $\gamma_\phi$.
\emph{Conformal welding} yields the corresponding univalent functions $f$ and $g$
as Riemann mappings of the domains inside and outside of $\gamma_\phi$.
Without loss of generality, we may assume that $\gamma_\phi$ separates $0$ from $\infty$, and that $g(\infty) = \infty$.

Consider the M\"obius invariant function of $\gamma_\phi$ defined as
\begin{equation}
	\lenergy{}(\gamma_\phi) = 
	\frac{1}{\pi} \int_{|z| < 1} \bigg|\frac{f''(z)}{f'(z)}\bigg|^2 \dz
	+ \frac{1}{\pi} \int_{|z| > 1} \bigg|\frac{g''(z)}{g'(z)}\bigg|^2 \dz
	+ 4 \log \bigg| \frac{f'(0)}{g'(\infty)} \bigg|.
	\label{eq:lenergy}
\end{equation}
Generalizing the work of Schiffer~\&~Hawley~\cite{Schiffer-Hawley:Connections_and_conformal_mapping}, in their monograph~\cite{Takhtajan-Teo:Weil-Petersson_metric_on_the_universal_Teichmuller_space} Takhtajan~\&~Teo defined the \emph{universal Liouville action} $\pi \lenergy{}(\gamma_\phi)$ as a function on $\univteichm$, and proved that it is a K\"ahler potential (for the Weil-Petersson metric on a subset of $\univteichm$, originating with~\cite{Cui:Integrably_asymptotic_affine_homeomorphisms_of_the_circle_and_Teichm_spaces}). 
Moreover, following Schiffer's pioneering works~\cite{Schiffer:Fredholm_eigenvalues_of_plane_domains, Schiffer:Fredholm_eigenvalues_of_multiply-connected_domains, Schiffer-Hawley:Connections_and_conformal_mapping}, 
they showed that the formula~\eqref{eq:lenergy} agrees with a Fredholm determinant involving the Grunsky operators associated to $f$ and $g$, and that it can also be written in terms of zeta-regularized determinants of Laplacians; see Remark~\ref{remark:logdet}.
(See also the related~\cite{Segal:Unitary_representations_of_some_infinite-dimensional_groups, Kirillov-Yuriev:Representations_of_the_Virasoro_algebra_by_the_orbit_method, Nag-Sullivan:Teichmuller_theory_and_the_universal_period_mapping_via_quantum_calculus_and_H_half_space_on_the_circle}.)
These concepts have led to many recent developments in geometry and function theory, 
see, e.g.,~\cite{Sharon-Mumford:2D-shape_analysis_using_conformal_mapping, 
Gay-Balmaz-Ratiu:The_geometry_of_the_universal_Teichmuller_space_and_the_Euler-Weil-Petersson_equation, 
RSS:Dirichlet_problem_and_Sokhotski-Plemelj_jump_formula_on_WP_class_quasidisks,
Feiszli-Narayan:Numerical_computation_of_WP_geodesics_in_the_universal_Teichmuller_space, 
Takhtajan-Zograf:Local_index_theorem_for_orbifold_Riemann_surfaces,
Bishop:The_traveling_salesman_theorem_for_Jordan_curves,
Schippers-Staubach:WP_Teich_theory_of_surfaces_of_infinite_conformal_type, 
Bishop:Weil_Petersson_curves_conformal_energies_beta-numbers_and_minimal_surfaces, 
Maibach-Peltola:From_the_conformal_anomaly_to_the_Virasoro_algebra,
Johansson-Viklund:Coulomb_gas_and_the_Grunsky_operator_on_Jordan_domain_with_corners}, 
and references therein.

The zeta-regularized determinants of Laplacians also realize partition functions of Schramm--Loewner evolution (SLE) curves and the Gaussian free field~\cite{LeJan:Markov_loops_determinants_and_Gaussian_fields, Dubedat:SLE_and_free_field}.
In this context, Wang~\cite{Wang:Equivalent_descriptions_of_Loewner_energy} popularized the term \emph{loop Loewner energy} for $\lenergy{}(\gamma_\phi)$ by proving that a limit of chordal Loewner energy defined in~\cite{Rohde-Wang:Loewner_energy_of_loops_and_regularity_of_driving_functions} agrees with the formula~\eqref{eq:lenergy}. 
This renewed interest in probability theory continues to inspire the discovery of relations to various other areas of mathematics, see, e.g.,~\cite{
	Johansson:Strong_Szego_theorem_on_a_Jordan_curve,
	Wiegmann-Zabrodin:Dyson_gas_on_a_curved_contour,
	Brock-Vargas-Pallete:W-volume_for_planar_domains_with_circular_boundary,
	ABKM:Pole_dynamics_and_an_integral_of_motion_for_multiple_SLE0,
	Carfagnini-Wang:OM_functional_for_SLE_loop_measures,
	Michelat-Wang:Loewner_energy_via_the_renormalised_energy_of_moving_frames,
	Peltola-Wang:LDP_of_multichordal_SLE_real_rational_functions_and_determinants_of_Laplacians,
	Viklund-Wang:Loewner_Kufarev_energy_and_foliations_of_WP_quasicircles,
	BBVW:Universal_Liouville_action_as_renormalized_volume_and_its_gradient_flow,
	Luo-Maibach:Two-loop_Loewner_potentials,
	VPWW:Epstein_curves_and_holography-of_the_Schwarzian_action,
	BJRW:Piecewise_geodesic_Jordan_curves_II}, and references therein.

An $\Rp$-extension $\RMFE$ of the Segal moduli spaces (Definition~\ref{def:modular_functor}) may be studied through a \emph{trivialization}~$Z$, which comprises trivializations of the bundles $\RMFE(\moduligb) \to \moduligb$.
A~fixed trivialization defines a number of \emph{cocycles}, which measure to what extent the sewing isomorphisms of $\RMFE$ preserve the trivialization.
For instance, the cocycle
\begin{equation}
	\Omega^Z_{j,k} \colon \moduli{\genus_1}{\boundaries_1} \times \moduli{\genus_2}{\boundaries_2} \to \R 
	\label{eq:Ecocycle_pair_intro}
\end{equation}
of sewing two surfaces is defined by the relation
\begin{equation}
	Z(\Sigma_1) \sewx{\RMFE}{j}{k} Z(\Sigma_2) = e^{\Omega^Z_{j, k}(\Sigma_1, \Sigma_2)} Z(\Sigma_1 \sew{j}{k} \Sigma_2), \qquad \Sigma_1 \in \moduli{\genus_1}{\boundaries_1}, \Sigma_2 \in \moduli{\genus_2}{\boundaries_2}.
\end{equation}
Importantly, the properties of $\RMFE$ are encoded by \emph{cocycle identities}.
We refer to Section~\ref{section:cocycles} for an exhaustive list of cocycles and cocycle identities.

In Section~\ref{section:disk_disk}, we prove that up to a constant, the disk-disk cocycle of $(\disk \diffActing{1} \phi) \sew{1}{1} \disk$ is given by the loop Loewner energy~\eqref{eq:lenergy} times the \emph{central charge} of $\RMFE$ (Definition~\ref{def:central_charge}).

\begin{theorem}[Universality of Loewner energy/universal Liouville action]
	Let $\RMFE$ be a local Fr\"olicher smooth $\Rp$-extension (Definition~\ref{def:modular_functor}) with reparametrization invariant trivialization $Z$ (Proposition~\ref{prop:local_equiv}) and central charge $\charge \in \R$.
	Up to a constant depending on $\RMFE$ and $Z$, but not on $\phi$, we have
	\begin{equation}
		\Omega^Z_{1,1}(\disk \diffActing{1} \phi, \disk) = - \frac{\charge}{24} \lenergy{}(\gamma_\phi) + \const, \qquad \phi \in \Diffpan,
	\end{equation}
	where $\gamma_\phi$ is the simple analytic loop associated to $\phi$ by conformal welding.
	\label{thm:lpot}
\end{theorem}
Our result motivates a more general geometric interpretation of cocycles, for example, as multi-curve/multi-loop Loewner energies~\cite{Peltola-Wang:LDP_of_multichordal_SLE_real_rational_functions_and_determinants_of_Laplacians, Luo-Maibach:Two-loop_Loewner_potentials}, and possible applications towards K\"ahler geometry~\cite{Takhtajan-Teo:Weil-Petersson_metric_on_the_universal_Teichmuller_space, Alekseev-Meinrenken:Symplectic_geometry_of_Teichmuller_spaces_for_surfaces_with_ideal_boundary, AST:Berezin_quantization_conformal_welding_and_the_Bott-Virasoro_group} 
(see also~\cite{Schippers-Staubach:WP_Teich_theory_of_surfaces_of_infinite_conformal_type} and references therein).

\subsection{Applications to SLE loop measures}
\label{section:sle}

Following Schramm's pioneering work~\cite{Schramm:Scaling_limits_of_LERW_and_UST}, there has been significant interest in the construction and applications of random fractal curves on Riemann surfaces with conformally invariant measures (see, e.g., the recent works~\cite{
AJKS:Random_conformal_weldings,
Dubedat:SLE_and_Virasoro_representations_localization, 
AHS:SLE_loop_via_conformal_welding_of_quantum_disks,
Binder-Kojar:Inverse_of_the_Gaussian_multiplicative_chaos,
ACSW:SLE_loop_measure_and_Liouville_quantum_gravity,
Baverez-Jego:Conformal_welding_and_the_matter-Liouville-ghost_factorisation,
Fan-Sung:Quasi-invariance_for_SLE_welding_measures,
Cai-Gao:Uniqueness_of_generalized_conformal_restriction_measures_and_MKS_measures_for_c_pos,
KMAS:Conformal_welding_of_independent_Gaussian_multiplicative_chaos_measures}, and the references below).

In the upper half-plane, chordal SLE$(\kappa)$ is a random curve defined by the Loewner differential equation using conformal maps and one-dimensional Brownian motion with diffusivity ${\kappa > 0}$ as the driving function~\cite{Lawler:Conformally_invariant_processes_in_the_plane, Kemppainen:SLE_book}. 
On the one hand, the parameter $\kappa$ controls the fractal structure (Hausdorff dimension) of the curve \cite{Rohde-Schramm:Basic_properties_of_SLE,Beffara:Dimension_of_SLE}, and on the other hand, motivated by the great success of using SLEs to describe critical interfaces in statistical physics (see~\cite{Schramm:ICM,Smirnov:Towards_conformal_invariance_of_2D_lattice_models,Chelkak:ICM_2018} and references therein), $\kappa$ was related to the central charge $\charge = \charge(\kappa)$ of CFT as~\cite{Bauer-Bernard:Conformal_field_theories_of_SLEs, Cardy:SLE_and_Dyson_circular_ensembles, Friedrich-Werner:Conformal_restriction_highest_weight_representations_and_SLE, Friedrich-Kalkkinen:On_CFT_and_SLE}
\begin{equation}
	\charge = \frac{(6 - \kappa)(3 \kappa - 8)}{2\kappa}.
\end{equation}
Moreover, chordal SLE$(\kappa)$ measures are fundamentally Markovian in the sense that they behave well under restriction to subsets of the upper half-plane~\cite{LSW:Conformal_restriction_the_chordal_case, Werner:Conformal_restriction_and_related_questions}.
There are two special values, $\kappa \in \{8/3,6\}$, for which the central charge vanishes. These are deeply related to the restriction \emph{invariance} property,
which enabled Lawler, Schramm~\&~Werner to present strong support to the conjecture that chordal SLE$(8/3)$ is the scaling limit of the infinite self-avoiding walk in the upper half-plane~\cite{LSW:On_the_scaling_limit_of_planar_SAW},
and relate these to outer boundaries of planar Brownian motion (and of SLE$(6)$), which also paved the way to their proof of the Mandelbrot conjecture~\cite{Mandelbrot:The_fractal_geometry_of_Nature,LSW:The_dimension_of_the_planar_Brownian_frontier_is_four_thirds}.
The loop version of the restriction invariance appeared in~\cite{Werner:The_conformally_invariant_measure_on_self-avoiding_loops, Bauer:Simple_construction_of_Werner_measure_from_chordal_SLE, Chavez-Pickrell:Werners_measure_on_self-avoiding_loops_and_welding}.
In accordance with the observed universal symmetry breaking for critical models and CFTs~\cite{Polyakov:Conformal_symmetry_of_critical_fluctuations, BPZ:Infinite_conformal_symmetry_of_critical_fluctuations_in_2D} (see also Section~\ref{section:cft}), Kontsevich~\&~Suhov~\cite{Kontsevich:CFT_SLE_and_phase_boundaries, Kontsevich-Suhov:On_Malliavin_measures_SLE_and_CFT} introduced a notion of restriction \emph{covariance} for $\charge \in (-\infty, 1]$, which is equivalently expressed in terms of Brownian loop measure (BLM)~\cite{Lawler-Werner:The_Brownian_loop_soup}.
However, what remained mathematically unclear is why this notion of restriction covariance is so universal. 

In the planar setup, conventional in probability theory, \emph{a restriction covariant loop measure}, also known as Malliavin--Kontsevich--Suhov (MKS) loop measure~\cite{Malliavin:Canonic_diffusion_above_the_diffeomorphism_group_of_the_circle, Airault-Malliavin:Unitarizing_probability_measures_for_representations_of_Virasoro_algebra, Kontsevich:CFT_SLE_and_phase_boundaries, Kontsevich-Suhov:On_Malliavin_measures_SLE_and_CFT, Baverez-Jego:The_CFT_of_SLE_loop_measures_and_the_Kontsevich-Suhov_conjecture,GQW:Infinitesimal_conformal_restriction_and_unitarizing_measures_for_Virasoro_algebra}
is a $\sigma$-finite Borel measure $\nu_\C$ on the space $\cloops{\C}$ of Jordan curves with the topology induced by the Hausdorff distance (see~\cite[Proposition~B.2 in Appendix~B]{Baverez-Jego:The_CFT_of_SLE_loop_measures_and_the_Kontsevich-Suhov_conjecture}).
Its restriction to simple loops in a domain $D \subset \C$ is defined by the Radon--Nikodym derivative
\begin{equation}
	\frac{\dd \nu_{D}}{\dd \nu_{\C}}(\loopzero) = e^{\frac{\charge}{2}\Lambda^*(\loopzero, \partial D)} \boldone_{\loopzero \subset D}, \qquad \loopzero \in \cloops{\C},
	\label{eq:restriction_blm}
\end{equation}
where $\Lambda^*(a, \partial D)$ is the (re-)normalized BLM~\cite{Field-Lawler:Reversed_radial_SLE_and_the_Brownian_loop_measure} of almost surely self-intersecting (Brownian) loops hitting both the loop $\loopzero$ and exiting the domain $D$.
This is a covariance property due to the reweighting by Brownian loops, and involves the central charge ${\charge \in (-\infty, 1]}$.
Moreover, one requires $F^* \nu_{D_2} = \nu_{D_1}$ under a biholomorphism $F \colon D_1 \to D_2$. 

In this specific setup, it was conjectured by Kontsevich~\&~Suhov~\cite{Kontsevich:CFT_SLE_and_phase_boundaries, Kontsevich-Suhov:On_Malliavin_measures_SLE_and_CFT}, 
and proven first for $\charge = 0$ by Werner~\cite{Werner:The_conformally_invariant_measure_on_self-avoiding_loops}, and subsequently for any $\charge \in (-\infty, 1]$ by Baverez~\&~Jego~\cite{Baverez-Jego:The_CFT_of_SLE_loop_measures_and_the_Kontsevich-Suhov_conjecture}, that such a measure is \emph{unique} up to normalization by a positive scalar. 
These results are motivated by the universality of SLE as the euclidean path-distribution of interfaces in conformal field theories~\cite{Schramm:Scaling_limits_of_LERW_and_UST, Friedrich-Kalkkinen:On_CFT_and_SLE, Bauer-Bernard:2d_growth_processes_SLE_and_Loewner_chains}.
Loop versions of SLE have been constructed by Werner~\cite{Werner:The_conformally_invariant_measure_on_self-avoiding_loops} for $\charge = 0$, Benoist~\&~Dub\'edat~\cite{Benoist-Dubedat:SLE2_loop_measure} for $\charge = -2$, Kemppainen~\&~Werner~\cite{Kemppainen-Werner:Nested_simple_CLEs_in_Riemann_sphere} for $\charge \in (0, 1]$, and finally by Zhan~\cite{Zhan:SLE_loop_measures} for all $\charge \in (-\infty, 1]$.
The SLE loop measures are restriction covariant measures in the sense of Equation~\eqref{eq:restriction_blm}, and by the Kontsevich--Suhov conjecture, they are unique as such.

In the present work, we answer a question posed by S.~Benoist~\cite{Benoist:Classifying_conformally_invariant_loop_measures} that conceptually precedes the construction and uniqueness of such loop measures, asking which ``restriction functions'' can consistently appear in Equation~\eqref{eq:restriction_blm}, replacing the normalized measure of Brownian loops in the plane, $\Lambda^*(a, \partial D)$, by a general function $f(a, D, \C)$ satisfying conformal invariance and a cocycle property (Items~\ref{def:restriction_function:b}~\&~\ref{def:restriction_function:c} in Definition~\ref{def:restriction_function}).
The generalized setup, moreover, extends very naturally to the setting of loops in Riemann surfaces as considered in~\cite{Kontsevich:CFT_SLE_and_phase_boundaries, Kontsevich-Suhov:On_Malliavin_measures_SLE_and_CFT, Werner:The_conformally_invariant_measure_on_self-avoiding_loops, Zhan:SLE_loop_measures} (see also~\cite{LeJan:Markov_paths_loops_and_fields, LeJan:Brownian_loops_topology} for related ideas).
Our result holds under a further symmetry assumption on the restriction functions (Item~\ref{def:restriction_function:d} in Definition~\ref{def:restriction_function}), which is very natural for both BLM and for SLE, and essential to relate the restriction functions to $\Rp$-extensions of Segal moduli spaces.

\begin{theorem}[Universality of loop measure restriction functions]
	If $f_1$ and $f_2$ are restriction functions (Definition~\ref{def:restriction_function}) with equal central charge which are Fr\"olicher smooth on analytic configurations, then $f_1$ and $f_2$ are equivalent.
	\label{thm:restriction_function}
\end{theorem}

We prove Theorem~\ref{thm:restriction_function} at the end of Section~\ref{section:loop_measures}, which also specifies the precise setting for Benoist's conjecture, as well as the notions of Fr\"olicher smoothness of analytic configurations and central charge.
The proof relies on Theorem~\ref{thm:universality}.

For an immediate application towards universality of SLE loop measures, 
in genus zero\footnote{The general case is immediate, but depends on the conjectured renormalization of Brownian loop measure in higher genus (see Remark~\ref{remark:BLM_conjecture}).}, 
one may combine the uniqueness result of Baverez~\&~Jego~\cite[Theorem~1.3]{Baverez-Jego:The_CFT_of_SLE_loop_measures_and_the_Kontsevich-Suhov_conjecture} with Theorem~\ref{thm:restriction_function}, to obtain:
\begin{corollary}
	Let $f$ be a restriction function and $\nu_D$ a $f$-restriction covariant loop measure in domains $D \subset \C$ with central charge $\charge \in (-\infty, 1]$ (Definitions~\ref{def:restriction_function} and~\ref{def:f_loop_meas}).
	Then, there exists a conformally invariant and symmetric family of continuous functions $g(\blank, D) \colon \cloops{D} \to \R$, indexed by domains $D \subset \C$, such that $\nu_D$ is absolutely continuous with respect to the SLE$(\kappa)$ loop measure in $D$ with Radon--Nikodym derivative $e^{g(\loopzero, D)}$.
\end{corollary}

\begin{remark} \label{remark:logdet}
	For smooth geometry, the zeta-regularized determinant $\detz{\Delta_g}$ of the Laplacian on a Riemann surface $\Sigma$ with conformal metric $g$~\cite{Ray-Singer:R-Torsion_and_the_Laplacian_on_Riemannian_manifolds} provides a simple way to link all the quantities under consideration. 
	On its own, $\detz{\Delta_g}$ may be understood as the partition function of the scalar free boson.
	It satisfies Weyl covariance by the Polyakov--Alvarez anomaly formula~\cite{Polyakov:Quantum_geometry_of_bosonic_strings, Alvarez:Theory_of_strings_with_boundaries_fluctuations_topology_and_quantum_geometry, OPS:Extremals_of_determinants_of_Laplacians}. It thus defines a trivialization\footnote{
	The quantity $\detz{\Delta_g}$ satisfies the exact Weyl covariance by the conformal anomaly~\eqref{eq:conformal_anomaly} when using a specific normalization~\cite[Appendix~A.2]{Luo-Maibach:Two-loop_Loewner_potentials} modified by the boundary term $B(g, \Sigma)$.
}
	\begin{equation}
		Z(\Sigma) = \bigg(\frac{\detz{\Delta_g}}{B(g, \Sigma)}\bigg)^{-\frac{\charge}{2}} [g] \in \Detrc(\Sigma), \qquad B(g, \Sigma) =  e^{\frac{1}{4\pi} \int_{\partial \Sigma} k_g \dboundary{g}} 
		\label{eq:triv_detz}
	\end{equation}
 of the real determinant line bundle.
 	With respect to this trivialization, the cocycle~\eqref{eq:Ecocycle_pair_intro} becomes
	\begin{equation}
	\begin{aligned}
			\Omega^Z_{j,k}(\Sigma_1, \Sigma_2) = - \frac{\charge}{2} \Bigg(
			& \log \frac{\detz{\Delta_{\restrict{g}{\Sigma_1 \sew{j}{k} \Sigma_2}}}}{\detz{\Delta_{\restrict{g}{\Sigma_1}}} \; \detz{\Delta_{\restrict{g}{\Sigma_2}}}}
\\
			& + B(g, \Sigma_1) + B(g, \Sigma_2) - B(g, \Sigma_1 \sew{j}{k} \Sigma_2)
		\Bigg).
		\label{eq:cocycle_detz}
	\end{aligned}
	\end{equation}
	Moreover, the universal Liouville action~\eqref{eq:lenergy} for a smooth simple loop $\gamma$ may be expressed as the following ratio of determinants~\cite[Corollary~3.12]{Takhtajan-Teo:Weil-Petersson_metric_on_the_universal_Teichmuller_space} (cf.~\cite{Schiffer:Fredholm_eigenvalues_of_plane_domains, Schiffer:Fredholm_eigenvalues_of_multiply-connected_domains}):  
	\begin{equation}
		\lenergy{}(\gamma) = 12 \log \Bigg(
			\frac{
				\detz{\Delta_{\restrict{g}{|z|<1}}} \;
				\detz{\Delta_{\restrict{g}{|z|>1}}}
			}{
				\detz{\Delta_{\restrict{g}{L}}} \;
				\detz{\Delta_{\restrict{g}{R}}}
			}
		\Bigg)
		= - \frac{24}{\charge} \Big( \Omega^Z_{j,k}(L, R) - \Omega^Z_{j,k}(\disk, \disk) \Big)  ,
		\label{eq:lenergy_detz}
	\end{equation}
	where $\hat{\C} \setminus \gamma = L \sqcup R$,
	and the $\charge \neq 0$ expression in terms of cocycles~\eqref{eq:cocycle_detz} is in agreement with Theorem~\ref{thm:lpot}.
	Finally, the (normalized) Brownian loop measure may be approximated by zeta-regularized determinants of the Laplacian~\cite{LeJan:Markov_loops_determinants_and_Gaussian_fields, Dubedat:SLE_and_free_field, APPS:Brownian_loops_and_the_central_charge_of_Liouville_random_surface, Peltola-Wang:LDP_of_multichordal_SLE_real_rational_functions_and_determinants_of_Laplacians, Luo-Maibach:Two-loop_Loewner_potentials, Xue-Wang:The_Brownian_loop_measure_on_Riemann_surfaces_and_applications_to_length_spectra}, leading to the restriction function
	\begin{equation}
		f(\loopzero, \Sigma_1, \Sigma_2) = \frac{\charge}{2} \log \Bigg(
			\frac{
				\detz{\Delta_{\restrict{g}{\Sigma_2}}} \;
				\detz{\Delta_{\restrict{g}{\Sigma_1 \setminus \loopzero}}}
			}{
				\detz{\Delta_{\restrict{g}{\Sigma_2 \setminus \loopzero}}} \;
				\detz{\Delta_{\restrict{g}{\Sigma_1}}}
			}
		\Bigg)
		= h(\loopzero, \Sigma_2) - h(\loopzero, \Sigma_1),
		\label{eq:f_detz}
	\end{equation}
	only defined on simple smooth loops $\loopzero \subset \Sigma_1 \subset \Sigma_2$, related to the cocycles~\eqref{eq:Ecocycle_pair_intro} through the auxiliary function
	\begin{equation}
		\gonlyan(\loopzero, \Sigma)
		= \begin{cases}
			\Omega_{1, 1}^Z(L, R), & \textnormal{$\loopzero$ a separating loop,} \\
			\Omega_{1, 2}^Z(\Sigma \setminus \loopzero), & \textnormal{$\loopzero$ a non-separating loop,}
		\end{cases}
	\end{equation}
	where $L$ and $R$ denote the connected components of $\Sigma \setminus \loopzero$ in the separating case.  
	This is a special case of a correspondence between restriction functions and $\Rp$-extensions given by Proposition~\ref{prop:an_coboundary} and Equation~\eqref{eq:cocycle_gonlyan}.
	Note that the formulas~(\ref{eq:triv_detz}, \ref{eq:cocycle_detz}, \ref{eq:lenergy_detz}, \ref{eq:f_detz}) above only depend on the conformal class of the metric $g$, yet this does not make them trivial, since $g$ must be smooth across any seams.
\end{remark}

\subsection{Applications to conformal field theory}
\label{section:cft}

In the context of euclidean two-dimensional quantum field theories, conformal field theories (CFTs) are distinguished by the conformal anomaly. 
A classical field theory is conformally invariant if the action functional is invariant under Weyl transformations: local scale changes replacing a Riemannian metric $g$ by $e^{2\sigma} g$, where $\sigma$ is a real-valued function.
The classical physics literature~\cite{BPZ:Infinite_conformal_symmetry_in_2D_QFT, DMS:CFT} applies general principles of quantization to show that the conformal symmetry is broken in quantum theories in various ways, collectively known as the conformal anomaly.
In mathematics, the conformal anomaly is either assumed as part of axiomatizations of CFTs (e.g.~\cite{Segal:Definition_of_CFT, Gawedzki:CFT_lectures, Segal:Definition_of_CFT_collection}), or proven as a property in constructions of particular CFTs (like Liouville theory; see the review~\cite{GKR:Review_on_the_probabilistic_construction_and_conformal_bootstrap_in_Liouville_theory}).
In this section, we briefly introduce the main instances of the conformal anomaly, and explain how the symmetry breaking relates to our classification of $\Rp$\nobreakdash-extensions of Segal moduli spaces (Theorem~\ref{thm:universality}).

An elementary quantity of CFT that exhibits a conformal anomaly is the partition function, which is a real number $Z_g(\Sigma) \in \R$ depending on a surface $\Sigma$ with Riemannian metric $g$.
It has the \emph{Weyl covariance} property $Z_{e^{2\sigma}}(\Sigma) = e^{\charge \lfunct(\sigma, g, \Sigma)} Z_g(\Sigma)$ with a central charge $\charge \in \R$. Here, the \emph{conformal anomaly} is the functional
\begin{equation}
    \lfunct(\sigma, g, \Sigma) =
    \frac{1}{12 \pi} \iint_\Sigma \bigg(
    \frac{1}{2} |\nabla_g \sigma|_g^2 + R_g \sigma
    \bigg) \dvol{g}
    + \frac{1}{12 \pi} \int_{\partial \Sigma} k_g \sigma \, \dboundary{g},
	\label{eq:conformal_anomaly}
\end{equation}
where $\nabla_g$, $R_g$, and $k_g$ are, respectively, the divergence, the Gaussian curvature, and the boundary curvature with respect to the metric $g$.
From a mathematical point of view, one naturally asks the universality question: \textit{Why does this particular functional $\lfunct(\sigma, g, \Sigma)$ appear in the Weyl covariance?}
In a bigger picture, this question may be motivated by universal properties of CFTs as fixed points of the renormalization group flow, and as scaling limits of two-dimensional critical lattice models (see~\cite{Cardy:CFT_and_statistical_mechanics, Mussardo:Statistical_field_theory}). 
Below, we will make the universality of the conformal anomaly of the form~\eqref{eq:conformal_anomaly} concrete. 

A partial answer can be found by studying infinitesimal conformal symmetry in terms of the representation theory of the Witt algebra --- an infinite-dimensional complex Lie algebra that may be constructed as the Lie algebra $\VectC$ of complex-valued vector fields on the unit circle (Section~\ref{section:setup}).
The anomaly manifests through the possibility of projective representations.
These are equivalent to non-projective representations of a central extension of the Witt algebra by $\C$.
In turn, central extensions of Lie algebras are classified by the second Lie algebra cohomology group, which in the case of the Witt algebra is one-dimensional $H^2(\Witt, \C) \cong \C$; see~\cite{Gelfand-Fuchs:Cohomologies_of_the_Lie_algebra_of_vector_fields_on_the_circle}.
Since a full CFT involves two representations of the Witt algebra, the conformal anomaly is unique up to two complex coefficients\footnote{
	These are the chiral and anti-chiral central charges $\charge_+, \charge_- \in \C$ associated to the chiral and anti-chiral representations of the Witt algebra.
	For a full CFT, they are expected to combine to the real overall central charge $\charge = \charge_+ + \charge_- \in \R$, which is the parameter appearing in the present work.
} of the generator of $H^2(\Witt, \C)$: The Gel'fand--Fuks cocycle
\begin{equation}
\gelfandfuks(\ell_n, \ell_m)
=
\frac{\ii}{12} n^3 \delta_{n + m}, \qquad n, m \in \Z, \qquad \ell_n = - z^{n+1} \partial_z,
\label{eq:gelfandfuks}
\end{equation}
also known as the cocycle defining the \emph{Virasoro algebra} as a central extension of the Witt algebra (see also Section~\ref{section:central_charge}).

In the physics literature, Weyl covariance and the Virasoro algebra are related via the trace anomaly of the stress-energy tensor ---  a third instance of the conformal anomaly, which has been studied mathematically in certain CFTs~\cite{Kupiainen-Oikarinen:Stress-energy_in_Liouville_CFT, GKRV:Conformal_bootstrap_in_Liouville_theory, GKRV:Segals_axioms_for_Liouville_theory, Losev:Probabilistic_correlation_functions_of_the_Schwarzian_field_theory}.
In earlier work~\cite{Maibach-Peltola:From_the_conformal_anomaly_to_the_Virasoro_algebra}, we provided a way to derive the Virasoro algebra from the Weyl anomaly through the real determinant line bundle --- bypassing the stress-energy tensor construction.
However, neither the trace anomaly, nor our derivation, give a satisfying answer to the universality question at the level of the Weyl anomaly itself, since there might be other formulas than~\eqref{eq:conformal_anomaly} defining a consistent Weyl-type anomaly, and which also result in the Virasoro algebra at the infinitesimal level.

Now, let us make the universality question above more concrete, by considering a sufficiently regular (e.g.\ Fr\"olicher smooth) function $Q(\sigma, g, \Sigma)$ of conformal metrics $g$ on Riemann surfaces $\Sigma$ and conformal factors $\sigma \in C^\infty(\Sigma, \R)$.
To obtain a consistent notion of covariance using $Q$ instead of $\lfunct$,
\begin{equation}
	W_{e^{2\sigma} g}(\Sigma) = e^{Q(\sigma, g, \Sigma)} W_{g}(\Sigma),
	\label{eq:Q_covariance}
\end{equation}
for functions $W_g(\Sigma)$ of the surface and metric, we require $Q$ to satisfy the following:
\begin{enumerate}
	\item \emph{Diffeomorphism invariance:} $Q(\sigma_1, g_1, \Sigma_1) = Q(\sigma_2, g_2, \Sigma_2)$ if there exists a diffeomorphism $f \colon \Sigma_1 \to \Sigma_2$ such that $g_1 = f^* g_2$, and $\sigma_1 = \sigma_2 \circ f$.
	\item \emph{The cocycle property:} $Q(\sigma_1 + \sigma_2, g, \Sigma) = Q(\sigma_1, g, \Sigma) + Q(\sigma_2, e^{2\sigma_1} g, \Sigma)$.
	\item \emph{Locality:} $Q(\sigma, g, \Sigma_1 \sew{j}{k} \Sigma_2) = Q(\restrict{\sigma}{\Sigma_1}, \restrict{g}{\Sigma_1}, \Sigma_1) + Q(\restrict{\sigma}{\Sigma_2}, \restrict{g}{\Sigma_2}, \Sigma_2)$, where $g$ is a conformal metric on $\Sigma_1 \sew{j}{k} \Sigma_2$, and $\sigma \in C^\infty(\Sigma_1 \sew{j}{k} \Sigma_2, \R)$.
\end{enumerate}
In particular, the conformal anomaly $\lfunct(\sigma, g, \Sigma)$ satisfies the properties 1, 2, and 3, and $\lfunct$-covariance is Weyl covariance.

The idea to encode the Weyl covariance of partition functions into a line bundle goes back to Friedan~\&~Shenker~\cite{Friedan-Shenker:The_analytic_geometry_of_two-dimensional_conformal_field_theory}.
The term \emph{real determinant line bundle} appears in the works of Kontsevich~\&~Suhov~\cite{Kontsevich-Suhov:On_Malliavin_measures_SLE_and_CFT}, Friedrich~\cite{Friedrich:On_connections_of_CFT_and_SLE}, and Benoist~\&~Dub\'edat~\cite{Dubedat:SLE_and_free_field, Dubedat:SLE_and_Virasoro_representations_localization, Benoist-Dubedat:SLE2_loop_measure} in the context of SLE.
Our previous results show that the real determinant line bundle $\Detrc$ is a nontrivial local
$\Rp$-extension~\cite[Corollary~1.5]{Maibach-Peltola:From_the_conformal_anomaly_to_the_Virasoro_algebra}.
By generalizing the construction of $\Detrc$ (Section~\ref{section:det_line_bundle}), we can define a local $\Rp$-extension $\RMFE$ associated to the function $Q(\sigma, g, \Sigma)$.

For $\RMFE$, the fiber over a surface $\Sigma \in \moduligb$ is a set of equivalence classes, denoted $\lambda [g]_Q \in \RMFE(\Sigma)$, comprising pairs of a conformal metric $g$ on a representative of $\Sigma$ and a positive real number $\lambda > 0$.
The equivalence is defined so that if $g_1$ and $g_2 = e^{2\sigma} g_2$ for $\sigma \in C^\infty(\Sigma, \R)$ are two conformally equivalent metrics, then the numbers $\lambda_1, \lambda_2 > 0$ of equivalent pairs are related as $\lambda_2 = e^{-\charge \; Q(\sigma, g, \Sigma)} \lambda_1$.
In particular, we have
\begin{equation}
[e^{2\sigma} g]_Q = e^{-Q(\sigma, g, \Sigma)} [g]_Q \in \RMFE(\Sigma).
\label{eq:factor_Q}
\end{equation}
Extending the equivalence to pairs related by pullback of the metric via isomorphisms of the Riemann surface, we obtain principal $\Rp$-bundles $\RMFE(\moduligb)$ over the Segal moduli spaces.
In~\cite{Maibach-Peltola:From_the_conformal_anomaly_to_the_Virasoro_algebra}, we added sewing isomorphisms to the real determinant line bundle by taking unions of admissible metrics.
Applying the same principle, we can define the sewing isomorphisms of $\RMFE$,
\begin{equation}
\begin{gathered}
\begin{aligned}
	\sewx{\RMFE}{j}{k} \colon &&\RMFE(\Sigma_1) &\otimes \RMFE(\Sigma_2) &&\longrightarrow && \RMFE(\Sigma_1 \sew{j}{k} \Sigma_2), \\
	&&\lambda_1 [g_1]_Q &\otimes \lambda_2 [g_2]_Q &&\longmapsto &&\lambda_1 \lambda_2 [g_1 \cup g_2]_Q,
\end{aligned} \\
\begin{aligned}
	\sewxself{\RMFE}{j}{k} \colon &&\RMFE(\Sigma) &\longrightarrow && \RMFE(\sewself{j}{k} \Sigma), \\
	&&\lambda [g]_Q &\longmapsto &&\lambda [g]_Q.
\end{aligned}
\end{gathered}
\end{equation}
To obtain a smooth conformal metric $g_1 \cup g_2$ on $\Sigma_1 \sew{j}{k} \Sigma_2$ and $g$ on $\sewself{j}{k} \Sigma$, one chooses appropriate representatives; e.g.,
one may apply Weyl transformations to $g_1$ and $g_2$ so that they agree with the pushforward of the flat metric via the boundary parametrizations near the respective boundary components.
The factors~\eqref{eq:factor_Q} obtained by these Weyl transformations are what makes the sewing isomorphisms nontrivial in the case of nontrivial~$Q$.

Let us now return to the question concerning the universality of the conformal anomaly. 
Theorem~\ref{thm:universality} combined with Theorem~\ref{thm:detrc_prop} in Section~\ref{section:det_line_bundle} imply that $\RMFE$ must be isomorphic to the real determinant line bundle $\Detrc$ for some central charge $\charge \in \R$.
Thus, the conformal anomaly~\eqref{eq:conformal_anomaly} is \emph{universal} in the sense that all local $\Rp$-extensions are isomorphic to a real determinant line bundles defined by $\charge \lfunct(\sigma, g, \Sigma)$, for $\charge \in \R$, including those defined by other local covariance functions $Q(\sigma, g, \Sigma)$.

\subsection*{Acknowledgements}

We would like to thank Yi-Zhi Huang, David Radnell, and Eric Schippers for insightful discussions on the Segal moduli space, and Eric Schippers also for useful comments on a draft version of this manuscript. 
Part of this work was included in S.M's doctoral thesis~\cite[Appendix~C]{Maibach:PhD_thesis}.
The first version of this manuscript was finished during the stay of the authors at the Centre International de Rencontres Math\'ematiques (CIRM), which we cordially thank for hospitality and support. 

S.M.~has been supported by the Deutsche Forschungsgemeinschaft (DFG, German Research Foundation) under Germany's Excellence Strategy EXC-2047/1-390685813, by the grant CRC 1060 ``The
Mathematics of Emergent Effects'' (Project-ID 211504053), and by the grant from the Simons Foundation International [SFI-MPS-PP-00012621-19].

This material is part of a project that has received funding from the  European Research Council (ERC) under the European Union's Horizon 2020 research and innovation programme (101042460): 
ERC Starting grant ``Interplay of structures in conformal and universal random geometry'' (ISCoURaGe) 
and from the Academy of Finland grant number 340461 ``Conformal invariance in planar random geometry.''

E.P.~is also supported by 
the Academy of Finland Centre of Excellence Programme grant number 346315 ``Finnish centre of excellence in Randomness and STructures (FiRST)'' 
and by the Deutsche Forschungsgemeinschaft (DFG, German Research Foundation) project number 390534769 ``Matter and Light for Quantum Computing (ML4Q).

\subsection*{Organization of this article}

Section~\ref{section:setup} gathers preliminaries concerning operations on Segal moduli spaces and its important sub-moduli spaces of hyperbolic and projective surfaces. Specifically, we review complex deformations of the unit circle and their action on the Segal moduli spaces, and we introduce a class of ```Fuchsian projective surfaces'' for which the boundary parametrizations are charts in the Fuchsian projective structure on the (closed) doubled surface.

In Section~\ref{section:modular_functor}, we introduce the precise definition of $\Rp$-extension of the Segal moduli spaces, and construct the central extension that such $\Rp$-extension induces on complex deformations. 
Moreover, we discuss cocycles and central charge, as well as fundamental properties of $\Rp$-extensions: locality, flat modular invariance, crossing invariance, and hyperbolic modular invariance.
We also show that the real determinant line bundle satisfies all of these properties. 

In Section~\ref{section:disk_disk}, we focus on cocycles obtained from sewing of a pair of disks.
We show the relation to the universal Liouville action, or loop Loewner energy, thus proving Theorem~\ref{thm:lpot} (whose proof comprises the totality of Section~\ref{section:disk_disk}).
Important steps in the proof are the identification of critical points for the disk-disk cocycle and the derivation of a variational formula that involves the Schwarzian derivative of a welding map, which then enables us to identify the cocycle and the associated central charge.

Section~\ref{section:iso} comprises the proof of Theorem~\ref{thm:universality}, proceeding by induction starting from the disk-disk case, and by proving the properties introduced in the earlier sections (flat modular invariance, crossing invariance, and hyperbolic modular invariance).

In the final Section~\ref{section:loop_measures}, we turn to the restriction covariant SLE/MKS loop measures, and to the proof of Theorem~\ref{thm:restriction_function}. We first detail the setup of restriction functions and compare with the precise setup of S.~Benoist. Finally, we construct local $\Rp$-extensions from restriction functions, whose universality yields Theorem~\ref{thm:restriction_function}.

\section{The Segal moduli spaces}
\label{section:setup}
We now fix notation and gather preliminary results for the subsequent computations involving Riemann surfaces and the sewing operation.
More generally, this framework applies to the Segal approach to CFT~\cite{Segal:Definition_of_CFT, Segal:Definition_of_CFT_collection}.
The Segal moduli spaces $\moduligb$ possess infinite-dimensional geometric structure, as well as algebraic structure given by the sewing operation~\eqref{eq:sewing_moduligb}.
The notion of ``complex deformations'' of the circle, $\DefC$, is a way of complexifying the group $\Diffpan$ of orientation-preserving analytic diffeomorphisms of the circle.
We summarize the results from our companion work \cite{Maibach-Peltola:Complex_deformations_of_the_circle} relating the geometry and algebra of $\moduligb$ through $\DefC$ (Section~\ref{section:previous}).
Then, we identify a subspace of $\moduligb$ consisting of surfaces whose boundary parametrizations are charts in a projective structure, and another subspace where this projective structure is the Fuchsian projective structure on the doubled surface (Section~\ref{section:projective}), and gather various useful properties. 

\subsection{Complex deformations and their Fr\"olicher structure}
\label{section:previous}

In this section, we recall basic notions from our companion work~\cite{Maibach-Peltola:Complex_deformations_of_the_circle} concerning complex deformations of the unit circle, and gather useful properties (Lemma~\ref{lemma:complex_deformations}).

\begin{definition}
A \emph{complex deformation} of the unit circle is an injective real-analytic map $\phi \colon S^1 \to \C \setminus \{0\}$ such that the winding number of $\phi$ around $0$ equals $+1$.
We denote the set of complex deformations by $\DefC$.
It is an infinite-dimensional manifold\footnote{
	$\DefC$ is an open submanifold of $\mathcal{O}(S^1)$, the convenient vector space of complex-valued real-analytic functions on $S^1$, endowed with the inductive limit topology over restrictions of spaces of holomorphic functions on neighborhoods of $S^1$ equipped with the topology of uniform convergence on compact sets~\cite{Kriegl-Michor:Convenient_setting_of_global_analysis}.
} with partially defined inversion and composition operations, characterized as follows:
\label{def:defc}
\begin{enumerate}
	\item
		A complex deformation $\phi \in \DefC$ is \emph{invertible} if $\phi^{-1} \colon \phi(S^1) \to \C$ extends univalently and without zeroes to an open neighborhood of $\DefCnbhd{\phi}$: the closed set of points bounded by the inner and outer boundaries of $\phi(S^1) \cup S^1$ as viewed from~$0$ and~$\infty$.
	We denote the subset of invertible complex deformations by ${\definvset \subset \DefC}$.
	\item
		A pair $\phi, \psi \in \DefC$ of complex deformations is \emph{composable} if $\phi$ extends univalently and without zeroes to an open neighborhood of $\DefCnbhd{\psi}$.
		We denote the subset of composable pairs of complex deformations by ${\defmultset \subset \DefC \times \DefC}$.
\end{enumerate}
\end{definition}

\begin{remark}
The subsets $\definvset \subset \DefC$ and $\defmultset \subset \DefC \times \DefC$ are not submanifolds.
\glsreset{fsmooth}
Instead, we can equip them with the Fr\"olicher structures\footnote{
	A Fr\"olicher structure on a set $X$ consists of a set $\curves{X}$ of curves $\R \to X$ and a set $\functions{X}$ of functions $X \to \R$ which are considered smooth.
	See~\cite{Kriegl-Michor:Convenient_setting_of_global_analysis} and \cite[Appendix~A]{Maibach-Peltola:Complex_deformations_of_the_circle}.
} induced by the inclusions, and the respective operations are \emph{\gls{fsmooth}}
with respect to this structure.
See~\cite{Maibach-Peltola:Complex_deformations_of_the_circle} for more details.
\end{remark}

Importantly for this article, in the Fr\"olicher structure a tangent vector at $\phi \in \DefC$ is represented by a smooth function $\gamma \colon S^1 \times \R \to \C \setminus \{0\}$ such that $\gamma(0, \blank) = \phi$, and we have $\gamma(t, \blank) \in \DefC$ for each fixed $t \in \R$.
Such a curve is \gls{fsmooth} if and only if it is the left-trivializing flow of a smoothly time-dependent vector field $v \in \curves{\VectC}$: 
\begin{equation}
	\gamma = \Phi^v(t, \blank) \circ \phi, \qquad \dot{\Phi}^v(t, z) = (\Phi^v)'(t, z) \; v(t, z), \qquad \Phi^v(0, z) = z.
	\label{eq:flow_right}
\end{equation}
The tangent vector denoted by $[\gamma]_\sim \in T_\phi \DefC$ only depends on the vector field at ${t = 0}$.
This identifies the tangent space at the identity $\id \in \DefC$ with the Lie algebra $\VectC$, the space of real-analytic complex-valued vector fields on~$S^1$, also known as the \emph{Witt algebra} with generators $\ell_n = z^{n + 1} \partial_z$, with $n \in \Z$, and relations
\begin{equation}
	[\ell_n, \ell_m] = (n - m) \ell_{n + m}, \qquad n, m \in \Z.
\end{equation}

The notion of complex deformation of the unit circle in Definition~\ref{def:defc} unifies several concepts in geometric function theory.
For instance, the set $\DefC$ of complex deformations contains the following subsets:
\begin{enumerate}
	\item
		A neighborhood of the identity of the group $\mob$ of M\"obius transformations, with Lie algebra $\mobvect$ spanned by $\ell_{-1} = \partial_z$, and $\ell_0 = z \partial_z$, and $\ell_{1} = z^2 \partial_z$.
	\item
		The diffeomorphism group $\Diffpan \subset \DefC$ as a Fr\'echet--Lie subgroup with Lie algebra $\VectR \subset \VectC$.
		In this sense, $\DefC$ serves as a complexification.
	\item
		The set $\univalent$ of univalent functions $\phi \colon \odisk \to \C$ such that $\phi(0) = 0$ and which extend univalently and without zeroes to a neighborhood of the closed unit disk.
\end{enumerate}

We denote a compact connected Riemann surface with enumerated and analytically parametrized boundary components as a tuple $\param{\Sigma, \zeta_1, \dots, \zeta_\boundaries}$, where $\zeta_j \colon S^1 \to \partial_j \Sigma$ are the boundary parametrizations, for $1 \leq j  \leq \boundaries$ and $\partial \Sigma = \bigsqcup_j \partial_j \Sigma$.
The Segal moduli spaces $\moduligb$ consist of equivalence classes denoted $\parameq{\Sigma, \zeta_1, \dots, \zeta_\boundaries}$, which we usually abbreviate just by $\Sigma \in \moduligb$.
For example, we consider the \emph{unit disk} $\disk$ as the equivalence class of the closed disk $\param{\cdisk, \inversion}$ with standard boundary parametrization given by the inversion
\begin{equation}
	\inversion \colon \hat \C \longrightarrow \hat \C, \quad z \longmapsto \frac{1}{z} ,
	\label{eq:inversion}
\end{equation}
and the \emph{standard annulus} $\A_\tau$ of modulus ${\tau > 0}$ as the equivalence class
\begin{equation}
\A_\tau = \param{\setsuchthatinline{z \in \C}{e^{-2\pi \tau} \leq |z| \leq 1},\quad \inversion,\quad \scaling{\tau}},
\end{equation}
with boundary parametrizations given by the inversion $\inversion$ and the scaling transformation
\begin{equation}
	\scaling{\tau} \colon \hat \C \longrightarrow \hat \C, \quad z \longmapsto e^{-2\pi \tau} z.
	\label{eq:scaling}
\end{equation}
To each generic surface $\Sigma \in \moduligb$, we associate the \emph{capped surface} $\Sigma \sewall \underline{\disk}  \in \moduli{\genus}{0}$, which is the closed surface obtained by sewing $\boundaries$ disks to the $\boundaries$ boundary components of $\Sigma$.

Recall that if $\genus$ and $\boundaries$ are such that the Euler characteristic is negative, there is a unique conformal metric $\gcc(\Sigma)$ of constant curvature $-1$ in the conformal class of the Riemann surface $\Sigma$ such that the boundary components are geodesics (termed \emph{uniform metric} of ``type I'' in~\cite{OPS:Extremals_of_determinants_of_Laplacians}). 
If each boundary parametrization has constant speed $|\partial_\theta \zeta_j(e^{\ii \theta})|_\gcc$ in the uniform constant-curvature metric 
$\gcc(\Sigma)$, we call $\param{\Sigma, \zeta_1, \dots, \zeta_{\boundaries}}$ \emph{hyperbolic}\footnote{
	Note that $\gcc(\Sigma)$ is not the complete hyperbolic metric unless $\boundaries=0$, that is, $\Sigma$ is a closed surface.
}. 
We denote the subspace of hyperbolic surfaces by $\hmoduligb \subset \moduligb$.

A complex deformation $\phi \in \DefC$ \emph{acts by deformation} of the $j$th boundary component if $\zeta_j \colon S^1 \to \Sigma \sew{j}{1} \disk$ as a function into the capped surface extends univalently to an open neighborhood of $\DefCnbhd{\phi}$ (see Definition~\ref{def:defc}) in $\C \setminus \{0\}$.
Then, the deformed surface $\Sigma \diffActing{j} \phi$ is obtained from composing $\zeta_j$ with $\phi$ using the analytic continuation.
Note that this action not only reparametrizes the boundary component, but it may also grow and shrink the surface itself, changing the moduli of the Riemann surface.
(See~\cite[Figure~4.1]{Maibach-Peltola:Complex_deformations_of_the_circle} for an illustration.)

As explained in~\cite[Section~4]{Maibach-Peltola:Complex_deformations_of_the_circle}, these actions, together with the sewing operation~\eqref{eq:sewing_moduligb}, equip $\moduligb$ with an initial Fr\"olicher structure, making the actions of $\DefC$ and the sewing operation \gls{fsmooth}.
Moreover, the Virasoro uniformization theorem~\cite[Theorem~4.9]{Maibach-Peltola:Complex_deformations_of_the_circle} shows that the boundary deformations alone generate a set of curves spanning the tangent spaces of the Segal moduli spaces $\moduligb$.

\noindent
The following lemma collects auxiliary results that are used throughout this work.
\begin{lemma}
	\label{lemma:complex_deformations}
	\leavevmode
	\begin{enumerate}
		\item
			Any complex deformation acts on $\disk$
			\label{lemma:complex_deformations:unit_disk}
		\item
			Any disk in $\disks$ may be represented in the form $\disk \diffActing{1} \phi$ for $\phi \in \Diffpan$.
			\label{lemma:complex_deformations:diff_disk}
		\item 
			The conjugation $\phi \mapsto \inversion \circ \phi \circ \inversion$ is a \gls{fsmooth} involution on $\DefC$.
			\label{lemma:complex_deformations:conj}
		\item
			The action of $\inversion \circ \phi \circ \inversion$, for $\phi \in \univalent$, preserves the unit disk, $\disk \diffActing{1} (\inversion \circ \phi \circ \inversion) = \disk$, as an equivalence class in $\disks$.
		\label{lemma:complex_deformations:disk_invariant}
	\end{enumerate}
\end{lemma}

\begin{proof}
The boundary parametrization $\inversion$ defined in~\eqref{eq:inversion} is a  M\"obius transformation, yielding Item~\ref{lemma:complex_deformations:unit_disk}.
	Any composition $\inversion \circ \phi$, for $\phi \in \DefC$, is a well-defined boundary parametrization of the closed disk bounded by $\inversion(\phi(S^1))$ and as seen from $0$. 
	For Item~\ref{lemma:complex_deformations:diff_disk}, given any disk $\param{D, \zeta_1}$, the Riemann mapping $F \colon D \to \odisk$ has real-analytic boundary behavior, so $\param{D, \zeta_1}$ is isomorphic to $\param{\cdisk, F \circ \zeta_1} = \disk \diffActing{1} (\inversion \circ F \circ \zeta_1)$, and $\inversion \circ F \circ \zeta_1 \in \Diffpan$.
	For \gls{fsmooth}ness in Item~\ref{lemma:complex_deformations:conj}, representing a \gls{fsmooth} curve $\gamma \in \curves{\DefC}$ by a time-dependent vector field $v \in \curves{\VectC}$ as in~\eqref{eq:flow_right}, we obtain $\inversion \circ \gamma \circ \inversion = \Phi^{\inversion^*v}(t, \blank) \circ (\inversion \circ \phi \circ \inversion)$, which is indeed a \gls{fsmooth} curve.
	Lastly, if $\phi \in \univalent$, applying $\phi$ as an isomorphism to compute
	$\disk \diffActing{1} (\inversion \circ \zeta \circ \inversion) = \parameq{\blank, \zeta \circ \inversion} = \disk$ yields Item~\ref{lemma:complex_deformations:disk_invariant}.
\end{proof}

The \gls{fsmooth} involution on invertible complex deformations defined by inversion and conjugation as in Item~\ref{lemma:complex_deformations:conj} of Lemma~\ref{lemma:complex_deformations},
\begin{equation}
	\label{eq:def_invinv}
	\invinv \colon \definvset \longrightarrow \definvset, \qquad \phi \longmapsto \inversion \circ \phi^{-1} \circ \inversion,
\end{equation}
appears when moving a complex deformation across the seam of a sewing operation,
\begin{equation}
	(\Sigma_1 \diffActing{j} \phi) \sew{j}{l} \Sigma_2
	= \Sigma_1 \sew{j}{l} (\Sigma_2 \diffActing{l} \invinv(\phi)).
	\label{eq:act_between}
\end{equation}
We investigate this kind of symmetry at the level of $\Rp$-extensions in Section~\ref{section:action_central_extension} --- see Item~\ref{prop:actionsewingE:invinv} of Proposition~\ref{prop:actionsewingE}.

\subsection{Projective and Fuchsian projective surfaces}
\label{section:projective}

A \emph{projective structure} on a smooth surface $\Sigma$ is a maximal atlas $\{\Psi \colon U \to \cdisk\}$, for $U \subset \Sigma$ open, such that the transition functions are M\"obius transformations\footnote{We let the chart take values in the unit disk instead of the upper half-plane, since we parametrize boundary components using $S^1$.}~\cite{Dumas:Complex_projective_structures}.
In the case of a Riemann surface, as the transition functions are already complex-analytic, we require that the atlas of the projective structure also defines the complex structure.
\begin{definition}
	\label{def:projective}
	A compact connected Riemann surface $\param{\Sigma, \zeta_1, \dots, \zeta_\boundaries}$ with enumerated and analytically para\-metrized boundary components 
	is \emph{projective} if there exists a projective structure on $\Sigma$ on which injective analytic continuations of $\zeta_j^{-1} \colon \partial_j \Sigma \to \hat{\C}$, for $1 \leq j \leq \boundaries$, are charts.
	We denote the subspace of projective surfaces by $\moduligbmob \subset \moduligb$.
\end{definition}

Indeed, for any isomorphism $F \colon \param{\Sigma_1, \zeta_1, \dots, \zeta_\boundaries} \to \param{\Sigma_2, \xi_1, \dots, \xi_\boundaries}$ injective analytic continuations of $\zeta_j^{-1} = \xi_j^{-1} \circ F$ are charts in the pullback projective structure, whence the notion of projective surface is well-defined in the Segal moduli spaces.
Note that such a projective structure, if it exists, is determined by any one of the boundary parametrizations.

\begin{proposition}
	\leavevmode
	\begin{enumerate}
	\item
		We have $\modulimob{0}{1} = \{\param{\disk, \inversion}\}$.
	\item
		Let $\Sigma \in \moduligbmob$ with $\boundaries \geq 2$, and $\phi \in \DefC$, and $1 \leq j \leq \boundaries$ be such that $\Sigma \diffActing{j} \phi$ exists. Then, we have $\Sigma \diffActing{j} \phi \in \moduligbmob$ if and only if $\phi$ is a M\"obius transformation.
		\label{prop:projective:mobius}
	\item
		If $\Sigma_1$ and $\Sigma_2$ are projective, then $\Sigma_1 \sew{j}{k} \Sigma_2$ is projective.
		\label{prop:projective:sewing}
	\item
		Let $\param{\Sigma, \zeta_1, \dots, \zeta_{\boundaries}}$ be a compact connected genus $0$ surface with enumerated and analytically parametrized boundary components.
		Consider a uniformizing map $F \colon \Sigma \sewall \underline{\disk} \to \hat \C$ of the capped surface. 
		Then, $\Sigma$~is projective if and only if $F \circ \zeta_1, \ldots, F \circ \zeta_\boundaries$ are M\"obius transformations.
	\end{enumerate}
	\label{prop:projective}
\end{proposition}

\begin{proof}
\leavevmode
\begin{enumerate}
	\item
		The disk carries a unique projective structure (as can be seen by considering the monodromy representation of projective structures~\cite{Dumas:Complex_projective_structures}).
		Moreover, any element of $\disks$ may be represented as $\param{\cdisk, \inversion \circ \phi}$ for $\phi \in \Diffpan$ (by Item~\ref{lemma:complex_deformations:diff_disk} of Lemma~\ref{lemma:complex_deformations}).
		If the latter represents a projective surface, then $\inversion \circ \phi$ is a M\"obius transformation, so $\param{\cdisk, \inversion \circ \phi}$ is isomorphic to $\disk$ via $\invinv(\phi)$ from Equation~\eqref{eq:def_invinv}.
	\item
		Since $\boundaries \geq 2$, there is another boundary component, which uniquely defines the projective structure compatible with that boundary component.
		Hence, the parametrizations of $\partial_j \Sigma$ must be compatible with the same projective structure before and after reparametrization.
		This is the case if and only if $\phi$ is a M\"obius transformation.
	\item
		As $\inversion$ is a M\"obius transformation, injective analytic continuations of the map $\zeta^{-1}_j \circ \inversion \circ \xi_k$, which identify $\partial_j \Sigma_1$ and $\partial_k \Sigma_2$, are transition functions of charts in the respective projective structures.
		Hence, the projective structures on $\Sigma_1$ and $\Sigma_2$ induce a projective structure of $\Sigma_1 \sew{j}{k} \Sigma_2$, compatible with each of the boundary parametrizations.
	\item
		On the one hand, if $F \circ \zeta_j$ are M\"obius transformations, then they indeed are charts in the unique projective structure on $\hat{\C}$ defined by the charts $\id$ (identity) and $\inversion$ (inversion).
		The restriction of that projective structure to $F(\Sigma)$ turns $\Sigma$ into projective surface.
		On the other hand, assuming that $\Sigma$ is projective, Item~\ref{prop:projective:sewing} implies that $\Sigma \sewall \underline{\disk}$ is also projective.
		Since the projective structure on the Riemann sphere, which is compatible with the complex structure is unique, the parametrizations $F \circ \zeta_j$ are charts in the standard projective structure on $\hat{\C}$ and therefore M\"obius transformations. \qedhere
\end{enumerate}
\end{proof}

The \emph{complex conjugate} $\conj{\Sigma}$ of a Riemann surface $\Sigma$ with (unparametrized) boundary is defined as follows~\cite{Abikoff:The_real_analytic_theory_of_Teichmuller_space}.
The surface $\conj{\Sigma}$ is the same smooth manifold, and the complex-analytic chart of $\conj{\Sigma}$ corresponding to a complex-analytic chart $\Psi \colon U \to \cdisk$ of $\Sigma$ is given by the complex conjugate map $\conj{\Psi} \colon U \to \cdisk$.
For two charts $\Phi_1$, $\Phi_2$ with domains $U_1$, $U_2$, the transition function $(\conj{\Psi_2} \circ (\conj{\Psi_1})^{-1})(z) = \conj{(\Psi_2(\Psi_1^{-1}(\conj{z})))}$ for $z \in \Psi_1(U_1 \cap U_2)$ is holomorphic.
Hence, with these charts, $\conj{\Sigma}$ is a Riemann surface.

The \emph{Schottky double} of $\Sigma$ is then defined as $(\Sigma \sqcup \conj{\Sigma})/_\sim$, where the equivalence relation identifies boundary points as follows.
Let $\Psi \colon U \to \cdisk$ be a complex-analytic chart of $\Sigma$, and $z \in \Psi(U) \cap S^1$.
Then, $\conj{\Psi} \colon U \to \cdisk$ is a chart of $\conj{\Sigma}$ and we identify $\Psi^{-1}(z)$ with the corresponding point in $\conj{\Sigma}$, which in the chart $\conj{\Psi}$ becomes $\conj{\Psi}(\Psi^{-1}(z)) = \conj{z}$.
This defines a Riemann surface without boundary: the charts $\Psi \cup (\inversion \circ \conj{\Psi}) \colon (U \sqcup U)_\sim \to \hat{\C}$ are holomorphic at boundary points of $\Sigma$ by Schwarz reflection.

Observe that, because complex deformations are only required to be real-analytic, conjugating them by complex conjugation yields a \gls{fsmooth} operation $\phi \mapsto \conj{\phi} \circ \inversion$, resulting again in a complex deformation (since the operation preserves the orientation). Moreover, we have $\conj{(\conj{\phi} \circ \inversion)} \circ \inversion = \phi$ and $(\conj{\phi} \circ \inversion)^{-1} = \conj{(\phi^{-1})} \circ \inversion$ if the inverses exist.
If $\phi \in \Diffpan$, we also have $\conj{\phi} \circ \inversion = \inversion \circ \phi \circ \inversion$, since $\inversion(\conj{z}) = z$ for $z \in S^1$.

Given a compact connected Riemann surface $\param{\Sigma, \zeta_1, \dots, \zeta_\boundaries}$ with enumerated and analytically parametrized boundary components, we define the \emph{complex conjugate} as the tuple $\param{\conj{\Sigma}, \zeta_1 \circ \inversion, \dots, \zeta_\boundaries \circ \inversion}$.
Since $\zeta_j$ is negatively oriented in $\Sigma$, it is positively oriented in $\conj{\Sigma}$, and thus, the boundary parametrization $\zeta_j \circ \inversion$ is again negatively oriented.
Hence, $\param{\conj{\Sigma}, \zeta_1 \circ \inversion, \dots, \zeta_\boundaries \circ \inversion}$ is a compact connected Riemann surface with enumerated and analytically parametrized boundary components.

The \emph{double} $\Sigma \sewall \conj{\Sigma}$ is obtained by sewing a surface to its complex conjugate\footnote{
	Here, $\sewall$ stands for sewing pairs of boundary components with the same label.
	Note that this involves the self-sewing operation: If we relabel the boundary components after each sewing operation, the double defined by $\Sigma \sewall \conj{\Sigma} =\sewself{1}{2} \cdots \sewself{1}{(\boundaries - 1)} \sewself{1}{\boundaries} (\Sigma \sew{1}{1} \conj{\Sigma})$.
}.
	When sewing a surface to its own complex conjugate, the chart defined by analytically continuing the boundary parametrizations at a seam becomes $\zeta_j^{-1}(z)$ for $z \in \Sigma$, and for $z \in \conj{\Sigma}$ one analytically continues $\inversion \circ \zeta_j^{-1}$ as $\conj{(\zeta_j^{-1})}$ and composes with $\inversion$ due to the sewing operation. This results in the same chart $\zeta_j^{-1} \sqcup (\inversion \circ \conj{(\zeta_j^{-1})})$ as for the Schottky double.
	Therefore the double $\Sigma \sewall \conj{\Sigma}$ agrees with the Schottky double of the surface without boundary parametrization.

\begin{lemma}
	\leavevmode
	\begin{enumerate}
		\item
			We have $\conj{(\Sigma_1 \sew{j}{k} \Sigma_2)} = \conj{\Sigma}_1 \sew{j}{k} \conj{\Sigma}_2$ and $\conj{(\sewself{j}{k} \Sigma)} = \sewself{j}{k} \conj{\Sigma}$.
			\label{lemma:conj_double:sew}
		\item
			We have $\conj{(\Sigma \diffActing{j} \phi)} = \conj{\Sigma} \diffActing{j} (\conj{\phi} \circ \inversion)$ and $\conj{(\Sigma \diffActing{\vartheta} \phi)} = \conj{\Sigma} \diffActing{\vartheta \circ \inversion} (\conj{\phi} \circ \inversion)$.
			\label{lemma:conj_double:act}
		\item
			Complex conjugation defines \gls{fsmooth} involutions of $\moduligb$.
			\label{lemma:conj_double:conj}
		\item
			The doubling construction defines \gls{fsmooth} maps $\moduligb \to \moduli{(2\genus + \boundaries - 1)}{0}$.
			\label{lemma:conj_double:double}
	\end{enumerate}
	Item~\ref{lemma:conj_double:act} involves the unraveling operation defined in~\textnormal{\cite[Section~4.2]{Maibach-Peltola:Complex_deformations_of_the_circle}}.
\end{lemma}
\begin{proof}
	In Item~\ref{lemma:conj_double:sew}, the sewing operation $\conj{\Sigma}_1 \sew{j}{k} \conj{\Sigma}_2$ involves the boundary components of the complex conjugated surfaces.
	These are identified by $(\zeta_j \circ \inversion) \circ \inversion \circ (\xi_k \circ \inversion)^{-1} = \zeta_j \circ \inversion \circ \xi_k^{-1}$, which is precisely the transition function at the seam in the complex conjugate of the sewn surfaces.
	Therefore, we have $\conj{(\Sigma_1 \sew{j}{k} \Sigma_2)} = \conj{\Sigma}_1 \sew{j}{k} \conj{\Sigma}_2$, and similarly, $\conj{(\sewself{j}{k} \Sigma)} = \sewself{j}{k} \conj{\Sigma}$.
	In Item~\ref{lemma:conj_double:act}, the action of $\phi \in \DefC$ at 
	$j$ is defined through analytic extension of the real-analytic boundary parametrization $(\zeta_j \circ \inversion)(z) = \zeta(\conj{z})$ of $\conj{\Sigma}$.
	Composition with $\conj{\phi} \circ \inversion$ yields $\zeta \circ \inversion$, which is the boundary parametrization of $\conj{(\Sigma \diffActing{j} \phi)}$.
	Hence, we have $\conj{(\Sigma \diffActing{j} \phi)} = \conj{\Sigma} \diffActing{j} (\conj{\phi} \circ \inversion)$.
	Assuming the case of a non-separating simple analytic loop~$\vartheta$, the definition of the unraveling operation (see \cite[Section~4.2]{Maibach-Peltola:Complex_deformations_of_the_circle}) and Item~\ref{lemma:conj_double:sew} imply that
	\begin{equation}
	\begin{aligned}
		\conj{(\Sigma \diffActing{\vartheta} \phi)}
		&= \sewself{\boundaries+1}{\boundaries+2} \conj{(\unravel{\vartheta} \Sigma \diffActing{\boundaries+2} \phi)}
		= \sewself{\boundaries+1}{\boundaries+2} \Big( \conj{(\unravel{\vartheta} \Sigma)} \diffActing{\boundaries+2} (\conj{\phi} \circ \inversion)\Big) \\
		&= \sewself{\boundaries+1}{\boundaries+2} \Big( (\unravel{(\vartheta \circ \inversion)} \conj{\Sigma}) \diffActing{\boundaries+2} (\conj{\phi} \circ \inversion)\Big) = \Sigma^* \diffActing{\vartheta \circ \inversion} (\phi^* \circ \inversion).
	\end{aligned}
	\end{equation}
	The separating case is computed analogously, also resulting in $\conj{(\Sigma \diffActing{\vartheta} \phi)} = \conj{\Sigma} \diffActing{\vartheta \circ \inversion} (\conj{\phi} \circ \inversion)$.

	An isomorphism $F \colon \param{\Sigma_1, \zeta_1, \dots, \zeta_{\boundaries}} \to \param{\Sigma_2, \xi_1, \dots, \xi_{\boundaries}}$ is also an isomorphism on the double  
	(it is the same diffeomorphism, which is still complex-analytic since both complex structures are complex conjugated).
	Moreover, it preserves the new boundary parametrizations, $F \circ \zeta_j \circ \inversion = \xi_j \circ \inversion$, proving the well-definedness of complex conjugation over Segal moduli space (Item~\ref{lemma:conj_double:conj}).
	Since the complex conjugate and sewing operation both preserve isomorphism classes, the doubling map is also well-defined (Item~\ref{lemma:conj_double:double}).
	Now, consider the sewing operation and the actions of $\DefC$, which generate the Fr\"olicher structure on the Segal moduli spaces (see \cite[Section~4.4]{Maibach-Peltola:Complex_deformations_of_the_circle} for details).
	In particular, the proven identities in Items~\ref{lemma:conj_double:sew} and~\ref{lemma:conj_double:act} show that an initial \gls{fsmooth} curve stays \gls{fsmooth} after conjugation or doubling of the surface, proving the \gls{fsmooth}ness in Items~\ref{lemma:conj_double:conj} and~\ref{lemma:conj_double:double}.
\end{proof}

Given a compact hyperbolic Riemann surface $\Sigma$ without boundary, it is uniformized by a covering map $\pi_\Sigma \colon \uhp \to \Sigma$. 
The \emph{Fuchsian projective structure} is the projective structure induced by the trivial projective structure on $\uhp$.
In the case of the sphere or torus, it is the projective structure induced from the universal covering by~$\hat{\C}$ or~$\C$.
The \emph{Fuchsian projective structure} on a compact connected Riemann surface with enumerated and analytically parametrized boundary components is the restriction of the that on the double $\Sigma \sewall \conj{\Sigma}$.

If $F \colon \Sigma_1 \to \Sigma_2$ is an isomorphism of Riemann surfaces with enumerated and analytically parametrized boundary components, then $\pi_{\Sigma_2 \sewall \conj{\Sigma}_2}$ and $(F \cup \conj{F}) \circ \pi_{\Sigma_1 \sewall \conj{\Sigma}_1}$ are related by a M\"obius transformation, so the Fuchsian projective structures on $\Sigma_1$ and $\Sigma_2$ are related by pullback, making them consistent over equivalence classes in $\moduligb$.
Since isomorphisms also preserve the boundary parametrizations, the following definition makes sense in~$\moduligb$.

\begin{definition}
	\label{def:fuchsian_projective}
	A compact connected Riemann surface with enumerated and analytically parametrized boundary components is \emph{Fuchsian projective} if it is projective 
	with respect to the Fuchsian projective structure.
	We denote the subset of Fuchsian projective surfaces by $\moduligbmobf \subset \moduligb$.
\end{definition}

\begin{proposition}
	\leavevmode
	\begin{enumerate}
	\item
		Hyperbolic surfaces with analytic boundary parametrizations are Fuchsian projective, and thus, we have a sequence of inclusions
		\label{prop:fprojective:hyperbolic}
		\begin{equation}
			\hmoduligb \subset \moduligbmobf \subset \moduligbmob \subset \moduligb.
			\label{eq:fprojective_inclusions}
		\end{equation}
	\item
		Any compact connected Riemann surface with enumerated and analytically parametrized boundary components is of the form
		$\Sigma \diffActing{1} \phi_1 \cdots \diffActing{\boundaries} \phi_{\boundaries}$, where $\Sigma$ is a Fuchsian projective surface and $\phi_1, \dots, \phi_{\boundaries} \in \Diffpan$ are diffeomorphisms.
		\label{prop:fprojective:repr}
	\item
			Let $\Sigma \in \moduligbmobf$ with $\boundaries \geq 1$, and $\phi \in \DefC$ and $1 \leq j \leq \boundaries$ be such that $\Sigma \diffActing{j} \phi$ exists. Then, we have $\Sigma \diffActing{j} \phi \in \moduligbmobf$ if and only if $\phi$ is a M\"obius transformation.
		\label{prop:fprojective:mobius}
	\end{enumerate}
	\label{prop:fprojective}
\end{proposition}

\begin{proof}
The uniform metric with geodesic boundary components may be obtained by restricting the uniform metric on the double obtained from uniformization.
Hence, analytic continuations of negatively oriented unit speed boundary parametrizations are indeed charts in the Fuchsian projective structure, proving the first inclusion in~\eqref{eq:fprojective_inclusions}. 
The other inclusions in~\eqref{eq:fprojective_inclusions} follow directly from the definitions.
For Item~\ref{prop:fprojective:repr}, the Fuchsian projective surface $\Sigma$ is found by taking negatively oriented unit speed boundary parametrizations in a constant curvature metric with geodesic boundary.
For Item~\ref{prop:fprojective:mobius}, note that Fuchsian projective structure is independent of the boundary parametrizations, so any boundary reparametrization preserving the Fuchsian structure must be a M\"obius transformation. 			
\end{proof}

\section{\texorpdfstring{$\Rp$}{R+}-extensions of Segal moduli spaces}
\label{section:modular_functor}
This section is devoted to the setup for our first main result, Theorem~\ref{thm:universality} (whose proof will be completed in Section~\ref{section:iso}).
We first show how a \gls{fsmooth} $\Rp$-extension (Definition~\ref{def:modular_functor}) induces a central extension the complex deformations by the multiplicative group $\Rp$ (Section~\ref{section:central_extension}), and how this central extension acts on the $\Rp$-extension (Section~\ref{section:action_central_extension}).
Moreover, we spell out a number of cocycle identities (Section~\ref{section:cocycles}), and discuss the notion of locality (Section~\ref{section:local_equiv}).
Then, we define the central charge of a \gls{fsmooth} $\Rp$-extension using the cohomology class of the Lie algebra cocycle of the central extension (Section~\ref{section:central_charge}).
Finally, we define natural, yet nontrivial, modular and crossing invariance properties of \gls{fsmooth} local $\Rp$-extension (Section~\ref{section:modular_invariance_properties}).
We show that these properties hold for the real determinant line bundle (Section~\ref{section:det_line_bundle}), and later, in Section~\ref{section:modular_invariance_proof}, we show that they actually hold for any local \gls{fsmooth} $\Rp$-extension (see Proposition~\ref{prop:E_prop}).

Let us begin by defining $\Rp$-extensions.
We use the operations on Segal moduli spaces from Section~\ref{section:setup}.
Additionally, we let $\blank \;_{j \leftrightarrow k} \colon \moduligb \to \moduligb$ denote the operation of swapping the labels $j$ and $k$ of boundary components.

\begin{definition}
	A \emph{\gls{fsmooth} $\Rp$-extension $\RMFE$ of the Segal moduli spaces} consists of \gls{fsmooth} principal $\Rp$-bundles
	\begin{equation}
		\RMFE(\moduligb) \longrightarrow \moduligb, \qquad \genus, \boundaries \geq 0,
	\end{equation}
	and \gls{fsmooth} morphisms of principal $\Rp$-bundles called \emph{sewing isomorphisms}\footnote{
		Since $\Rp$ is abelian, the fiber-wise tensor product, denoted by $\boxtimes$, is again a principal $G$-bundle over $\moduli{\genus_1}{\boundaries_1} \times \moduli{\genus_2}{\boundaries_2}$.
		Morphisms of principal $\Rp$-bundles are isomorphisms on fibers.
	},
	\begin{align}
		\blank \sewx{\RMFE}{j}{k\phantom{,j}} \blank &: \RMFE(\moduli{\genus_1}{\boundaries_2}) \boxtimes \RMFE(\moduli{\genus_2}{\boundaries_2}) \to \RMFE(\moduli{\genus_1 + \genus_2}{\boundaries_1 + \boundaries_2 - 2}),\\
		\sewxself{\RMFE}{j}{k} \blank &: \RMFE(\moduligb) \to \RMFE(\moduli{\genus + 1}{\boundaries - 2}), \\
		\blank \;_{j \overset{\scriptstyle \RMFE}{\leftrightarrow} k} &: \RMFE(\moduligb) \to \RMFE(\moduligb),
	\end{align}
	They cover the respective operations $\blank \sew{j}{k} \blank$, $\sewself{j}{k} \blank$, and $\blank \;_{j \leftrightarrow k}$ preserving associativity\footnote{
		We use the convention that we assign new labels to boundary components only after all sewing operations have been performed.
	}
	\begin{align}
		\label{eq:assoc_E}
		\blank \sewx{\RMFE}{j}{k} (\blank \sewx{\RMFE}{l}{m} \blank)
		&=
		(\blank \sewx{\RMFE}{j}{k} \blank) \sewx{\RMFE}{l}{m} \blank, \\
		\sewxself{\RMFE}{j}{k} (\blank \sewx{\RMFE}{l}{m} \blank)
		&=
		(\sewxself{\RMFE}{j}{k} \blank) \sewx{\RMFE}{l}{m} \blank, \\
		\sewxself{\RMFE}{j}{k} (\sewxself{\RMFE}{l}{m} \blank)
		&=
		\sewxself{\RMFE}{l}{m} (\sewxself{\RMFE}{j}{k} \blank),
	\end{align}
	and symmetry
	\begin{align}
		( \blank \;_{j \overset{\scriptstyle \RMFE}{\leftrightarrow} l} \sewx{\RMFE}{l}{k} \blank ) \sewx{\RMFE}{j}{m} \blank
		= \; & (\blank \sewx{\RMFE}{j}{k} \blank) \sewx{\RMFE}{l}{m} \blank, \\
		\alpha \sewx{\RMFE}{j}{k} \beta = \; &
		\beta \sewx{\RMFE}{k}{j} \alpha,
			\qquad \alpha \in \RMFE(\moduli{\genus_1}{\boundaries_2}), \beta \in \RMFE(\moduli{\genus_2}{\boundaries_2}), \\
		\sewxself{\RMFE}{j}{k} \alpha = \; & \sewxself{\RMFE}{k}{j} \alpha,
			\qquad \alpha \in \RMFE(\moduli{\genus}{\boundaries}).
	\end{align}
    \label{def:modular_functor}
\end{definition}

\begin{remark}
	Note that the principal $\Rp$-bundles are always trivial bundles.
	What is of interest is not the bundle structure, but the algebraic structure given by the sewing isomorphisms.
	Equivalently, $\Rp$-extensions may be defined as real oriented line bundles.
	Our main example, the real determinant line bundle, is orientable and nontrivial as an $\Rp$\nobreakdash-extension (or a real one-dimensional modular functor)~\cite[Corollary~1.5]{Maibach-Peltola:From_the_conformal_anomaly_to_the_Virasoro_algebra}.
\end{remark}

\begin{remark}
	The trivial $\Rp$-extension comprises the trivial bundles $\RMFE(\moduligb) = \Rp \times \moduligb$ endowed with the trivial sewing operations $(1, \Sigma_1) \sewx{\RMFE}{j}{k} (1, \Sigma_2) = (1, \Sigma_1 \sew{j}{k} \Sigma_2)$ and $\sewxself{\RMFE}{j}{k} (1, \Sigma) = (1, \sewself{j}{k} \Sigma)$.
	Alternatively, triviality is characterized by the existence of a trivialization such that all cocycles vanish (see Section~\ref{section:cocycles}).
\end{remark}

\subsection{Central extensions of complex deformations}
\label{section:central_extension}

To construct the central extensions, we use Segal's idea~\cite{Segal:Definition_of_CFT,Segal:Two-dimensional_conformal_field_theories_and_modular_functors,Segal:Definition_of_CFT_collection}, also used by Huang in \cite[Appendix~D]{Huang:2D_Conformal_geometry_and_VOAs} in genus $0$, and by Radnell in genus $1$~\cite[Chapter~6.3]{Radnell:PhD}.
There, central extensions of the diffeomorphism group $\Diffpan$ are constructed for complex determinant line bundles (or complex one-dimensional modular functors), relative to the standard annulus $\A_1$. 
In~\cite{Maibach-Peltola:From_the_conformal_anomaly_to_the_Virasoro_algebra}, we carried out this construction concretely for the real determinant line bundle and complex deformations, relative to the standard annulus $\A_1$. 
In the following, we will show that the central extension can be defined relative to arbitrary surfaces, and for a general \gls{fsmooth} $\Rp$-extension $\RMFE$.

Conceptually, a complex deformation $\phi \in \DefC$ encodes the difference between a surface $\Sr \in \moduligb$ and a deformed surface $\Sr \diffActing{j} \phi$.
The fiber $\RMFE(\Sr)$ is a $\Rp$-torsor, that is, a set on which $\Rp$ acts freely and transitively.
Since $\Rp$ is abelian, there are natural notions of tensor products $T_1 \otimes T_2$ and duals $T^\vee$ of $\Rp$-torsors $T$, with the evaluation map
\begin{equation}
	\evaluation \colon T \otimes T^\vee \longrightarrow \Rp, \qquad \alpha \otimes \alpha^\vee \longmapsto \alpha^\vee(\alpha) = 1.
	\label{eq:torsor_eval}
\end{equation}
For fibers of $\RMFE$, given that $\Sr \diffActing{j} \phi$ exists, we define
\begin{equation}
	\RMFE_{\Sr, j}(\phi) = \RMFE(\Sr \diffActing{j} \phi) \otimes \Big(\RMFE(\Sr)\Big)^\vee.
		\label{eq:fibers_of_E}
\end{equation}
Tensoring with the dual of $\RMFE(\Sr)$ is a way of ``dividing out'' the original surface $\Sr$ after the deformation by $\phi$, retaining only the structure of $\RMFE$ on the difference between the surfaces $\Sr \diffActing{j} \phi$ and $\Sr$.
In fact, the $\Rp$-torsors $\RMFE_{\Sr, j}(\phi)$ do not actually depend on $\Sr$ and $j$, in the sense that we can consistently identify them by taking a (co-)limit in the category of $\Rp$-torsors:
\begin{equation}
	\RMFE(\phi) = \varinjlim \RMFE_{\Sr, j}(\phi) = \varprojlim \RMFE_{\Sr, j}(\phi).
	\label{eq:E_phi_lim}
\end{equation}
The category of $\Rp$-torsors forms a connected groupoid, meaning that all morphisms are isomorphisms and any two objects are isomorphic (in many ways).
For such a category, limits and colimits of connected diagrams exist and agree if they are related by inverting the isomorphisms.
To elaborate this, we define the following diagrams for~\eqref{eq:E_phi_lim}:
\begin{itemize}
	\item
		The vertices are $\RMFE_{\Sr, j}(\phi)$ indexed by all surfaces $\Sr \in \moduligb$, for $\genus \geq 0$, $\boundaries \geq 1$, with at least one boundary component and $1 \leq j \leq \boundaries$ such that $\Sr \diffActing{j} \phi$ exists.
	\item
		For the limit in~\eqref{eq:E_phi_lim}, there is an edge from $\RMFE_{\Sra, j}(\phi)$ to $\RMFE_{\Srb, k}(\phi)$ if and only if there exists a boundary component $l \neq j$ of $\Sra$ and a third surface $\Srx$ with boundary component $m$, such that $\Srb = (\Sra \sew{l}{m} \Srx)_{j \leftrightarrow k}$.
		The edges are given by isomorphisms
		\begin{equation}
			\natisophi{\phi}{\Sra, j}{\Srb, k} \colon \RMFE_{\Sra, j}(\phi) \longrightarrow \RMFE_{\Srb, k}(\phi),
			\label{eq:def_natisophi_map}
		\end{equation}
		yet to be defined, such that for a loop and two consecutive edges, we respectively~have
		\begin{equation}
			\natisophi{\phi}{\Sra, j}{\Sra, j} = \id_{\RMFE_{\Sra, j}(\phi)} \qquad \textnormal{and} \qquad
			\natisophi{\phi}{\Srb, k}{\Src, l} \circ \natisophi{\phi}{\Sra, j}{\Srb, k} = \natisophi{\phi}{\Sra, j}{\Src, l}.
			\label{eq:natisophi_prop}
		\end{equation}
	\item
		For the colimit in~\eqref{eq:E_phi_lim}, we invert the edges for the limit, that is, $\natisophi{\phi}{\Srb, k}{\Sra, l} = \Big(\natisophi{\phi}{\Sra, j}{\Srb, k}\Big)^{-1}$.
\end{itemize}
These diagrams are indeed connected, since any two vertices have a common neighbor\footnote{For instance, a standard annulus $\A_\tau$, which for sufficiently small $\tau > 0$ embeds into any surface by aligning it at a boundary component through analytic continuation of the boundary parametrization.}. 
The isomorphisms~\eqref{eq:def_natisophi_map} are defined in terms of the sewing isomorphisms of $\RMFE$ associated to the relation $\Srb = (\Sra \sew{l}{m} \Srx)_{j \leftrightarrow k}$ as follows\footnote{
	$\big(\applyat{\evaluation^{\RMFE(\Srx)}}{2}{4}\big)^{-1}$ is the inverse of the evaluation map~\eqref{eq:torsor_eval} applied to the 2nd and 4th tensor components.
}:
\begin{equation}
	\label{eq:def_natiso_incl}
	\natisophi{\phi}{\Sra, j}{\Srb, k} =
	\Big(\big( \blank \sewx{\RMFE}{l}{m} \blank\big)_{j \overset{\scriptstyle \RMFE}{\leftrightarrow} k} \otimes \big(\blank \sewx{\RMFE}{l}{m} \blank\big)_{j \overset{\scriptstyle \RMFE}{\leftrightarrow} k}^\vee \Big) \circ \Big(\applyat{\evaluation^{\RMFE(\Srx)}}{2}{4}\Big)^{-1}.
\end{equation}
In words, this isomorphism takes an element $\alpha \otimes \beta^\vee \in \RMFE_{\Sra, j}(\phi)$ and tensors it with any $\delta \in \RMFE(\Srx)$ and the dual element $\delta^\vee\in \RMFE(\Srx)^\vee$ to obtain $\alpha \otimes \delta \otimes \beta^\vee \otimes \delta^\vee$.
Then, the sewing isomorphisms of $\RMFE$ are applied to sew $\alpha$ to $\delta$, and $\beta^\vee$ to $\delta^\vee$. 
(The latter is the uniquely defined sewing isomorphism induced on the dual of $\RMFE$.)
This results in
\begin{equation}
	\natisophi{\phi}{\Sra, k}{\Srb, l}(\alpha \otimes \beta^\vee)
	= (\alpha \sewx{\RMFE}{l}{m} \delta)_{j \overset{\scriptstyle \RMFE}{\leftrightarrow} k} \otimes (\beta \sewx{\RMFE}{l}{m} \delta)_{j \overset{\scriptstyle \RMFE}{\leftrightarrow} k}^\vee,
	\label{eq:def_natiso_incl_concrete}
\end{equation}
which is independent of the choice of $\delta \in \RMFE(\Srx)$.
The identity~\eqref{eq:natisophi_prop} for two consecutive edges follows from associativity of the sewing isomorphisms of $\RMFE$.
Note that the (co-)limit $\RMFE(\phi)$ gives us isomorphisms~\eqref{eq:def_natisophi_map} satisfying~\eqref{eq:natisophi_prop}, also between surfaces not connected by an edge, defined by composing isomorphisms to and from the limiting object.

To define the central extension for $\DefC$, we combine the $\Rp$-torsors $\RMFE(\phi)$ defined by Equation~\eqref{eq:E_phi_lim} into the fibers of a principal $\Rp$-bundle over $\DefC$,
\begin{equation}
	\RMFE(\DefC) = \bigsqcup_{\phi \in \DefC} \RMFE(\phi).
	\label{eq:E_phi_bundle}
\end{equation}
Trivializations $Z_{\disk, 1}(\phi) = Z(\disk \diffActing{1} \phi) \otimes \big(Z(\disk)\big)^\vee$ relative to a \gls{fsmooth} trivialization~$Z$ of~$\RMFE(\disks)$ generate a Fr\"olicher structure on $\RMFE(\DefC)$.
These trivializations are global (by Item~\ref{lemma:complex_deformations:unit_disk} of Lemma~\ref{lemma:complex_deformations}).
More generally, consider trivializations of the form
\begin{equation}
	\label{eq:local_triv_ext}
	Z_{\Sr, j}(\phi) = Z(\Sr \diffActing{j} \phi) \otimes \Big(Z(\Sr)\Big)^\vee, \qquad \Sr \in \moduligb, 1 \leq j \leq \boundaries.
\end{equation}
These trivializations are local, because they only make sense where $\Sr \diffActing{j} \phi$ is defined.

Finally, we turn the principal $\Rp$-bundle $\RMFE(\DefC)$ into a central extension of by $\Rp$ in the following manner (see also~\cite[Section~3.3]{Maibach-Peltola:From_the_conformal_anomaly_to_the_Virasoro_algebra})).
The composition law is a bundle map covering the composition law of $\DefC$.
For composable pairs $\phi, \psi \in \DefC$, the fiber-wise multiplication isomorphism
\begin{equation}
	\blank \diffActingE{} \blank \colon \RMFE(\phi) \otimes \RMFE(\psi) \longrightarrow \RMFE(\phi \circ \psi)
	\label{eq:E_defc_composition}
\end{equation}
is an isomorphism of (co-)limits~\eqref{eq:E_phi_lim}.
It is fully determined by representing it over a single triple of vertices in the corresponding diagrams in the following way.
Let $\Sr$ be a surface with boundary component $j$ such that $\Sr \diffActing{j} \phi$, and $\Sr \diffActing{j} \psi$, and $\Sr \diffActing{j} (\phi \circ \psi)$ exist, and define $\detmulta{\phi}{\psi}{\Sr, j} \colon \RMFE_{\Sr, j}(\phi) \otimes \RMFE_{\Sr, j}(\psi) \to \RMFE_{\Sr, j}(\phi \circ \psi)$ through the composition\footnote{Here $\exchange \colon v \otimes w \mapsto w \otimes v$ denotes the isomorphism that exchanges the tensor components.}
\begin{equation}
	\label{eq:def_mult_a}
	\detmulta{\phi}{\psi}{\Sr, j} = \exchange \circ \applyat{\evaluation^{\RMFE(\Sr \diffActing{j} \phi)}}{1}{4} \circ (\id \otimes \natisophi{\psi}{\Sr, j}{\Sr \diffActing{j} \phi, j}),
\end{equation}
which depends on the sewing isomorphisms of $\RMFE$ through~\eqref{eq:def_natiso_incl}.
Let us spell this out for elements $\alpha \otimes \beta^\vee \in \RMFE_{\Sr, j}(\phi)$ and $\gamma \otimes \beta^\vee \in \RMFE_{\Sr, j}(\psi)$, where the respective second factors may be chosen to agree without loss of generality.
The first step is to change the vertex of the second element by finding the unique $\delta \in \RMFE(\Sr \diffActing{j} (\phi \circ \psi))$ matching $\alpha^\vee$ such that
\begin{equation}
	\label{eq:assoc_nu}
	\natisophi{\psi}{\Sr, j}{\Sr \diffActing{j} \phi, j}(\gamma \otimes \beta^\vee) = \delta \otimes \alpha^\vee.
\end{equation}
Then, the evaluation and flip of $\alpha \otimes \beta^\vee \otimes \delta \otimes \alpha^\vee$ results in $\delta \otimes \beta^\vee$.
This defines the composition law (multiplication isomorphism) of Equation~\eqref{eq:E_defc_composition} as
\begin{equation}
	(\alpha \otimes \beta^\vee) \diffActingE{} (\delta \otimes \alpha^\vee) = \detmulta{\phi}{\psi}{\Sr, j}(\alpha \otimes \beta^\vee \otimes \gamma \otimes \beta^\vee) = \delta \otimes \beta^\vee,
	\label{eq:comp_law_E}
\end{equation}
where it is convenient to start with $\alpha \otimes \beta^\vee \in \RMFE_{\Sr, j}(\phi)$ and $\delta \otimes \beta^\vee \in \RMFE_{\Sr \diffActing{j} \phi, j}(\psi)$.

\begin{lemma}
	\leavevmode
	\begin{enumerate}
		\item
			The multiplication isomorphism~\eqref{eq:def_mult_a} is canonical in the sense that
			\begin{equation}
				\label{eq:mult_canon}
				\natisophi{\phi \circ \psi}{\Sra, j}{\Srb, k} \circ \detmulta{\phi}{\psi}{\Sra, j}
				= {\detmulta{\phi}{\psi}{\Srb, k}} \circ \Big(\natisophi{\phi}{\Sra, j}{\Srb, k} \otimes \natisophi{\psi}{\Sra, j}{\Srb, k}\Big)
			\end{equation}
			for common edges from $(\Sra, j)$ to $(\Srb, k)$ for $\phi$, $\psi$, and $\phi \circ \psi$.
		\item
			The composition law is associative in the sense that
			$(\blank \diffActingE{} \blank) \diffActingE{} \blank = \blank \diffActingE{} (\blank \diffActingE{} \blank)$.
		\item
			The composition law is \gls{fsmooth}.
		\end{enumerate}
\label{lemma:central_extension}
\end{lemma}
\begin{proof}
	Let $\Srb = (\Sra \sew{l}{m} \Srx)_{j \leftrightarrow k}$.
	We first consider the left-hand side of~\eqref{eq:mult_canon}.
	Starting with $\alpha \otimes \beta^\vee \in \RMFE_{\Sra, j}(\phi)$ and $\gamma \otimes \beta^\vee \in \RMFE_{\Sra, j}(\psi)$ and following Equation~\eqref{eq:comp_law_E}, the multiplication isomorphism maps the tensor product of these elements to $\delta \otimes \beta^\vee$, where $\delta$ is determined by~\eqref{eq:assoc_nu}.
	Taking any $\mu \in \RMFE(\Srx)$, we apply Equation~\eqref{eq:def_natiso_incl_concrete} to find
	\begin{equation}
		\label{eq:assoc_left}
		(\natisophi{\phi \psi}{\Sra, j}{\Srb, k} \circ \detmulta{\phi}{\psi}{\Sra, j})(\alpha \otimes \beta^\vee \otimes \gamma \otimes \beta^\vee)
		= (\delta \sewx{\RMFE}{l}{m} \mu)_{j \overset{\scriptstyle \RMFE}{\leftrightarrow} k} \otimes (\beta \sewx{\RMFE}{l}{m} \mu)_{j \overset{\scriptstyle \RMFE}{\leftrightarrow} k}^\vee.
	\end{equation}
	For the right-hand side of~\eqref{eq:mult_canon}, we can factor through $\RMFE_{\Sra \diffActing{j} \phi, j}(\psi)$, where the defining relation~\eqref{eq:assoc_nu} for $\delta$ applies.
	We thus obtain
	\begin{equation}
		\exchange \circ \applyat{\evaluation}{1}{4} \circ \Big(\natisophi{\phi}{\Sra, j}{\Srb, k} \otimes \underbrace{\big(
			\natisophi{\psi}{\Srb, k}{\Srb \diffActing{k} \phi, k} \circ \natisophi{\psi}{\Sra, j}{\Srb, k}
		\big)}_{
			\mathclap{\textstyle
			= \natisophi{\psi}{\Sra \diffActing{j} \phi, j}{\Srb \diffActing{k} \phi, k} \circ \natisophi{\psi}{\Sra, j}{\Sra \diffActing{j} \phi, j}
		}} \Big).
	\end{equation}
	Before the evaluation and flip, $\alpha \otimes \beta^\vee \otimes \gamma \otimes \beta^\vee$ is mapped to
	\begin{equation}
		(\alpha \sewx{\RMFE}{l}{m} \mu)_{j \overset{\scriptstyle \RMFE}{\leftrightarrow} k}
		\otimes (\beta \sewx{\RMFE}{l}{m} \mu)_{j \overset{\scriptstyle \RMFE}{\leftrightarrow} k}^\vee
		\otimes (\delta \sewx{\RMFE}{l}{m} \mu)_{j \overset{\scriptstyle \RMFE}{\leftrightarrow} k}
		\otimes (\alpha \sewx{\RMFE}{l}{m} \mu)_{j \overset{\scriptstyle \RMFE}{\leftrightarrow} k}^\vee.
	\end{equation}
	After applying the final evaluation and flip, we find that the right-hand side of~\eqref{eq:mult_canon} is equal to left-hand side, by Equation~\eqref{eq:assoc_left}.
	Associativity follows by arguments similar to the proof of~\cite[Theorem~3.13]{Maibach-Peltola:From_the_conformal_anomaly_to_the_Virasoro_algebra} (which readily generalize to the present setting).
	Lastly, \gls{fsmooth}ness follows from the \gls{fsmooth}ness of the sewing isomorphisms of $\RMFE$, since other than those, the composition law is given by (\gls{fsmooth}) linear operations.
\end{proof}

\begin{theorem}
	\label{thm:central_extension}
	The following sequence is exact:
	\begin{equation}
	\begin{alignedat}{2}
\{1\} \; \longrightarrow \; \Rp \; & \longhookrightarrow \; \RMFE(\DiffC) && \; \longtwoheadrightarrow \; \DiffC \; \longrightarrow \; \{1\} ,
		\label{eq:defc_ext_sequence}
\end{alignedat}
	\end{equation}
	 where the second map is the inverse of $\evaluation \colon  \RMFE(\id) \to \Rp$ and the third map is the projection of the bundle~\eqref{eq:E_phi_bundle}.	It defines a \emph{central extension} in the sense that, in addition to exactness, the maps are \gls{fsmooth} and $\Rp$ maps into the center of $\RMFE(\DefC)$.
\end{theorem}
\begin{proof}
	The inverse of the evaluation map is $\evaluation^{-1}(\lambda) = \lambda (\gamma \otimes \gamma^\vee) \in \RMFE_{\Sr, j}(\id)$, represented by an arbitrary $\gamma \in \RMFE(\Sr)$ for some choice of surface $\Sr$ and boundary component $j$.
	The composition of two such elements $\lambda (\gamma \otimes \gamma^\vee)$ and $\mu (\gamma\otimes \gamma^\vee)$ in $\RMFE(\id)$, for $\lambda, \mu \in \Rp$, is
	$\detmulta{\id}{\id}{\Sr, j}\big(\lambda (\gamma \otimes \gamma^\vee) \otimes \mu (\gamma \otimes \gamma^\vee)\big) = (\lambda \; \mu) (\gamma \otimes \gamma^\vee)$,
	since the isomorphism in Equation~\eqref{eq:assoc_nu} is the identity in this case.
	Note that this agrees with the multiplication in $\Rp$.
	The kernel of the projection $\RMFE(\DefC) \twoheadrightarrow \DefC$ is the fiber $\RMFE(\id)$ over $\id \in \DefC$, which is clearly the image of this map, showing that the sequence in~\eqref{eq:defc_ext_sequence} is indeed exact.

	The involved maps are clearly \gls{fsmooth}.
	Lastly, let $\alpha \otimes \beta^\vee \in \RMFE_{\Sr, j}(\phi)$ be an element in the fiber of $\phi \in \DefC$ represented using a surface and boundary component such that $\Sr \diffActing{j} \phi$ and $\Sr \diffActing{j} \phi^{-1}$ exist.
	We represent any element of $\RMFE(\id)$ by $\lambda (\alpha \otimes \alpha^\vee) \in \RMFE_{\Sr \diffActing{1} \phi, j}(\id)$, and equivalently by $\lambda (\beta \otimes \beta^\vee) \in \RMFE_{\Sr \diffActing{j} \phi^{-1}, j}(\id)$, where $\lambda \in \Rp$.
	By Equation~\eqref{eq:comp_law_E}, we have
	\begin{equation}
		(\alpha \otimes \beta^\vee) \diffActingE{} \big(\lambda (\alpha \otimes \alpha^\vee)\big) 
		= \lambda (\alpha \otimes \beta^\vee)
		= \big(\lambda (\beta \otimes \beta^\vee)\big) \diffActingE{} (\alpha \otimes \beta^\vee).
	\end{equation}
	Hence, the image of the second map in~\eqref{eq:defc_ext_sequence} is central in $\RMFE(\DefC)$.
\end{proof}

\begin{remark}
	The inverse of an element $\alpha \otimes \beta^\vee \in \RMFE_{\Sr, j}(\phi) = \RMFE(\Sr \diffActing{j} \phi) \otimes \RMFE(\Sr)$ with respect to $\blank \diffActingE{} \blank$ can conveniently be represented by
	\begin{equation}
		(\alpha \otimes \beta^\vee)^{-1}
		= \beta \otimes \alpha^\vee \in \Big(\RMFE_{\Sr \diffActing{j} \phi, j}(\phi^{-1})\Big)^\vee = \RMFE(\Sr) \otimes \Big(\RMFE(\Sr \diffActing{j} \phi)\Big)^\vee.
	\end{equation}
	Indeed, then we have already found the right-hand side of Equation~\eqref{eq:assoc_nu} such that
	\begin{equation}
		\detmulta{\phi}{\phi^{-1}}{\Sr, j}(\alpha \otimes \beta^\vee \otimes \beta \otimes \alpha^\vee)
		= \exchange \circ \applyat{\evaluation^{\RMFE(\Sr \diffActing{j} \phi)}}{1}{4}(\alpha \otimes \beta^\vee \otimes \beta \otimes \alpha^\vee)
		= \beta \otimes \beta^\vee.
	\end{equation}
	\label{remark:inverse_E}
\end{remark}

\subsection{Action of central extensions on \texorpdfstring{$\Rp$}{R+}-extensions}
\label{section:action_central_extension}

Consider an element $\gamma \in \RMFE(\Sigma)$ in a fiber over any surface $\Sigma \in \moduligb$.
Fix a boundary component~$j$ and $\phi \in \DefC$ such that the deformed surface $\Sigma \diffActing{j} \phi$ exists.
Any element of $\RMFE(\phi)$ may be represented by $\alpha \otimes \gamma^\vee \in \RMFE_{\Sigma, j}(\phi)$.
Conceptually, we can read the formula $\alpha \otimes \gamma^\vee$ as ``replace $\gamma$ by $\alpha$'' (this is what $\phi$ does geometrically when deforming the surface~$\Sigma$ in~\eqref{eq:fibers_of_E}).
Indeed, this defines actions of the central extension $\RMFE(\DefC)$ on the principal $\Rp$-bundles $\RMFE(\moduligb)$ covering the actions of $\DefC$ on $\moduligb$.
The latter is partially defined on a set denoted $\compdef{\genus}{\boundaries}{j} = \setsuchthat{(\Sigma, \phi) \in \moduligb \times \DefC}{\textnormal{$\Sigma \diffActing{j} \phi$ exists}}$.
\begin{proposition}
	There exist \gls{fsmooth} bundle morphisms
	\begin{equation}
	\begin{aligned}
		\blank \diffActingE{j} \blank \colon \RMFE(\compdef{\genus}{\boundaries}{j}) &\longrightarrow \RMFE(\moduligb), \\
		(\gamma, \alpha \otimes \gamma^\vee) &\longmapsto \alpha, 
	\end{aligned}
	\qquad 1 \leq j \leq \boundaries,
	\label{eq:def_action_E}	
	\end{equation}
	such that the following compatibilities with the sewing isomorphisms of $\RMFE$ hold:
	\label{prop:actionsewingE}
	\begin{enumerate}
		\item
			$
			(\blank \diffActingE{j} \blank) \diffActingE{j} \blank
			= \blank \diffActingE{j} (\blank \diffActingE{j} \blank)
			$,
			\label{prop:actionsewingE_samesew}
		\item
			$
			(\blank	 \sewx{\RMFE}{j}{k} \blank) \diffActingE{l} \blank
			= \blank \sewx{\RMFE}{j}{k} (\blank \diffActingE{l} \blank)
			$, or $
			(\blank	 \sewx{\RMFE}{j}{k} \blank) \diffActingE{l} \blank
			= (\blank \diffActingE{l} \blank) \sewx{\RMFE}{j}{k} \blank
			$,
			depending on which of the surfaces has the boundary component $l$,
			\label{prop:actionsewingE_assoc}
		\item
			$
			(\blank \diffActingE{j} \blank) \diffActingE{k} \blank
			= (\blank \diffActingE{k} \blank) \diffActingE{j} \blank
			$,
			\label{prop:actionsewingE_diffsew}
		\item
			$
			(\sewxself{E}{j}{k} \blank) \diffActingE{l} \blank
			= \sewxself{E}{j}{k} (\blank \diffActingE{l} \blank)
			$, and
			\label{prop:actionsewingE_selfsew}
		\item
			there exists a unique \gls{fsmooth} involution $\invinvx{\RMFE} \colon \RMFE(\definvset) \to \RMFE(\definvset)$, which is a bundle morphism covering $\invinv(\phi) = \inversion \circ \phi^{-1} \circ \inversion$ \textnormal{(}defined in Equation~\eqref{eq:def_invinv}\textnormal{)} such that
			\begin{equation}
				(\blank \diffActingE{j} \blank) \sewx{\RMFE}{j}{k} \blank = \blank \sewx{\RMFE}{j}{k} (\blank \diffActingE{k} \invinvx{\RMFE} \blank), \qquad \sewxself{\RMFE}{j}{k} (\blank \diffActingE{j} \blank) = \sewxself{\RMFE}{j}{k} (\blank \diffActingE{k} \invinvx{\RMFE} \blank).
			\label{eq:actionsewingE:invinv}
			\end{equation}
			\label{prop:actionsewingE:invinv}
	\end{enumerate}
\end{proposition}
Note that these properties are rather elementary, covering the various compatibility relations of the actions of $\DefC$ and the sewing operations.
In particular, Item~\ref{prop:actionsewingE:invinv} covers the relation in Equation~\eqref{eq:act_between} of acting by a complex deformation at the seam of a sewing operation.
\begin{proof}
	Let $Z$ denote a trivialization of $\RMFE$. Consider the corresponding trivialization~\eqref{eq:local_triv_ext} of $\RMFE(\DefC)$ relative to a fixed surface $\Sigma$ and boundary component $j$.
	The action~\eqref{prop:actionsewingE} is defined by changing the representative from $\Sigma$ and $j$ to match the surface whose fiber is acted on.
	As the change of representative uses the isomorphisms~\eqref{eq:def_natisophi_map}, which are defined using the \gls{fsmooth} sewing operation of $\RMFE$, we conclude that $\blank \diffActingE{j} \blank$ is smooth.
\begin{enumerate}
	\item
		In Item~\ref{prop:actionsewingE_samesew}, two elements of $\RMFE(\DefC)$ act consecutively at the same boundary component.  
		We represent the arbitrary elements as $\alpha \otimes \gamma^\vee \in \RMFE_{\Sigma, j}(\phi)$ and $\beta \otimes \alpha^\vee \in \RMFE_{\Sigma \diffActing{j} \phi, j}(\psi)$, using the same $\gamma \in \RMFE(\Sigma)$ without loss of generality.
		On the one hand, we have
		\begin{equation}
			\Big(\gamma \diffActingE{j} (\alpha \otimes \gamma^\vee) \Big) \diffActingE{j} (\beta \otimes \alpha^\vee)
			= 
			\alpha \diffActingE{j} (\beta \otimes \alpha^\vee)
			= \beta,
		\end{equation}
		and on the other hand, by Equation~\eqref{eq:comp_law_E}, we have
		\begin{equation}
			\gamma \diffActingE{j} \Big(
				(\alpha \otimes \gamma^\vee) \diffActingE{}
				(\beta \otimes \alpha^\vee)
			\Big)
			= \gamma \diffActingE{j} (\beta \otimes \gamma^\vee)
			= \beta.
		\end{equation}
	\item
		For the first case, with the underlying relation $(\Sigma_1 \sew{j}{k} \Sigma_2) \diffActing{l} \phi = \Sigma_1 \sew{j}{k} (\Sigma_2 \diffActing{l} \phi)$, take $\gamma \in \RMFE(\Sigma_2)$ and $\mu \in \RMFE(\Sigma_1)$.
		Then, the element in $\RMFE(\phi)$ may be represented by either $\alpha \otimes \gamma^\vee \in \RMFE_{\Sigma_2, l}(\phi)$ or $(\mu \sewx{\RMFE}{j}{k} \alpha) \otimes (\mu \sewx{\RMFE}{j}{k} \gamma)^\vee \in \RMFE_{\Sigma_1 \sew{j}{k} \Sigma_2, l}(\phi)$ for some $\alpha \in \RMFE(\Sigma_2 \diffActing{j} \phi)$, where the latter representative is independent of $\mu$.
		The asserted relation follows:
		\begin{equation}
			\mu \sewx{\RMFE}{j}{k} \Big(\gamma \diffActingE{} (\alpha \otimes \gamma^\vee)\Big) = \mu \sewx{\RMFE}{j}{k} \alpha =
			(\mu \sewx{\RMFE}{j}{k} \gamma) \diffActingE{j} \Big((\mu \sewx{\RMFE}{j}{k} \alpha) \otimes (\mu \sewx{\RMFE}{j}{k} \gamma)^\vee\Big).
		\end{equation}
		The second case of Item~\ref{prop:actionsewingE_assoc} follows analogously.
	\item
		Item~\ref{prop:actionsewingE_diffsew} follows from Item~\ref{prop:actionsewingE_assoc} by cutting the surface in two, so that the boundary components $j$ and $k$ end up in distinct parts, and the underlying equation becomes
		\begin{equation}
			\Big((\Sigma_1 \sew{l}{m} \Sigma_2) \diffActing{j} \phi\Big) \diffActing{k} \psi = (\Sigma_1 \diffActing{j} \phi) \sew{l}{m} (\Sigma_2 \diffActing{k} \psi) = \Big((\Sigma_1 \sew{l}{m} \Sigma_2) \diffActing{k} \psi\Big) \diffActing{j} \phi.
		\end{equation}
	\item
		The self-sewn surface in Item~\ref{prop:actionsewingE_selfsew} may similarly be cut so that the self-sewing and action happen in different connected components.
		The asserted equation then follows from Item~\ref{prop:actionsewingE_assoc} (using also the compatibility of sewing and self-sewing from Definition~\ref{def:modular_functor}).
	\item
		The asserted identities~\eqref{eq:actionsewingE:invinv} cover the relations $(\Sigma_1 \diffActing{j} \phi) \sew{j}{k} \Sigma_2 = \Sigma_1 \sew{j}{k} (\Sigma_2 \diffActing{k} \invinv(\phi))$ and $\sewself{j}{k} (\Sigma \diffActing{j} \phi) = \sewself{j}{k} (\Sigma \diffActing{k} \invinv(\phi))$.
		Concretely, the involution $\invinvx{\RMFE}$ maps representative $\alpha \otimes \gamma^\vee \in \RMFE_{\Sigma_1, j}(\phi)$ to a representative $\beta \otimes \delta^\vee \in \RMFE_{\Sigma_2, k}(\invinv(\phi))$ satisfying
		\begin{equation}
			\Big(\gamma \diffActingE{j} (\alpha \otimes \gamma^\vee)\Big) \sewx{\RMFE}{j}{k} \delta
			= \alpha \sewx{\RMFE}{j}{k} \delta
			= \gamma \sewx{\RMFE}{j}{k} \beta
			= \gamma \sewx{\RMFE}{j}{k} \Big(\delta \diffActingE{k} (\beta \otimes \delta^\vee)\Big),
		\label{eq:sewing_E_between}
		\end{equation}
		and equivalently for representatives $\alpha \otimes \gamma^\vee \in \RMFE_{\Sigma, j}(\phi)$ and $\beta \otimes \gamma^\vee \in \RMFE_{\Sigma, k}(\invinv(\phi))$,
		\begin{equation}
			\sewxself{\RMFE}{j}{k} \Big(\gamma \diffActingE{j} (\alpha \otimes \gamma^\vee)\Big) 
			= \sewxself{\RMFE}{j}{k} \alpha
			= \sewxself{\RMFE}{j}{k} \beta
			= \sewxself{\RMFE}{j}{k} \Big(\gamma \diffActingE{k} (\beta \otimes \gamma^\vee)\Big).
		\label{eq:sewing_E_between_self}
		\end{equation}
		Fixing $\alpha$ and $\gamma$ for the first identity, we may take any $\delta$ and treat Equation~\eqref{eq:sewing_E_between} as a linear equation, solving it for the unique~$\beta$.
		The representative $\beta \otimes \delta^\vee \in \RMFE_{\Sigma_2, k}(\invinv(\phi))$ is then independent of the choice of~$\delta$.
		Moreover, changing the surfaces $\Sigma_1$ and $\Sigma_2$ or the boundary components does not affect the involution, since the change along an edge in the sense of~\eqref{eq:def_natisophi_map} may be resolved using Item~\ref{prop:actionsewingE_assoc}.
		The second identity similarly yields a unique $\beta$, and by Item~\ref{prop:actionsewingE_selfsew}, it is independent under a change of surface. 
		The compatibility of sewing and self-sewing in Definition~\ref{def:modular_functor} implies that the sewing and self-sewing approaches equivalently define $\invinvx{\RMFE}$.
		\qedhere
	\end{enumerate}
\end{proof}

\subsection{Various cocycle identities}
\label{section:cocycles}

As before, let $\RMFE$ be a \gls{fsmooth} $\Rp$-extension.
Let $Z$ collectively denote a trivialization of $\RMFE$, which comprises \gls{fsmooth} trivializations of the principal $\Rp$-bundles $Z_{\genus, \boundaries} \colon \moduligb \to \RMFE(\moduligb)$ for $\genus, \boundaries \geq 0$.
For each choice $j,k$ of boundary labels, 
the trivialization defines the \emph{sewing cocycle} $\Omega^Z_{j, k} \colon \moduli{\genus_1}{\boundaries_1} \times \moduli{\genus_2}{\boundaries_2} \to \R$ via the relation
\begin{equation}
	\label{eq:Ecocycle_pair}
	Z(\Sigma_1) \sewx{\RMFE}{j}{k} Z(\Sigma_2) = e^{\Omega^Z_{j, k}(\Sigma_1, \Sigma_2)} Z(\Sigma_1 \sew{j}{k} \Sigma_2).
\end{equation}
Similarly, we define the \emph{self-sewing cocycles} $\Omega^Z_{j, k} \colon \moduli{\genus}{\boundaries} \to \R$ associated to self-sewing at distinct boundary labels
$j \neq k$ via the relation
\begin{equation}
	\sewxself{\RMFE}{j}{k} Z(\Sigma) = e^{\Omega^Z_{j, k}(\Sigma)} Z(\sewself{j}{k} \Sigma).
\end{equation}
We call these functions the \emph{cocycles} of $\RMFE$ with respect to the trivialization $Z$. They satisfy \emph{cocycle identities} arising from the associativity of the sewing isomorphisms of $\RMFE$:
\begin{align}
	\Omega^Z_{j, k}(\Sigma_1, \Sigma_2) + \Omega^Z_{l, m}(\Sigma_1 \sew{j}{k} \Sigma_2, \Sigma_3)
	&= \Omega^Z_{j, k}(\Sigma_1, \Sigma_2 \sew{l}{m} \Sigma_3) + \Omega^Z_{l, m}(\Sigma_2, \Sigma_3),
	\label{eq:basic_cocycle_idenity}
	\\
	\Omega^Z_{j, k}(\Sigma_1) + \Omega^Z_{l, m}(\sewself{j}{k} \Sigma_1, \Sigma_2)
	&= \Omega^Z_{j, k}(\Sigma_1 \sew{l}{m} \Sigma_2) + \Omega^Z_{l, m}(\Sigma_1, \Sigma_2),
	\label{eq:basic_self_cocycle_idenity}
	\\
	\Omega^Z_{j, k}(\Sigma_1 \sew{l}{m} \Sigma_2) + \Omega^Z_{l,m}(\Sigma_1, \Sigma_2)
	&= \Omega^Z_{l, m}(\Sigma_1 \sew{j}{k} \Sigma_2) + \Omega^Z_{j,k}(\Sigma_1, \Sigma_2)
	,
	\label{eq:cycle_sewing}
	\\
	\Omega^Z_{j, k}(\Sigma) + \Omega^Z_{l, m}(\sewself{j}{k} \Sigma)
	&= \Omega^Z_{j, k}(\sewself{l}{m} \Sigma) + \Omega^Z_{l, m}(\Sigma).
	\label{eq:cycle_selfsewing}
\end{align}
In Equation~\eqref{eq:local_triv_ext}, we already defined local trivializations $Z_{S, j}(\phi)$ of the central extension $\RMFE(\DiffC)$ on the subset of $\phi \in \DefC$ such that $S \diffActing{j} \phi$ exists.
The composition law of the central extension defines a \emph{deformation-deformation cocycle}
\begin{equation}
	\label{eq:def_cocycle_def_def}
	Z_{\Sr, j}(\phi) \diffActingE{} Z_{\Sr, j}(\psi) = e^{\Omega^Z_{\Sr, j}(\phi, \psi)} Z_{\Sr, j}(\phi \circ \psi), \qquad \phi, \psi \in \defmultset,
\end{equation}
which is defined where the trivialization is defined at all three complex deformations $\phi$, $\psi$, and $\phi \circ \psi$.
For surfaces related by $\Srb= (\Sra \sew{l}{m} \Srx)_{j \leftrightarrow k}$ as in the definition of an edge in the limit~\eqref{eq:E_phi_lim}, the trivializations are related by
\begin{equation}
	\label{eq:triv_def_change_sigma}
	Z_{\Sra, j}(\phi)
	=
	e^{
		\Omega^Z_{l, m}(\Sra \diffActing{j} \phi,\Srx)
		- \Omega^Z_{l, m}(\Sra, \Srx)
	}
	Z_{\Srb, j}(\phi). 
\end{equation}
By Equation~\eqref{eq:mult_canon}, the above cocycles are related by the identity
\begin{equation}
\label{eq:cocycle_AB_relation}
\begin{aligned}
	&\Omega^Z_{\Sra, j}(\phi, \psi) - \Omega^Z_{\Srb, k}(\phi, \psi)
	\\ = \; &
	\Omega^Z_{l, m}(\Srx, \Sra \diffActing{j} \phi)
	+ \Omega^Z_{l, m}(\Srx, \Sra \diffActing{j} \psi)
	- \Omega^Z_{l, m}(\Srx, \Sra \diffActing{j} (\phi \circ \psi))
	- \Omega^Z_{l, m}(\Srx, \Sra).
\end{aligned}
\end{equation}
Observe that in the cohomology group defined in~\cite[Section~3]{Maibach-Peltola:Complex_deformations_of_the_circle}, the right-hand side of~\eqref{eq:def_def_def} is the differential of the function $\Omega^Z_{l, m}(\Sra \diffActing{j} \blank, \Srx) - \Omega^Z_{l, m}(\Sra, \Srx)$.
Associativity of the composition law of the central extension $\RMFE(\DiffC)$ (see Lemma~\ref{lemma:central_extension}) implies the \emph{deformation-deformation-deformation cocycle identity}
\begin{equation}
	\label{eq:def_def_def}
	\Omega^Z_{\Sr, j}(\phi_1 \circ \phi_2, \phi_3) +
	\Omega^Z_{\Sr, j}(\phi_1, \phi_2) = 
	\Omega^Z_{\Sr, j}(\phi_1, \phi_2 \circ \phi_3) +
	\Omega^Z_{\Sr, j}(\phi_2, \phi_3).
\end{equation}
The action of $\RMFE(\DiffC)$ on $\RMFE(\moduligb)$ (see Proposition~\ref{prop:actionsewingE}) yields the \emph{mixed cocycles}
\begin{equation}
	\label{eq:def_cocycle_mixed}
	Z(\Sigma) \diffActingE{n} Z_{\Sr, j}(\phi) = e^{\Omega^Z_{{\Sr, j, n}}(\Sigma, \phi)} Z(\Sigma \diffActing{n} \phi),
\end{equation}
where $\Sr$ is the surface relative to which the trivialization of $\RMFE(\DefC)$ is defined, and~$\Sigma$ is another surface whose $n$th boundary component the complex deformation ${\phi \in \DefC}$ acts on.
For $\Srb = (\Sra \sew{l}{m} \Srx)_{j \leftrightarrow k}$ as above, the cocycles are related by
\begin{equation}
	\label{eq:mixed_cocycle_change}
	\Omega^Z_{\Sra, j, n}(\Sigma, \phi) - \Omega^Z_{\Srb, k, n}(\Sigma, \phi)
	=
	\Omega^Z_{l, m}(\Sra \diffActing{j} \phi, \Srx) - \Omega^Z_{l, m}(\Sra, \Srx),
\end{equation}
which is the function whose differential appears on the right-hand side of~\eqref{eq:cocycle_AB_relation}, and which is independent of $\Sigma$.
If we take $\Sigma = \Sr$ and $j = n$ in Equation~\eqref{eq:def_cocycle_mixed}, then the definition of the action in Proposition~\ref{prop:actionsewingE} implies that
\begin{equation}
	\label{eq:cocycle_surfaces_equal_vanishing}
	\Omega^Z_{{\Sigma, j, j}}(\Sigma, \phi) = 0, \qquad \phi \in \DefC.
\end{equation}
By Items~\ref{prop:actionsewingE_assoc}~\&~\ref{prop:actionsewingE:invinv} of Proposition~\ref{prop:actionsewingE}, the mixed cocycles satisfy \emph{mixed cocycle identities},
\begin{align}
	\Omega^Z_{j,k}(\Sigma_1, \Sigma_2 \diffActing{l} \phi)
	+ \Omega^Z_{\Sr, m, l}(\Sigma_2, \phi)
	&= \Omega^Z_{\Sr, m, l}(\Sigma_1 \sew{j}{k} \Sigma_2, \phi)
	+ \Omega^Z_{j,k}(\Sigma_1, \Sigma_2),
	\label{eq:cocycle_id_mixed_1}
	\\
	\Omega^Z_{j,k}(\Sigma_1 \diffActing{l} \phi, \Sigma_2)
	+ \Omega^Z_{\Sr, m, l}(\Sigma_1, \phi)
	&= \Omega^Z_{\Sr, m, l}(\Sigma_1 \sew{j}{k} \Sigma_2, \phi)
	+ \Omega^Z_{j,k}(\Sigma_1, \Sigma_2),
	\label{eq:cocycle_id_mixed_2}
	\\
	\Omega^Z_{j,k}(\Sigma_1 \diffActing{j} \phi, \Sigma_2)
	+ \Omega^Z_{\Sr, m, j}(\Sigma_1, \phi)
	&= \Omega^Z_{j,k}(\Sigma_1, \Sigma_2 \diffActing{k} \invinv(\phi)) + \Omega^Z_{\Sr, m, k}(\Sigma_2, \invinv(\phi)),
	\label{eq:cocycle_id_mixed_3}
\end{align}
depending on where $\phi$ acts.
Moreover, by Items~\ref{prop:actionsewingE_samesew}~\&~\ref{prop:actionsewingE_diffsew} of Proposition~\ref{prop:actionsewingE}, we have
\begin{align}
	\Omega^Z_{\Sr, l, k}(\Sigma \diffActing{j} \phi, \psi) + \Omega^Z_{\Sr, l, j}(\Sigma, \phi)
	&= \Omega^Z_{\Sr, l, j}(\Sigma \diffActing{k} \psi, \phi) + \Omega^Z_{\Sr, l, k}(\Sigma, \psi),
	\label{eq:cocycle_mixed_compat_1}
	\\
	\Omega^Z_{\Sr, l, j}(\Sigma \diffActing{j} \phi, \psi) + \Omega^Z_{\Sr, l, j}(\Sigma, \phi)
	&= \Omega^Z_{\Sr, l, j}(\Sigma, \phi \circ \psi) + \Omega^Z_{\Sr, l}(\phi, \psi),
	\label{eq:cocycle_mixed_compat_2}
\end{align}
and more cocycle identities follow from self-sewing in Items~\ref{prop:actionsewingE_selfsew}~\&~\ref{prop:actionsewingE:invinv} of Proposition~\ref{prop:actionsewingE}:
\begin{align}
	\Omega^Z_{j, k}(\Sigma \diffActing{l} \phi) + \Omega^Z_{\Sr, m, l}(\Sigma, \phi)
	&=
	\Omega^Z_{j, k}(\Sigma) + \Omega^Z_{\Sr, m, l}(\sewself{j}{k} \Sigma, \phi), \\
	\Omega^Z_{j, k}(\Sigma \diffActing{j} \phi) + \Omega^Z_{\Sr, m, j}(\Sigma, \phi)
	&= 
		\Omega^Z_{j, k}(\Sigma \diffActing{k} \invinv(\phi)) + \Omega^Z_{\Sr, m, k}(\Sigma, \invinv(\phi)).
\end{align}

Lastly, let us return to the involution $\invinvx{\RMFE}$ defined by Item~\ref{prop:actionsewingE:invinv} of Proposition~\ref{prop:actionsewingE}.
The cocycle identities yield the following corollary.
\begin{corollary}
	For any $\Sr \in \moduligb$, and $1 \leq j \leq \boundaries$, and $\phi \in \DefC$ such that $\Sr \diffActing{j} \phi$ and $\Sr \diffActing{j} \invinv(\phi)$ exist, we have $\invinvx{\RMFE}(Z_{\Sr, j}(\phi)) = Z_{\Sr, j}(\invinv(\phi))$.
	\label{cor:involution_trivial}
\end{corollary}
\begin{proof}
	 Equation~\eqref{eq:sewing_E_between} defining $\invinvx{\RMFE}$ reads
	\begin{equation}
		\Big(Z(\Sigma_1) \diffActingE{j} Z_{\Sr, l}(\phi)\Big) \sewx{\RMFE}{j}{k} Z(\Sigma_2)
		=
		Z(\Sigma_1) \sewx{\RMFE}{j}{k} \Big(Z(\Sigma_2) \diffActingE{k} \invinvx{\RMFE}\big(Z_{\Sr, l}(\phi)\big)\Big).
		\label{eq:sewing_E_between_Z}
	\end{equation}
	Consider the cocycle 
	\begin{equation}
		\invinvx{\RMFE}(Z_{\Sr, m}(\phi)) = e^{\Omega_{\Sr, m}(\phi)} Z_{\Sr, m}(\invinv(\phi)), \qquad \phi \in \definvset.
	\end{equation}
	It depends only on a single invertible complex deformation such that $\Sr \diffActing{j} \phi$ exists, and has the symmetry $\Omega_{\Sr, m}(\invinv(\phi)) = - \Omega_{\Sr, m}(\phi)$.
	In terms of cocycles, Equation~\eqref{eq:sewing_E_between_Z} reads
\begin{equation}
	\Omega^Z_{j, k}(\Sigma_1 \diffActing{j} \phi, \Sigma_2) + \Omega^Z_{\Sr, m, j}(\Sigma_1, \phi)
	= 
	\Omega^Z_{j, k}(\Sigma_1, \Sigma_2 \diffActing{k} \invinv(\phi)) + \Omega_{\Sr, m}^Z(\Sigma_2, \invinv(\phi)) - \Omega^Z_{\Sr, m}(\phi).
\end{equation}
	Hence, the surface-deformation-surface cocycle~\eqref{eq:cocycle_id_mixed_3} implies that $\Omega^Z_{\Sr, m}(\phi) = 0$.
\end{proof}

\subsection{Equivalent descriptions of locality}
\label{section:local_equiv}

Generally, the sewing isomorphisms of an $\Rp$-extension can depend on the boundary parametrizations of the respective surfaces in any way --- as long as associativity holds.
Following the example of the real determinant line bundle in Section~\ref{section:cft} (to which we also return in Section~\ref{section:det_line_bundle}), it makes sense to assume that the sewing isomorphisms only depend on the boundary parametrizations that are involved in the sewing.
In other words, if a boundary parametrization that is external to the sewing operation $\Sigma_1 \sew{j}{k} \Sigma_2$ is reparametrized, say, by a diffeomorphism in $\Diffpan$ acting on $\Sigma_1$ at $l \neq j$, then the sewing isomorphism should be somehow invariant under this reparametrization.
We call such $\Rp$-extensions ``local,'' for in the case of the real determinant line bundle, this follows from the locality of the conformal anomaly formula~\eqref{eq:conformal_anomaly}.

In Proposition~\ref{prop:local_equiv}, we characterize this property in several equivalent ways, and it serves as the definition of locality.
Equation~\eqref{eq:local_surface_surface} in Proposition~\ref{prop:local_equiv} is perhaps the most instructive realization of locality, as it emphasizes the aforementioned reparametrization invariance in terms of cocycles.
Such cocycles are always relative to a trivialization $Z$, which we call ``reparametrization invariant'' if they realize locality.

\begin{proposition}
	Given a trivialization $Z$ of a \gls{fsmooth}
	$\Rp$-extension $\RMFE$, the following are equivalent:
	\begin{enumerate}
		\item
			For any surfaces $\Sigma_1$ and $\Sigma_2$ with boundary components $j$ and $k$, we have
			\label{prop:local_equiv_1}
			\begin{equation}
				\label{eq:local_surface_surface}
				\Omega^Z_{j, k}(\Sigma_1 \diffActing{l} \phi, \Sigma_2) = \Omega^Z_{j, k}(\Sigma_1, \Sigma_2), \qquad 
				\phi \in \Diffpan, \quad l \neq j.
			\end{equation}
		\item
			For any surfaces $\Sr$ and $\Sigma$ with boundary components $j$ and $k$, we have
			\label{prop:local_equiv_2}
			\begin{equation}
				\label{eq:local_surface_def}
				\Omega^Z_{\Sr, j, k}(\Sigma, \phi) = 0, \qquad \phi \in \Diffpan.
			\end{equation}
		\item
			For any surface $\Sr$ with boundary component $j$, we have
			\label{prop:local_equiv_3}
			\begin{equation}
				\label{eq:local_def_def}
				\Omega^Z_{\Sr, j}(\phi, \psi) = 0, \qquad \phi, \psi \in \Diffpan,
			\end{equation}
		\item
			For any surfaces $\Sra$ and $\Srb$ with boundary components $j$ and $k$, we have
			\label{prop:local_equiv_4}
			\begin{equation}
				\Omega^Z_{\Sra, j}(\phi, \psi) = \Omega^Z_{\Srb, k}(\phi, \psi), \qquad \phi, \psi \in \Diffpan.
			\end{equation}
		\item
			For any surfaces $\Sigma_1$ and $\Sigma_2$ with boundary components $j$ and $k$, we have
			\label{prop:local_equiv_5}
			\begin{equation}
				\label{eq:local_inversion}
				\Omega^Z_{j, k}(\Sigma_1 \diffActing{j} \phi, \Sigma_2) = \Omega^Z_{j, k}(\Sigma_1, \Sigma_2 \diffActing{k} \invinv(\phi)), \qquad \phi \in \Diffpan.
			\end{equation}
	\end{enumerate}
	If the conditions hold for a trivialization $Z$, we call it \emph{reparametrization invariant}, and if any such $Z$ exists, we call $E$ \emph{local}.
	\label{prop:local_equiv}
\end{proposition}
\begin{proof} 
We prove the implications $\ref{prop:local_equiv_1} \Longrightarrow \ref{prop:local_equiv_2} \Longrightarrow \ref{prop:local_equiv_3} \Longrightarrow \ref{prop:local_equiv_4} \Longrightarrow \ref{prop:local_equiv_1}$, and $\ref{prop:local_equiv_2} \Longleftrightarrow \ref{prop:local_equiv_5}$. 
\begin{itemize}[leftmargin=4em]
\item[\ref{prop:local_equiv_1} $\Longrightarrow$ \ref{prop:local_equiv_2}:]
By Equations~(\ref{eq:local_surface_surface},~\ref{eq:mixed_cocycle_change}), the cocycle in Item~\ref{prop:local_equiv_2} does not depend on $\Sr$ and $j$. Thus, by changing to $\Sr= \Sigma$ and $j = k$, Equation~\eqref{eq:cocycle_surfaces_equal_vanishing} implies Item~\ref{prop:local_equiv_2}.

\item[\ref{prop:local_equiv_2} $\Longrightarrow$ \ref{prop:local_equiv_3}:]
This implication follows directly from Equation~\eqref{eq:cocycle_mixed_compat_2}.

\item[\ref{prop:local_equiv_3} $\Longrightarrow$ \ref{prop:local_equiv_4}:]
If all the cocycles on $\Diffpan$ vanish, then they also agree.

\item[\ref{prop:local_equiv_4} $\Longrightarrow$ \ref{prop:local_equiv_1}:]
Assuming Item~\ref{prop:local_equiv_4}, the coboundaries in~\eqref{eq:cocycle_AB_relation} vanish, so they define homomorphisms $\phi \mapsto \Omega^Z_{l, m}(\Srx, \Sra \diffActing{j} \phi) - \Omega^Z_{l, m}(\Srx, \Sra)$ from $\Diffpan$ to the abelian additive group $\R$.
As $\Diffpan$ is a perfect group (see, e.g.,~\cite[Theorem~4.4.2]{Guieu-Roger:Virasoro_book}), any such homomorphism vanishes. It follows that $\Omega^Z_{l,m}(\Sra \diffActing{j} \phi, \Srx) - \Omega^Z_{l,m}(\Sra, \Srx) = 0$, which is precisely Equation~\eqref{eq:local_surface_surface} in Item~\ref{prop:local_equiv_1} for the surfaces $\Sigma_1 = \Sra$ and $\Sigma_2 = \Srx$.

\item[\ref{prop:local_equiv_2} $\Longleftrightarrow$ \ref{prop:local_equiv_5}:]
This follows directly from Equation~\eqref{eq:cocycle_id_mixed_3}.
\end{itemize}
Finally, observe that none of the implications above require changing the trivialization.
\end{proof}

In fact, reparametrization invariant trivializations only depend on the internal moduli of the surfaces, i.e., the moduli of the surface without boundary parametrizations.
Let us conclude with a characterization of locality emphasizing this perspective.

\begin{proposition}
	Any two globally defined reparametrization invariant trivializations $Z$ and $W$ are related by \gls{fsmooth} functions $B \in \functions{\moduligb}$ such that $Z(\Sigma) = e^{B(\Sigma)} \; W(\Sigma)$.
	Moreover, for any boundary component $j$, we have $B(\Sigma \diffActing{j} \phi) = B(\Sigma)$ for all $\phi \in \Diffpan$.
	\label{prop:reparam_indep}
\end{proposition}
\begin{proof}
	The existence of $B \in \functions{\moduligb}$ is a direct consequence of the \gls{fsmooth}ness of $Z$ and $W$.
	By Equation~\eqref{eq:local_triv_ext}, we have $Z_{\Sigma, j}(\phi) = e^{B(\Sigma \diffActing{j} \phi) - B(\Sigma)} \; W_{\Sigma, j}(\phi)$ for any reference surface $\Sigma \in \moduligb$ and diffeomorphism $\phi \in \Diffpan$.
	Setting $B(\phi) = B(\Sigma \diffActing{j} \phi) - B(\Sigma)$, Equations~(\ref{eq:def_cocycle_def_def},~\ref{eq:local_def_def}) imply that $B(\phi) + B(\psi) = B(\phi \circ \psi)$ for $\phi, \psi \in \Diffpan$.
	Since $\Diffpan$ is a perfect group (see, e.g.,~\cite[Theorem~4.4.2]{Guieu-Roger:Virasoro_book}), this yields $B(\phi) = 0$.
\end{proof}

\begin{remark}
	One may also characterize locality without referencing a reparametrization invariant trivialization, by considering the sewing isomorphisms on smaller moduli spaces $\jmoduligb{\underline{j}}$ of surfaces which carry a boundary parametrization only at some boundary components labeled by $\underline{j} = \{j_1, \dots, j_n\}$.
	In other words, $\jmoduligb{\underline{j}}$ is the quotient of $\moduligb$ by the actions of $\Diffpan$ acting at all boundary components except $\underline{j}$.

	Denoting the projection by $\projj{\underline{j}} \colon  \moduligb \to \jmoduligb{\underline{j}}$, we have $\projj{\underline{j}}(\Sigma \diffActing{k} \phi) = \projj{\underline{j}}(\Sigma)$ for all $\phi \in \Diffpan$ and $k \notin \underline{j}$.
	Now, if $Z$ is a reparametrization invariant trivialization of a local $\Rp$-extension
	$\RMFE$, then by Item~\ref{prop:local_equiv_2} in Proposition~\ref{prop:local_equiv}, we can also take the quotient of $\RMFE(\moduligb)$ by the actions of $\RMFE(\Diffpan)$ at all boundary components except $\underline{j}$, to obtain an $\Rp$-bundle $\RMFE(\jmoduligb{\underline{j}}) \to \jmoduligb{\underline{j}}$ with trivialization $Z(\projj{\underline{j}}(\Sigma)) = Z(\Sigma)$; for any reference surface $\Sigma_1$ with boundary component $l$, the actions
	\begin{equation}
		Z(\Sigma) \diffActingE{k} Z_{\Sigma_1, l}(\phi) = Z(\Sigma \diffActing{k} \phi), \qquad \phi \in \Diffpan, \quad k \notin \underline{j},
	\end{equation}
	leave the trivialization invariant. 
		Here, the quotients $\RMFE(\jmoduligb{\underline{j}})$ are independent of~$Z$ by Proposition~\ref{prop:reparam_indep}.
	Moreover, $\RMFE(\jmoduligb{\underline{j}})$ still come with well-defined sewing isomorphisms for the still parametrized boundary components, such that the following diagrams commute\footnote{
		Here, $\jmoduligb{\emptyset}$ is the classical finite-dimensional moduli space~\cite{Abikoff:The_real_analytic_theory_of_Teichmuller_space}, which is the quotient of $\moduligb$ by all reparametrizations, leaving only the enumeration of the boundary components.
	}:
	\begin{align}
	\begin{tikzcd}[column sep = 4em, cells={minimum width=14em}, ampersand replacement=\&]
		\RMFE(\moduli{\genus_1}{\boundaries_2}) \boxtimes \RMFE(\moduli{\genus_2}{\boundaries_2}) \&
		\RMFE(\moduli{\genus_1 + \genus_2}{\boundaries_1 + \boundaries_2 - 2}) \\
		\RMFE(\jmoduli{\genus_1}{\boundaries_2}{\{j\}}) \boxtimes \RMFE(\jmoduli{\genus_2}{\boundaries_2}{\{k\}}) \&
		\RMFE(\jmoduli{\genus_1 + \genus_2}{\boundaries_1 + \boundaries_2 - 2}{\emptyset})
		\arrow["{\displaystyle \sewx{\RMFE}{j}{k}}", from=1-1, to=1-2]
		\arrow[from=1-1, to=2-1]
		\arrow[from=1-2, to=2-2]
		\arrow["{\displaystyle \sewx{\RMFE}{j}{k}}", from=2-1, to=2-2, densely dashed]
	\end{tikzcd}\label{diag:localsewing} \\
	\begin{tikzcd}[column sep = 4em, cells={minimum width=14em}, ampersand replacement=\&]
		\RMFE(\moduligb)\&
		\RMFE(\moduli{\genus + 1}{\boundaries - 2}) \\
		\RMFE(\jmoduligb{\{j, k\}})\&
		\RMFE(\jmoduli{\genus + 1}{\boundaries - 2}{\emptyset}).
		\arrow["{\displaystyle \sewxself{\RMFE}{j}{k}}", from=1-1, to=1-2]
		\arrow[from=1-1, to=2-1]
		\arrow[from=1-2, to=2-2]
		\arrow["{\displaystyle \sewxself{\RMFE}{j}{k}}", from=2-1, to=2-2, densely dashed]
	\end{tikzcd}
	\label{diag:localsewing_self}
	\end{align}

	On the contrary, just the $\Rp$-bundles $\RMFE(\jmoduligb{\emptyset})$ define the $\Rp$-bundles $\RMFE(\moduligb)$ by pullback via $\projj{\emptyset}$.
	Together with just the dotted sewing isomorphisms in the diagrams~(\ref{diag:localsewing},~\ref{diag:localsewing_self}), this fully determines the $\Rp$-extension $\RMFE$. 
	If $Z$ is a trivialization of $\RMFE(\jmoduligb{\emptyset})$, then it pulls back to a reparametrization invariant trivialization of $\RMFE(\moduligb)$.
	For such a trivialization, the cocycles of Equation~\eqref{eq:local_surface_surface} agree in their defining relations:
	\begin{equation}
	\begin{aligned}
		{Z(\Sigma_1) \sewx{\RMFE}{j}{k} Z(\Sigma_2)}
		&= e^{\Omega^Z_{j,k}(\Sigma_1, \Sigma_2)} Z(\Sigma_1 \sew{j}{k} \Sigma_2),
		\label{eq:forget_rels}
		\\
		Z(\Sigma_1 \diffActing{l} \phi) \sewx{\RMFE}{j}{k} Z(\Sigma_2)
		&= e^{\Omega^Z_{j,k}(\Sigma_1 \diffActing{l} \phi, \Sigma_2)} Z((\Sigma_1 \sew{j}{k} \Sigma_2) \diffActing{l} \phi),
	\end{aligned}
		\qquad \phi \in \Diffpan,
	\end{equation}
		because the corresponding elements in the first and second lines of~\eqref{eq:forget_rels} each project to the same element in $\RMFE(\jmoduligb{\emptyset})$.
	This shows that any $\Rp$-extension
	constructed by pullback from the classical finite-dimensional moduli space $\jmoduligb{\emptyset}$ is local.
\end{remark}

\subsection{Central charge}
\label{section:central_charge}

Consider the unit disk $\Sr = \disk$ as the reference surface, and the deformation-deformation cocycle identity~\eqref{eq:def_def_def} for a trivialization $Z$ of an $\Rp$-extension $\RMFE$. 
By Item~\ref{lemma:complex_deformations:unit_disk} of Lemma~\ref{lemma:complex_deformations}, the disk is a special reference surface in the sense that the cocycle
\begin{equation}
	\Omega^Z_{\disk, 1} \colon  \defmultset \to \R
\end{equation}
is defined on any pair of composable complex deformations $\phi,\psi \in \defmultset$.
Thus, $\Omega^Z_{\disk, 1}(\phi, \psi)$ defines a cohomology class in the cohomology group $H^2(\DefC; \R)$, 
which was defined in~\cite[Section~3]{Maibach-Peltola:Complex_deformations_of_the_circle}.
This cohomology group was computed in~\cite[Figure~3.2]{Maibach-Peltola:Complex_deformations_of_the_circle}, showing that there are unique real constants $\mathbf{a}, \mathbf{b}, \charge \in \R$ such that
\begin{equation}
	[\Omega^Z_{\disk, 1}] = \mathbf{a} \; [\Re \bott] + \mathbf{b} \; [\Re \rotcocycle] + \charge \; [\Im \bott].
	\label{eq:a_b_c}
\end{equation}
in the basis of cocycles described below.
It is important to note that, because changing the trivialization $Z$ only changes the cocycle $\Omega^Z_{\disk, 1}$ by a coboundary, the constants $\mathbf{a}, \mathbf{b}, \charge$ are independent of the choice of $Z$.
\begin{definition}
	The \emph{central charge} of an $\Rp$-extension $\RMFE$ is the constant $\charge \in \R$ in Equation~\eqref{eq:a_b_c}.
	\label{def:central_charge}
\end{definition}
The central charge is also accessible at the Lie algebra level.
Let $\omega^Z_{\disk, 1}$ be the Lie algebra $2$-cocycle on $\VectC$ obtained from differentiation of $\Omega^Z_{\disk, 1}$.
By the characterization of the Lie algebra cohomology $H^2(\VectC, \R)$ in~\cite[Figure~3.3]{Maibach-Peltola:Complex_deformations_of_the_circle}, we have
\begin{equation}
	[\omega^Z_{\disk, 1}] = \mathbf{a} \; [\Re \gelfandfuks] + \charge \; [\Im \gelfandfuks].
	\label{eq:a_c_alg}
\end{equation}
The following complex-valued cocycles appear in Equations~(\ref{eq:a_b_c},~\ref{eq:a_c_alg}):
\begin{enumerate}
	\item
		$\bott$ is the Bott--Thurston cocycle~\cite{Bott:On_the_characteristic_classes_of_groups_of_diffeomorphisms}.
	\item
		$\rotcocycle$ is a cocycle combining the mean rotation and conformal radius of the complex deformations.
		The Lie algebra cocycle $\rotcocyclealg$ associated to $\rotcocycle$ is a coboundary.
	\item
		$\gelfandfuks$ is the Gel'fand--Fuks cocycle~\cite{Gelfand-Fuchs:Cohomologies_of_the_Lie_algebra_of_vector_fields_on_the_circle} given in~Equation~\eqref{eq:gelfandfuks}, also known as the cocycle which defines the Virasoro algebra as a central extension of $\VectC$.
\end{enumerate}
The relation between the cocycles is that the Lie algebra cocycle associated to $\bott$ is
\begin{equation}
	\gelfandfuks(v, w) - \rotcocyclealg(v, w) = \frac{1}{24\pi} \int_{S^1} v'(z) \; w''(z) \; \dd z, \qquad v, w \in \VectC.
	\label{eq:cocycle_difference}
\end{equation}
See \cite[Section~3.2]{Maibach-Peltola:Complex_deformations_of_the_circle} for details.

We showed in the previous work~\cite{Maibach-Peltola:From_the_conformal_anomaly_to_the_Virasoro_algebra} that the Lie algebra cocycle of the real determinant line bundle (discussed in Section~\ref{section:cft}) is the imaginary part of the Gel'fand--Fuks cocycle~\eqref{eq:gelfandfuks} times the central charge $\charge$.
This motivates Definition~\ref{def:central_charge} for the central charge of $\RMFE$ as the coefficient of the imaginary part of the Gel'fand--Fuks cocycle.

\begin{remark}\label{rem:cocycle_is_bott_nonrel}
If $\RMFE$ is local, and $Z$ is reparametrization invariant, then the cocycle $\Omega^Z_{\disk, 1}$ vanishes on diffeomorphisms by Item~\ref{prop:local_equiv_3} of Proposition~\ref{prop:local_equiv} --- thus, it is a cocycle \emph{relative to} $\Diffpan$.
The relative cohomology $H^2(\DefC; \Diffpan; \R)$ was computed in~\cite[Figure~3.2]{Maibach-Peltola:Complex_deformations_of_the_circle} as well: it only contains $[\Im \bott]$, so
\begin{equation}
	\label{eq:cocycle_is_bott_nonrel}
	\begin{aligned}
		[\Omega^Z_{\disk, 1}] &= \charge \; [\Im \bott] \in H^2(\DefC; \Diffpan; \R), \quad
		\\
		[\omega^Z_{\disk, 1}] &= \charge \; [\Im \gelfandfuks] \in H^2(\VectC; \VectR; \R).
	\end{aligned}
\end{equation}
This characterizes the cocycles up differentials of trivialization-dependent functions which vanish respectively on $\Diffpan$ and $\VectR$.
Since $\mathbf{a} = \mathbf{b} = 0$ here, we see that the assumption of locality singles out the central charge as the only coefficient.
\end{remark}

Let us also point out that Remark~\ref{rem:cocycle_is_bott_nonrel} confirms our earlier result~\cite[Theorem~1.1]{Maibach-Peltola:From_the_conformal_anomaly_to_the_Virasoro_algebra}: the real part of the Gel'fand--Fuks cocycle does not appear for the real determinant line bundle (which is local by Theorem~\ref{thm:detrc_prop}).

\begin{remark}
	There might exist non-local $\Rp$-extensions whose deformation-deformation cocycles lie in the nontrivial cohomology classes $[\Re \gelfandfuks]$ and $[\Re \rotcocycle]$ in $H^2(\DefC; \R)$.
	Hence, the locality in Theorem~\ref{thm:universality} might be necessary.
	\label{remark:transnumber}
\end{remark}

\subsection{Modular and crossing invariance properties}
\label{section:modular_invariance_properties}

In this section, we discuss several properties of local \gls{fsmooth} $\Rp$-extensions that serve as intermediate steps in the proof of Theorem~\ref{thm:universality}, also possibly interesting by their own right.
The real determinant line bundle satisfies all of these properties (Theorem~\ref{thm:detrc_prop}).
We prove these properties for a general local \gls{fsmooth}
$\Rp$-extension in Section~\ref{section:modular_invariance_proof}, by relating them to the real determinant line bundle in genus $0$.
Before the definitions, let us briefly describe each property.
\begin{enumerate}
	\item \emph{Flat modular invariance} (Definition~\ref{def:flat_modular_invariance}).
		Sewing the two boundary components of an annulus with the standard parametrization and cylindrical metric results in a torus also with a flat metric, and the seam on the torus is a geodesic.
		If the same torus is obtained from two different annuli, the geodesic seams may not be homotopic.
		Flat modular invariance asserts that certain cocycles still agree in this case.
	\item \emph{Crossing invariance} (Definition~\ref{def:crossing}).
		A hyperbolic surface in $\hmoduli{0}{4}$ with four boundary components may be decomposed into two hyperbolic pairs of pants in several ways, changing the homotopy class of the seam.
		Crossing invariance asserts that it is possible to trivialize $\RMFE(\pants)$ such that any pants decomposition along hyperbolic geodesics induces a trivialization of $\RMFE(\moduli{0}{4})$ independent of the homotopy class of the seam.
	\item \emph{Hyperbolic modular invariance} (Definition~\ref{def:modular_invariant_hyperbolic}).
		Similarly, a handle in $\moduli{1}{1}$ may be decomposed into a pair of pants along a hyperbolic geodesic in a chosen homotopy class.
		Hyperbolic modular invariance asserts the existence of a trivialization of $\RMFE(\pants)$ compatible with the self-sewing isomorphisms of hyperbolic pairs of pants.
		Together with crossing invariance, this implies invariance under S- and A-moves between hyperbolic pants decompositions, as defined by Hatcher and Thurston~\cite{Hatcher-Thurston:A_presentation_for_the_mapping_class_group_of_a_closed_orientable_surface, Hatcher:Pants_decompositions_of_surfaces}.
\end{enumerate}

Let $\RMFE$ be a local \gls{fsmooth} $\Rp$-extension with reparametrization invariant trivialization $Z$.
First, we require a more detailed characterization of the coboundaries that can appear in Equation~\eqref{eq:cocycle_is_bott_nonrel}.
Consider the cocycles~$\Omega^Z_{\A_\tau, 1}$ relative to a standard annulus $\A_\tau$ of modulus $\tau > 0$.
These are related to $\Omega^Z_{\disk, 1}$ as
\begin{equation}
\begin{aligned}
	&\Omega^Z_{\disk, 1}(\phi, \psi)
	- \Omega^Z_{\A_\tau, 1}(\phi, \psi)
	\\ = \; & \Omega^Z_{1,2}(\disk, \A_\tau \diffActing{1} \phi)
	+ \Omega^Z_{1,2}(\disk, \A_\tau \diffActing{1} \psi)
	- \Omega^Z_{1,2}(\disk, \A_\tau \diffActing{1} (\phi \circ \psi))
	- \Omega^Z_{1,2}(\disk, \A_\tau),
\end{aligned}
\label{eq:cocycles_A_D}
\end{equation}
which is obtained as a special case of Equation~\eqref{eq:cocycle_AB_relation} using the fact that ${\A_\tau \sew{2}{1} \disk = \disk}$. 
Assuming that $\Omega^Z_{\A_\tau, 1}$ vanishes on pairs of scaling transformations~\eqref{eq:scaling}, the Lie algebra cocycle $\omega^Z_{\A_\tau, 1}$ yields a cohomology class\footnote{
	Technically, $\Omega^Z_{\A_\tau, 1}$ does not lie in the cohomology group $H^2(\DefC; \R)$, because the domain of $\Omega^Z_{\A_\tau, 1}$ is the subset of composable pairs in $\defmultset$ acting on $\A_\tau$ (without making the boundary components intersect).
	Yet, the Lie algebra cocycle $\omega^Z_{\A_\tau, 1}$ is still well-defined.
} relative to the Lie algebra $\R \ell_0$ generating the scaling transformations
(by Remark~\ref{rem:cocycle_is_bott_nonrel}, we have $[\omega^Z_{\disk, 1}] = \charge \; [\Im \gelfandfuks]$):
\begin{equation}
\begin{aligned}
	[\omega^Z_{\A_\tau, 1}] = \charge \; [\Im \gelfandfuks] + \mathbf{h}_Z [\Im \rotcocyclealg] \in H^2(\VectC; \VectR, \R \ell_0; \R).
\end{aligned}
\label{eq:def_h}
\end{equation}
When changing $\tau$, the cocycle is changed by the coboundary in Equation~\eqref{eq:cocycle_AB_relation}, 
but by the invariance under scaling transformations, the constant $\mathbf{h}_Z \in \R$ does not depend on $\tau$.
(However, it does depend on the reparametrization invariant trivialization~$Z$.)

The following notion of ``cylindrical'' trivialization is named after the corresponding trivialization for the real determinant line bundle, which is given by cylindrical metrics on annuli (Theorem~\ref{thm:detrc_prop}).
It is obtained by making a trivialization scaling invariant, and forcing $\mathbf{h}_Z = 0$ in Equation~\eqref{eq:def_h}, which is convenient when considering annuli.
\begin{proposition}
	For a local \gls{fsmooth} $\Rp$-extension $\RMFE$, there exists a unique reparametrization invariant trivialization $Z$ of $\RMFE(\annuli)$ such that
	\begin{enumerate}
		\item
			$\Omega^Z_{1,2}(\A_{\tau_1}, \A_{\tau_2}) = 0$ for all $\tau_1, \tau_2 > 0$, and
							\label{prop:cylindrical_annuli}
		\item
			$[\omega_{\A_\tau, 1}^Z] = \charge \; [ \Im \gelfandfuks] \in H^2(\VectC; \VectR, \R \ell_0; \R)$ for all $\tau > 0$.
				\label{prop:cylindrical_cohomol}
	\end{enumerate}
	We call this trivialization \emph{cylindrical}.
	\label{prop:cylindrical}
\end{proposition}
	A similar choice of $\charge \neq 0$ and $\mathbf{h}_Z = 0$ also appears in the K\"ahler geometry of the homogeneous space $\Diffpan / S^1$; see~\cite{Bowick-Rajeev:The_holomorphic_geometry_of_closed_bosonic_string_theory_and_DiffS1modS1, Hong-Rajeev:Universal_Teich_space_and_DiffS1modS1, Nag-Verjovsky:DiffS1_and_Teichmuller_spaces}.
\begin{proof}
	Let $Z$ be any reparametrization invariant trivialization of $\RMFE$.
	The assignment $(\tau_1, \tau_2) \mapsto \Omega^Z_{1,2}(\A_{\tau_1}, \A_{\tau_2})$ defines a cocycle in $Z^2(\R; \R)$.
	Because the Lie group cohomology $H^2(\R; \R)$ is trivial (e.g.\ by~\cite[Figure~3.2]{Maibach-Peltola:Complex_deformations_of_the_circle}, where $\scalinggroup \cong \R$), there exists a function $f \in C^\infty(\R, \R)$ such that $\Omega^Z_{1,2}(\A_{\tau_1}, \A_{\tau_2}) = f(\tau_1) + f(\tau_2) - f(\tau_1 + \tau_2)$ for $\tau_1, \tau_2 \in \R$.
		Now, the trivialization $W(A) = e^{f(\modulus{A})} Z(A)$, where $A \in \annuli$ and $\modulus{A}$ is the modulus\footnote{We normalize the modulus so that $\modulus{\A_\tau} = \tau$.} of $A$, 
	is another reparametrization invariant trivialization of $\RMFE(\annuli)$.
	For this trivialization, we have $\Omega^W_{1,2}(\A_{\tau_1}, \A_{\tau_2}) = 0$ for all $\tau_1, \tau_2 > 0$, so the property in Item~\ref{prop:cylindrical_annuli} holds for $W$.

	In terms of scaling transformations $\scaling{\tau_1}, \scaling{\tau_2}$ with $\tau_1, \tau_2 > 0$, we have
	\begin{equation}
		\Omega^W_{\A_\tau, 1}(\scaling{\tau_1}, \scaling{\tau_2})
		= \Omega^W_{1,2}(\A_{\tau + \tau_2}, \A_{\tau_1})
		- \Omega^W_{1,2}(\A_{\tau}, \A_{\tau_1}) = 0,
	\end{equation}
	by (\ref{eq:cocycle_surfaces_equal_vanishing},~\ref{eq:cocycle_id_mixed_2},~\ref{eq:cocycle_mixed_compat_2}). 
	Thus, the cocycles $\Omega^W_{\A_\tau, 1}$ are relative to scaling transformations, defining $\mathbf{h}_W$ via~\eqref{eq:def_h}.
	Fix $\lambda \in \R$ and consider the trivialization $X(A) = e^{\lambda \; \modulus{A}} \; W(A)$ for $A \in \annuli$.
	The Lie algebra cocycles of $X$ and $W$ differ by a coboundary $\beta$, which is the Lie algebra cohomology differential of the derivative of the exponent relating the trivializations: 
	\begin{equation}
	\begin{aligned}
		\omega_{\A_\tau, 1}^X(v, w)
		= \; & \omega_{\A_\tau, 1}^W(v, w) + \beta([v, w]), \qquad v, w \in \VectC \\
		\beta(v) = \; &
		\eval{\pdv{t}}_{t = 0}
		\lambda \; \modulus{\A_{\tau} \diffActing{1} \Phi^v(t, \blank)} ,
	\end{aligned}	
	\end{equation}
	using the flow~\eqref{eq:flow_right}.
	Consider the exact sequence from~\cite[Figure~3.3]{Maibach-Peltola:Complex_deformations_of_the_circle},
	\begin{equation}
	\begin{aligned}
		H^1(\R \ell_0; \R) &\xlongrightarrow{\delta_1} H^2(\VectC; \VectR, \R \ell_0; \R) \\
		&\xlongrightarrow{\delta_2} H^2(\VectC; \VectR; \R),
	\end{aligned}
	\end{equation}
	where the first map $\delta_1$ is the transgression map obtained from the long exact sequence of relative Lie algebra cohomology, and the second map $\delta_2$ projects onto cohomology classes also up to coboundaries not vanishing on $\R \ell_0$.
	Using the fact that the flow of the vector field $\ell_0 = - z \partial_z$ generates a scaling transformation $\Phi^{\ell_0}(t, \blank) = \scaling{\frac{t}{2\pi}}$, we may compute
	\begin{equation}
		\beta(\ell_0) =
		\eval{\pdv{t}}_{t = 0}
		\lambda \; \modulus{\A_{\tau + \frac{t}{2\pi}}} =
		\frac{\lambda}{2\pi}.
		\label{eq:modulus_derivative}
	\end{equation}
	It follows that the Lie algebra cohomology differential $\beta([v, w])$ is a nontrivial cocycle relative to $\R \ell_0$, even though it looks like a coboundary.
	The non-relative cohomology class $\delta_2([\beta([v, w])])$, however, is trivial. 
	Hence, by exactness, we find a preimage under $\delta_1$, which is just $\delta_1([\beta]) = [\beta([v, w])]$.
	By~\cite[Figure~3.3]{Maibach-Peltola:Complex_deformations_of_the_circle}, we have $H^1(\R \ell_0; \R) = \R \Im \rotnumberalg$, where $\rotnumberalg$ is the derivative of the mean rotation~\cite[Section~3.2]{Maibach-Peltola:Complex_deformations_of_the_circle} satisfying $\Im \rotnumberalg(\ell_0) = \frac{1}{24}$.
	Therefore,  $[\beta([v, w])] = \frac{12 \lambda}{\pi}[\Im \rotcocyclealg]$,
	giving the relation $\mathbf{h}_X = \mathbf{h}_W + \frac{12 \lambda}{\pi}$.
	We conclude that by setting $\lambda = - \frac{\pi}{12}\mathbf{h}_W$, the Lie algebra cocycle of $X$ satisfies the property in Item~\ref{prop:cylindrical_cohomol}.

	Finally, to prove uniqueness, let $Y$ be another cylindrical trivialization of $\RMFE$.
	By Proposition~\ref{prop:reparam_indep}, it is related to $X$ by some smooth function $f \colon  \R \to \R$ such that $Y(A) = e^{f(\modulus{A})} X(A)$ for $A \in \annuli$.
	By Item~\ref{prop:cylindrical_annuli}, $f$ is additive, so there exists some $\lambda > 0$ such that $f(\tau) = \lambda \; \tau$ for $\tau > 0$.
	By the above, we have the relation $\mathbf{h}_Y = \mathbf{h}_X + \frac{12 \lambda}{\pi}$.
	Hence, we see that $\lambda = 0$, concluding that $Y = X$ on $\annuli$.
\end{proof}

Sewing the two boundary components of an annulus $A \in \annuli$ results in a torus $T = \sewself{1}{2} A \in \tori$.
Up to scale, both the annulus and the torus come with a unique conformal flat metric with geodesic boundary.
Since $A$ conformally maps into $T$, we may ask whether the seam inside~$T$ defined by the parametrizations of $A$ is a geodesic with respect to the flat metric on~$T$.
For annuli creating equivalent tori and having the geodesic seam, we define the following notion, assuming that the sewing isomorphisms agree on the cylindrical metrics.
Note that the interesting case is where the seams are not homotopic.
\begin{definition}
	A local \gls{fsmooth} $\Rp$-extension $\RMFE$ is \emph{flatly modular invariant} if for all annuli $A, B \in \annuli$ such that $\sewself{1}{2} A = \sewself{1}{2} B \in \tori$ and both seams in the torus are geodesics in a conformal flat metric, we have
	\begin{equation}
		\sewxself{\RMFE}{1}{2} Z(A) = \sewxself{\RMFE}{1}{2} Z(B),
		\label{eq:flat_modular_invariance}
	\end{equation}
	where $Z$ is the cylindrical trivialization 	(Proposition~\ref{prop:cylindrical}).
	\label{def:flat_modular_invariance}
\end{definition}

\begin{remark}
	By reparametrization invariance of the cylindrical trivialization, the relation~\eqref{eq:flat_modular_invariance} also holds if the seam is a reparametrization of a geodesic.
	Basic examples of annuli with geodesic seam in the torus are the standard annuli $\A_\tau$ for $\tau > 0$, and $\smash{\A_\tau \diffActing{1} \rotation_\alpha}$, where one boundary parametrization is twisted by a rotation $\rotation_\alpha(z) = e^{\ii \alpha} z$ for $\alpha \in \R$. Any other annulus with this geodesic seam property is of the form $\A_\tau \diffActing{1} \rotation_\alpha \diffActing{1} \phi \diffActing{2} \invinv(\phi)$
	for some diffeomorphism $\phi \in \Diffpan$, which does not change the torus by Equation~\eqref{eq:act_between}.
	Given a torus and a homotopy class, there are unique $\tau > 0$ and $\alpha \in [0, 2\pi)$ such that the torus is isomorphic to $\sewself{1}{2} (\A_\tau \diffActing{1} \rotation_\alpha)$ with the seam belonging to the given homotopy class.
	\label{remark:geodesic_annuli}
\end{remark}

\begin{definition}
	A reparametrization invariant trivialization $Z$ of a \gls{fsmooth} $\Rp$-extension $\RMFE$ is \emph{crossing invariant} if for all hyperbolic pairs of pants $P_1, P_2, P_3, P_4 \in \hpants$ such that $P_1 \sew{1}{1} P_2 = P_3 \sew{1}{1} P_4 \in \moduli{0}{4}$, and such that the hyperbolic lengths of $\partial_1 P_1$ and $\partial_1 P_2$ (resp. $\partial_1 P_3$ and $\partial_1 P_4$) agree, we have
	\begin{equation}
		Z(P_1) \sewx{\RMFE}{1}{1} Z(P_2) = Z(P_3) \sewx{\RMFE}{1}{1} Z(P_4).
		\label{eq:crossing_identity}
	\end{equation}
	If any such $Z$ exists, we call $\RMFE$ \emph{crossing invariant}.
	\label{def:crossing}
\end{definition}

\begin{proposition}
	Any two crossing invariant trivializations of agree on $\pants$ up to a real multiplicative constant factor.
\end{proposition}
\begin{proof}
	Let $Z$ be a crossing invariant trivialization of $\RMFE$.
	By Proposition~\ref{prop:reparam_indep}, any other crossing invariant trivialization is of the form $e^{f(L_1, L_2, L_3)} Z(P)$, where $f \in C^\infty(\Rp^3, \R)$ and $L_1, L_2, L_3 \in \Rp$ are the hyperbolic lengths of the boundary components of $P \in \pants$.
	Consider the sewing operation $(P \diffActing{3} \rotation_\alpha) \sew{3}{3} P$, twisted by a rotation $\rotation_\alpha$ with angle $\alpha \in \R$.
	It yields a surface with boundary lengths $L_1, L_1, L_2, L_2$, which may also be obtained by sewing pants with boundary lengths $L_1, L_1, L_4(\alpha)$ and $L_2, L_2, L_4(\alpha)$, where $L_4(\alpha)$ is the length of the hyperbolic geodesic cutting $(P \diffActing{3} \rotation_\alpha) \sew{3}{3} P$ so that the boundary components of equal length are in the same connected component.
	Equation~\eqref{eq:crossing_identity} yields the identity
	\begin{equation}
		2 f(L_1, L_2, L_3) = f(L_1, L_1, L_4(\alpha)) + f(L_2, L_2, L_4(\alpha)).
		\label{eq:crossing_alpha}
	\end{equation}
	The length $L_4(\alpha)$ changes as a function of the angle $\alpha \in \R$, while the lengths $L_1, L_2, L_3$ remain constant.
	Hence, we conclude that the right-hand side of Equation~\eqref{eq:crossing_alpha} is independent of $L_4(\alpha)$, and thus, $f(L_1, L_2, L_3)$ is independent of $L_3$.
	By applying the same argument sewing at $L_1$ and $L_2$, we find that $f$, and thus $e^{f(L_1, L_2, L_3)} \in \R$, is constant.
\end{proof}

\begin{definition}
	A crossing invariant trivialization $Z$ of a \gls{fsmooth} $\Rp$-extension $\RMFE$ is \emph{hyperbolically modular invariant} if for all hyperbolic pairs of pants $P_1, P_2 \in \hpants$ such that $\sewself{j}{k} P_1 = \sewself{l}{m} P_2 \in \moduli{1}{1}$, and such that the hyperbolic lengths of $\partial_1 P_1$ and $\partial_1 P_2$ agree, we have
	\begin{equation}
		\sewxself{\RMFE}{j}{k} Z(P_1) = \sewxself{\RMFE}{l}{m} Z(P_2).
	\end{equation}
	If any such $Z$ exists, we call $\RMFE$ \emph{hyperbolically modular invariant}.
	\label{def:modular_invariant_hyperbolic}
\end{definition}

\subsection{Example: The real determinant line bundle}
\label{section:det_line_bundle}

Let us now return to the guiding example of the real determinant line bundle $\Detrc$, discussed in Section~\ref{section:cft}.
We will show in Theorem~\ref{thm:detrc_prop} below that it satisfies all the properties from Section~\ref{section:modular_invariance_properties}. 
To begin, consider the trivialization
\begin{equation}
	Z(\Sigma) = [\gcc(\Sigma)] \in \Detrc(\Sigma),
	\label{eq:detrc_triv}
\end{equation}
defined by metrics $\gcc(\Sigma)$ of constant Gaussian curvature $+1$, $-1$, or $0$ such that the boundary components are geodesics~\cite{OPS:Extremals_of_determinants_of_Laplacians}. 
The metrics are unique, except for the flat case (tori and annuli). By compactness, we may additionally require unit volume to fix unique flat metrics\footnote{
In fact, the conformal anomaly of a constant conformal factor $\sigma$ is $\frac{\sigma}{6} \chi(\Sigma)$, by the Gauss--Bonnet formula (here, $\chi(\Sigma)$ is the Euler characteristics of the surface $\Sigma$). Thus, for the $\chi(\Sigma) = 0$ case, the trivialization $Z(\Sigma) = [\gcc(\Sigma)] \in \Detrc(\Sigma)$ is independent of the global scale of the metric $\gcc(\Sigma)$.
}.
This trivialization is the canonical example for the modular and crossing invariance properties above, as will be evident from the proof of Theorem~\ref{thm:detrc_prop}.
Other trivializations of $\Detrc$ are reviewed in~\cite[Appendix~B]{Luo-Maibach:Two-loop_Loewner_potentials}.

\begin{theorem}
Fix $\charge \in \R$.
	The real determinant line bundle $\Detrc$ is a {local}, {flatly modular invariant}, {crossing invariant}, and hyperbolically modular invariant \gls{fsmooth} $\Rp$-extension of central charge $\charge$.
	\label{thm:detrc_prop}
\end{theorem}

\begin{proof}
	By~\cite[Corollary~1.5]{Maibach-Peltola:From_the_conformal_anomaly_to_the_Virasoro_algebra}, the real determinant line bundle satisfies the defining properties of an $\Rp$-extension (Definition~\ref{def:modular_functor}).
	We can define \gls{fsmooth}ness of $\Detrc$ with respect to the constant curvature metric trivialization~\eqref{eq:detrc_triv}.

	Consider a cocycle $\Omega_{j, k}^{Z}(\Sigma_1 \diffActing{l} \phi, \Sigma_2)$ as in Item~\ref{prop:local_equiv_1} of Proposition~\ref{prop:local_equiv}.
	It is given by comparing $\big[\gcc\big((\Sigma_1 \diffActing{l} \phi) \sew{j}{k} \Sigma_2\big)\big]$ to $\big[\gcc(\Sigma_1 \diffActing{l} \phi) \cup \gcc(\Sigma_2)\big]$ in the real determinant line of the sewn surface.
	Because $\gcc(\Sigma)$, and also the conformal anomaly formula~\eqref{eq:conformal_anomaly}, do not depend on the boundary parametrizations of $\Sigma$, this is equivalent to comparing $[\gcc(\Sigma_1 \sew{j}{k} \Sigma_2)]$ to $[\gcc(\Sigma_1) \cup \gcc(\Sigma_2)]$, which yields the cocycle $\Omega_{j, k}^{Z}(\Sigma_1, \Sigma_2)$.
	Hence, $Z$ is reparametrization invariant by Item~\ref{prop:local_equiv_1} of Proposition~\ref{prop:local_equiv}, and consequently, $\Detrc$ is local.
	
	The constant curvature metric trivialization is also the cylindrical trivialization on $\Detrc(\annuli)$.
	(This follows from~\cite[Proposition~3.5]{Maibach-Peltola:Complex_deformations_of_the_circle}, which uses the computation of the Lie algebra cocycle in~\cite{Maibach-Peltola:From_the_conformal_anomaly_to_the_Virasoro_algebra}.)
	Consider now an annulus $A \in \annuli$ such that the seam in the torus $T = \sewself{1}{2} A$ is a geodesic in the metric $\gcc(T)$.
	This metric also induces a (unit volume) flat metric on $A$ for which the boundary components are geodesic, which by uniqueness agrees with $\gcc(A)$.
	Thus, $\Detrc$ is flatly modular invariant (Definition~\ref{def:flat_modular_invariance}). 
	Crossing invariance (Definition~\ref{def:crossing}) and hyperbolic modular invariance (Definition~\ref{def:modular_invariant_hyperbolic}) for $\Detrc$ are immediate consequences of the choice of the uniqueness of the hyperbolic metrics for the trivializations (indeed, they are preserved under sewing and self-sewing of hyperbolic pairs of pants, when the boundary components have matching length).
\end{proof}

\section{Disk-disk cocycles and the universal Liouville action}
\label{section:disk_disk}
In this section, we identify the cocycle of sewing two disks with the universal Liouville action --- or loop Loewner energy --- Equation~\eqref{eq:lenergy}, up to a multiplicative constant proportional to the central charge, and an additive constant depending on the trivialization (Theorem~\ref{thm:lpot}).
The disk-disk cocycle is a function $\Omega^Z_{1,1} \colon \disks \times \disks \to \R$ defined by a local \gls{fsmooth} $\Rp$-extension $\RMFE$ together with a reparametrization invariant trivialization $Z$ over $\spheres$ and $\disks$.

First, we identify a pair of unit disks as a critical point of $\Omega^Z_{1,1}$ (Section~\ref{section:diskdisk_critical}).
Then, we use conformal welding to represent the variation of $\Omega^Z_{1,1}$  as a pairing of a vector field and a quadratic differential through the Cauchy--Hilbert transform (Section~\ref{section:diskdisk_variational}, see also~\cite{Pommerenke:Univalent_functions, Duren:Univalent_functions, Morimoto:An_introduction_to_Satos_hyperfunctions}).
The key step is to prove that the quadratic differential is proportional to the Schwarzian derivative of a univalent function in the conformal welding decomposition, which implies that the Cauchy--Hilbert transform agrees with the variational formula of loop Loewner energy (Section~\ref{section:diskdisk_schwarzian}).
We conclude the proof of Theorem~\ref{thm:lpot} by identifying the proportionality constant via the central charge (Section~\ref{section:diskdisk_charge}):
We relate the second variation of $\Omega^Z_{1,1}$ to the Lie algebra cocycle of the central extension $\RMFE(\DefC)$ discussed in Section~\ref{section:central_charge}.

\subsection{Critical points of the disk-disk cocycle}
\label{section:diskdisk_critical}

Any element of $\disks$ may be represented by the unit disk with reparametrized boundary (Item~\ref{lemma:complex_deformations:diff_disk} of Lemma~\ref{lemma:complex_deformations}), so by locality, the disk-disk cocycle is fully characterized by $\Omega^Z_{1,1}(\disk \diffActing{1} \phi, \disk)$ as a function of $\phi \in \Diffpan$.
Moreover, by the characterization of locality in Item~\ref{prop:local_equiv_5} of Proposition~\ref{prop:local_equiv}, this function enjoys the symmetry
\begin{equation}
	\label{eq:lpot_inversion}
	\Omega^Z_{1,1}\Big(\disk \diffActing{1} \phi, \disk\Big) = \Omega^Z_{1,1}\Big(\disk \diffActing{1} \invinv(\phi), \disk\Big), \qquad \phi \in \Diffpan,
\end{equation}
where $\invinv(\phi) = \inversion \circ \phi^{-1} \circ \inversion$ is the involution~\eqref{eq:def_invinv} and $\inversion(z) = z^{-1}$ the inversion~\eqref{eq:inversion}. 
We use this to find that any M\"obius transformation $\phi \in \psl \subset \Diffpan$ is a critical point.
\begin{lemma}
	\label{lemma:crit_pt}
	For any \gls{fsmooth} curve $\gamma \in \curves{\DefC}$ rooted at $\gamma(0) \in \psl$, we~have
	\begin{equation}
		\eval{\pdv{t}}_{t = 0}
		\Omega^Z_{1,1}\Big(\disk \diffActing{1} \gamma(t, \blank), \disk\Big) = 0.
		\label{eq:crit_pt}
	\end{equation}
\end{lemma}
\begin{proof}
	By Item~\ref{lemma:complex_deformations:disk_invariant} of Lemma~\ref{lemma:complex_deformations}, we have $\disk \diffActing{1} \phi = \disk$ for $\phi \in \psl$.
	Thus, we may assume that $\gamma$ is rooted at $\id$. 
	We give a proof in terms of a vector field $v \in \VectC$, for which the curve becomes $\gamma = \Phi^{v}$ as in Equation~\eqref{eq:flow_right}.
	Denote the left-hand side of Equation~\eqref{eq:crit_pt} by $\lpotvect(v)$.
	First, for $v = \ell_{-n}$ with $n > -1$, the result follows from Item~\ref{lemma:complex_deformations:disk_invariant} of Lemma~\ref{lemma:complex_deformations}.
	For the remaining $\ell_{n}$ with $n > 1$, we begin with the tangential vector field $\ellpar{n} = \frac{\ell_n - \ell_{-n}}{2}$ (see also~\cite[Section~2.1]{Maibach-Peltola:Complex_deformations_of_the_circle}).
	By $\R$-linearity and the case above, we have
	\begin{equation}
		\lpotvect(\ellpar{n})
		=
		\frac{1}{2} \lpotvect(\ell_{n}), \qquad
		\lpotvect(\inversion^* \ellpar{n})
		=
		\frac{1}{2} \lpotvect(\ell_{n}).
	\end{equation}
	By the symmetry~\eqref{eq:lpot_inversion}, we also have 
	\begin{equation}	
	\lpotvect(\ellpar{n}) = - \lpotvect(\inversion^* \ellpar{n}). 
	\end{equation}	
	Both equalities can only hold if $\lpotvect(\ell_{-n}) = 0$.

	As $\lpotvect$ is only $\R$-linear, we have to treat the case of the vector fields $\ii \ell_{\pm n}$ separately.
	However, the argument for $\ell_{n}$ does not work for $\ii \ell_{n}$, since the tangential vector $\ellpari{n}$ does not have a sign difference between $\ii \ell_n$ and $\ii \ell_{-n}$.
	Conjugating $\ell_n$ by the rotation $\rotation_{\alpha}$ yields
	\begin{equation}
		\label{eq:rotation_i}
		\rotation_\alpha^* \ell_n = -(e^{\ii \alpha} z)^{n+1} e^{-\ii \alpha}\partial_z = e^{\ii n \alpha} \ell_n.
	\end{equation}
	Thus, for $n \neq 0$ and $\alpha_n = \frac{\pi}{2n}$, we have $\rotation_{\alpha_n}^* \ell_n = \ii \ell_n$.
	Considering the identities
	\begin{equation}
		\inversion \circ \rotation_\alpha \circ \inversion = \rotation_{-\alpha}, \qquad
		\rotation_\alpha^{-1} = \rotation_{-\alpha}, \qquad
		\disk \diffActing{1} \rotation_\alpha = \disk \in \disks,
	\end{equation}
	and the reparametrization invariance~\eqref{eq:local_inversion} applied to $\Omega^Z_{1,1}$, the function $\Omega^Z_{1,1}(\disk \diffActing{1} \phi, \disk)$ is invariant under both pre- and postcomposition by rotations at any $\phi \in \DefC$:
	\begin{equation}
	\begin{aligned}
		\Omega^Z_{1,1}\Big(\disk \diffActing{1} (\rotation_\alpha \circ \phi \circ \rotation_{\beta}), \disk\Big)
		= \; & \Omega^Z_{1,1}\Big(\disk \diffActing{1} (\rotation_\alpha \circ \phi), \disk \diffActing{1} \invinv(\rotation_{\beta})\Big) \\
		= \; & \Omega^Z_{1,1}\Big(\disk \diffActing{1} \phi, \disk\Big) , \qquad \alpha, \beta \in [0, 2\pi).
	\end{aligned}
	\end{equation}
	Returning to the variation, we now find that
	\begin{equation}
		\lpotvect(\ii \ell_{n})
		= \lpotvect(\rotation_{\alpha_n}^* \ell_{n})
		=
		\eval{\pdv{t}}_{t = 0}
		\Omega^Z_{1,1}\Big(\disk \diffActing{1}
		\rotation_{\alpha_n} \circ \Phi_{\ell_{n}} \circ \rotation_{-\alpha_n},
		\disk\Big)
		= \lpotvect(\ell_{n}) = 0,
	\end{equation}
	and analogously for $\ii \ell_{-n}$. This finishes the proof of~\eqref{eq:crit_pt}. 
\end{proof}

\subsection{The variational formula}
\label{section:diskdisk_variational}

The relation between the diffeomorphism $\phi \in \Diffpan$ and the analytic loop $\gamma_\phi$ in Theorem~\ref{thm:lpot} is given by conformal welding.
For real-analytic boundary behavior of the univalent functions and $\Diffpan$, we provide a rather elementary proof using the Riemann mapping theorem --- see~\cite{Takhtajan-Teo:Weil-Petersson_metric_on_the_universal_Teichmuller_space} for the more general statement using quasiconformal boundary behavior.
In our formulation of conformal welding (Proposition~\ref{prop:conformal_welding}), the Riemann uniformizing mappings of the left and right sides of the loop are interpreted as univalent functions, which are also complex deformations (this is the set $\univalent$ defined in Section~\ref{section:previous}).
We also prove \gls{fsmooth}ness properties of conformal welding.

\begin{proposition}[Conformal welding for analytic diffeomorphisms]
    \label{prop:conformal_welding}
	With the following normalizations, either the diffeomorphism or one of the univalent functions, all related through composition of complex deformations as
	\begin{equation}
		\label{eq:conformal_welding}
		(\inversion \circ \zeta^- \circ \inversion)	\circ \phi = \zeta^+, \qquad \phi \in \Diffpan, \quad \zeta^+, \zeta^- \in \univalent,
	\end{equation}
	determines the other two:
	\begin{enumerate}
		\item
			Given $\phi \in \Diffpan$ and $a \in \C \setminus \{0\}$, there exist unique $\zeta^+, \zeta^- \in \univalent$ such that~\eqref{eq:conformal_welding} holds and
			\begin{equation}
				\zeta^+(S^1) = (\inversion \circ \zeta^- \circ \inversion)(S^1), \quad \zeta^+(0) = 0, \quad (\zeta^+)'(0) = a, \quad \zeta^-(0) = 0.
			\end{equation}
			Moreover, the map $\Diffpan \to \univalent \times \univalent$ sending $\phi \mapsto (\zeta^+, \zeta^-)$ is \gls{fsmooth}.
			\label{prop:conformal_welding_item1}
		 \item
			Given $\phi \in \Diffpan$ and $a \in \C \setminus \{0\}$, there exist unique $\zeta^+, \zeta^- \in \univalent$ such that~\eqref{eq:conformal_welding} holds and
			\begin{equation}
				\zeta^+(S^1) = (\inversion \circ \zeta^- \circ \inversion)(S^1), \quad \zeta^+(0) = 0, \quad \zeta^-(0) = 0, \quad (\zeta^-)'(0) = a.
			\end{equation}
			Moreover, the map $\Diffpan \to \univalent \times \univalent$ sending $\phi \mapsto (\zeta^+, \zeta^-)$ is \gls{fsmooth}.
			\label{prop:conformal_welding_item2}
		\item
			Given $\zeta^+ \in \univalent$, there exist unique $\phi \in \Diffpan$ and $\zeta^- \in \univalent$ such that~\eqref{eq:conformal_welding} holds and $(\zeta^-)'(0) > 0$.
			Moreover, the map $\univalent \to \Diffpan$ sending $\zeta^+ \mapsto \phi$ is \gls{fsmooth}.
			\label{prop:conformal_welding_item3}
		\item
			Given $\zeta^- \in \univalent$, there exist unique $\phi \in \Diffpan$ and $\zeta^+ \in \univalent$ such that~\eqref{eq:conformal_welding} holds and $(\zeta^+)'(0) > 0$.
			Moreover, the map $\univalent \to \Diffpan$ sending $\zeta^- \mapsto \phi$ is \gls{fsmooth}.
			\label{prop:conformal_welding_item4}
	\end{enumerate}
\end{proposition}
\noindent
See also the inner part of Figure~\ref{fig:univalent_composition}.
\begin{proof}
	Given $\phi \in \Diffpan$, consider the disks $\disk \diffActing{1} \invinv(\phi) = \param{\cdisk, \phi^{-1} \circ \inversion}$ and $\disk = \param{\cdisk, \inversion}$ with analytically parametrized boundaries.
	The sphere ${(\disk \diffActing{1} \invinv(\phi)) \sew{1}{1} \disk \in \spheres}$, obtained by sewing these two disks, is isomorphic to the Riemann sphere by an isomorphism which is unique only up to M\"obius transformations.
	Given a choice of isomorphism, it restricts to the two embeddings $F_1$ and $F_2$	of the closed unit disk into $\hat \C$.
	Then, the normalization may be changed by post-composition with a M\"obius transformation $F \in \mob$:
	\begin{equation}
	\begin{tikzcd}[column sep = small, row sep = tiny, ampersand replacement=\&]
		{S^1} \&\& \disk \\
		\&\&\&\& {\hat{\C}} \&\& {\hat{\C}} \\
		{S^1} \&\& \disk
		\arrow["{\phi^{-1} \circ \inversion}", from=1-1, to=1-3]
		\arrow["\inversion"', from=1-1, to=3-1]
		\arrow["{F_1}", from=1-3, to=2-5]
		\arrow["{\zeta^+}", curve={height=-12pt}, from=1-3, to=2-7]
		\arrow["{\inversion \circ \phi}"', from=1-3, to=3-3]
		\arrow["F", from=2-5, to=2-7]
		\arrow["\inversion"', from=3-1, to=3-3]
		\arrow["{F_2}"', from=3-3, to=2-5]
		\arrow["{\inversion \circ \zeta^-}"', curve={height=12pt}, from=3-3, to=2-7]
	\end{tikzcd}
	\end{equation}
	Since the identification of the boundary components of the two disks is given by $\inversion \circ \phi$, we find that, up to normalization, the compositions $F \circ F_1$ and $\inversion \circ F_2 \circ F$ respectively become the univalent functions $\zeta^+$ and $\zeta^-$ in Equation~\eqref{eq:conformal_welding}.
	Thus, we let $F$ be the unique M\"obius transformation such that $\zeta^+(0) = F(F_1(0)) = 0$, and $(\zeta^+)'(0) = F'(F_1(0)) F_1'(0) = a$, and $\inversion(\zeta^-(0)) = F(F_2(0)) = \infty$, for Item~\ref{prop:conformal_welding_item1} and analogously for Item~\ref{prop:conformal_welding_item2}.

	Given $\zeta^+ \in \univalent$, let $F$ be the Riemann mapping from $\cdisk$ to the complement of $\zeta^+(\cdisk)$ such that $F(0) = \infty$ and $F'(0) < 0$. Define $\zeta^- = \inversion \circ F$.
	As $\zeta^-(0) = \inversion(\infty) = 0$ and $(\zeta^-)'(0) = -F'(0) / F(0)^2 > 0$, the definition matches the normalization, and the composition $(\inversion \circ \zeta^- \circ \inversion)^{-1} \circ \zeta^+ = \inversion \circ F^{-1} \circ \zeta^+$ indeed restricts to an orientation-preserving diffeomorphism $\phi \in \Diffpan$.
	This implies Item~\ref{prop:conformal_welding_item3}. 
	Item~\ref{prop:conformal_welding_item4} can be proven analogously.

	For the \gls{fsmooth}ness, note that a one-parameter family of Riemann maps with the same normalization associated to a smooth one-parameter family of smooth parametrizations of a boundary curve depends smoothly on the parameter~\cite[Theorem~28.1]{Bell:Cauchy_transform_potential_theory_and_conformal_mapping}.
\end{proof}

\noindent 
\emph{Proof of Theorem~\ref{thm:lpot}.}
Now we begin with the proof of Theorem~\ref{thm:lpot}, computing the variation of the disk-disk cocycle at a general point, which comprises the rest of this section.
Let $\phi^+, \phi^- \in \Diffpan$ be two diffeomorphisms related by $\phi^- = \invinv(\phi^+)$.
Then, any tangent vector $[\gamma^+]_\sim \in T_{\phi^+} \Diffpan$ represented by a curve $\gamma^+ \in \curves{\DiffC}$ rooted at $\gamma^+(0, \blank) = \phi^+$ is related to a tangent vector $[\gamma^-]_\sim \in T_{\phi^-} \Diffpan$ rooted at $\gamma^-(0, \blank) = \phi^-$ by $\gamma^- = \invinv(\gamma^+)$.
This setup simplifies working with the symmetry~\eqref{eq:lpot_inversion}, which implies that
\begin{equation}
	\Omega^Z_{1,1}\Big(\disk \diffActing{1} \gamma^+(t, \blank), \disk\Big) =
	\Omega^Z_{1,1}\Big(\disk \diffActing{1} \gamma^-(t, \blank), \disk\Big).
	\label{eq:ddsymmetry}
\end{equation}
By applying Item~\ref{prop:conformal_welding_item2} of Proposition~\ref{prop:conformal_welding} for any $|a| > 4$, normalizing $(\zeta^-)'(0) = a$, we find conformal welding decompositions of the diffeomorphisms $\phi^+, \phi^-$ and the curves $\gamma^+, \gamma^-$:
\begin{equation}
	\label{eq:disk_disk_welding}
	(\inversion \circ \zeta^\mp \circ \inversion) \circ \phi^\pm = \zeta^\pm, \qquad
	\big(\inversion \circ \eta^\mp(t, \blank) \circ \inversion\big) \circ \gamma^\pm(t, \blank) = \eta^\pm(t, \blank),
\end{equation}
where $\zeta^+, \zeta^- \in \univalent$, and $\eta^+, \eta^- \in \curves{\univalent}$ are such that $\eta^+(0, \blank) = \zeta^+$ and $\eta^-(0, \blank) = \zeta^-$.
The normalization is such that the by the Koebe $1/4$ theorem, the image of $\zeta^-$ contains the unit disk, and thus, $\zeta^+(S^1) = (\inversion \circ \zeta^- \circ \inversion)(S^1)$ is a subset of the open unit disk --- see Figure~\ref{fig:univalent_composition} (this becomes relevant later in the proof).
By Item~\ref{lemma:complex_deformations:disk_invariant} of Lemma~\ref{lemma:complex_deformations}, we can substitute $\disk = \disk \diffActing{1} (\inversion \circ \eta^\mp \circ \inversion)$ into the left-hand side of Equation~\eqref{eq:disk_disk_welding}, resulting in
\begin{equation}
	\label{eq:cocycle_zeta_1}
	\Omega^Z_{1,1}\Big(\disk \diffActing{1} \gamma^\pm(t, \blank), \disk\Big) = \Omega^Z_{1,1}\Big(\disk \diffActing{1} \eta^\pm(t, \blank), \disk\Big).
\end{equation}
Let $v^+, v^- \in \VectC$ be the vector fields such that 
\begin{equation}
[\eta^\pm(t, \blank)]_\sim = [\Phi^{v^\pm}(t, \blank) \circ \zeta^\pm]_\sim \in T_{\zeta^\pm} \univalent,
\end{equation}
in terms of the left-trivializing flow; cf.~Equation~\eqref{eq:flow_right} and~\cite[Remark~2.7]{Maibach-Peltola:Complex_deformations_of_the_circle}.
We use the special case of the variational formula for the universal Liouville action by Takhtajan~\&~Teo~\cite{Takhtajan-Teo:Weil-Petersson_metric_on_the_universal_Teichmuller_space} in terms of the vector fields $v^\pm$ given in~\cite[Theorem~5.1]{VPWW:Epstein_curves_and_holography-of_the_Schwarzian_action}.
Over the corresponding one-parameter family of simple analytic loops $\zeta^\pm(t, S^1) \subset \hat \C$, the variation is given by pairing the \emph{Schwarzian derivative} $\schwarzian{f}(z) = \big(f''(z)/f'(z)\big)^2 - \frac{1}{2} \big(f''(z)/f'(z)\big)^2$ as a quadratic differential with the vector field $v^\pm$:
\begin{equation}
	\eval{\pdv{t}}_{t = 0}
	\lenergy{}\Big(\zeta^\pm(t, S^1)\Big) = -\frac{2}{\pi} \Im \int_{S^1} v^\pm \; \schwarzian{{\invinv(\zeta^\mp)}}.
	\label{eq:lenergy_variation}
\end{equation}
We proceed to relate this formula to the derivative of Equation~\eqref{eq:cocycle_zeta_1}, which in terms of $v^\pm$ becomes
\begin{equation}
	\label{eq:defthetav}
	\eval{\pdv{t}}_{t = 0}
	\Omega^Z_{1,1}\Big(\disk \diffActing{1} \gamma^\pm, \disk\Big)
	=
	\eval{\pdv{t}}_{t = 0}
	\Omega^Z_{1,1}\Big((\disk \diffActing{1} \Phi^{v^\pm}) \diffActing{1} \zeta^\pm, \disk\Big).
\end{equation}
Consider the right-hand side of Equation~\eqref{eq:defthetav} for general vector fields $v \in \VectC$ replacing $v^\pm$.
If $v$ is holomorphic on a neighborhood of $\hat{\C} \setminus \disk$, we have $\disk \diffActing{1} \Phi^{v} = \disk$ by Item~\ref{lemma:complex_deformations:disk_invariant} of Lemma~\ref{lemma:complex_deformations}.
Thus, these variations only see the positive modes of $v$, and they vanish for $v \in \mobvect$.
By \gls{fsmooth}ness of both the cocycle $\Omega^Z_{1,1}$ and the action of $\DefC$ on $\disks$, we find that Equation~\eqref{eq:defthetav} is a \gls{fsmooth} function of $v$.
We now use the following result to identify unique integral representations of $\eval{\pdv{t}}_{t = 0} \Omega^Z_{1,1}\big((\disk \diffActing{1} \Phi^v) \diffActing{1} \zeta^\pm, \disk\big)$.

\begin{proposition}
	\label{prop:integral_representation}
	Given any \gls{fsmooth} $\R$-linear functional $F \in (\VectC)^\vee$ vanishing on $\mobvect = \vspan_{\C}\{\ell_{-1}, \ell_0, \ell_{1}\}$, there exist unique holomorphic quadratic differentials $Q^+$ on $\odisk$ and $Q^-$ on $\jodisk$ such that $Q^-$ vanishes at $\infty$, and
	\begin{equation}
		\label{eq:integral_repr}
		F(v) = \Re \Bigg( \int_{(1-\varepsilon)S^1}  v \; Q^+ + \int_{(1+\varepsilon)S^1} v \; Q^- \Bigg),
	\end{equation}
	where $\varepsilon > 0$ must be chosen depending on $v$ so that the integral exists (yet, the value of $F(v)$ in~\eqref{eq:integral_repr} is independent of $\varepsilon$).
\end{proposition}
\begin{proof}
	We will represent \gls{fsmooth} $\R$-linear functions $F \colon \VectC \to \R$ by integrals using the Cauchy--Hilbert transform (see~\cite[Definition~2.1.7]{Morimoto:An_introduction_to_Satos_hyperfunctions}).
	We first consider the vector space $\mathcal{O}(S^1)$ of real-analytic complex-valued functions on $S^1$.	
	Let $G \in (\mathcal{O}(S^1))^\vee$ be a \gls{fsmooth} real-valued $\R$-linear functional.
	Consider the complexification of $G$ defined as $H(f) = G(f) - \ii G(\ii f)$, which is a \gls{fsmooth} complex-valued $\C$-linear functional such that $G(f) = \Re H(f)$.
	Now, we apply~\cite[Theorem~2.1.9]{Morimoto:An_introduction_to_Satos_hyperfunctions} to $H$ (note that smooth functionals are, in particular, continuous, and thus, ``analytic functionals''):
	There exists unique holomorphic functions $q^+ \colon \odisk \to \C$ and $q^- \colon \jodisk \to \C$ such that $q^-(\infty) = 0$ and the complexification $H$ has the integral representation
	\begin{equation}
		\label{eq:fun_int_repr}
		H(f) =
		\int_{(1 - \varepsilon) S^1} f(z) \; q^+(z) \; \dd z
		+ \int_{(1 + \varepsilon) S^1} f(z) \; q^-(z) \; \dd z
		,\qquad f \in \mathcal{O}(S^1),
	\end{equation}
	where the integral does not depend on $\varepsilon > 0$ (but it must be chosen small enough depending on $f$ so that it is defined).
	The integral representation $\eqref{eq:fun_int_repr}$ with $G(f) = F(f(z)\partial_z)$ and $F \in (\VectC)^\vee$ now yields the asserted (coordinate independent) integral representation~\eqref{eq:integral_repr}  for $F$, by regarding the functions $q^+(z)$ and $q^-(z)$ as quadratic differentials $Q^+ = q^+(z) \; \dd z^2$ and $Q^- = q^-(z) \; \dd z^2$, pairing them with the vector field $v = f(z) \; \partial_z$.
\end{proof}

Since the variations~\eqref{eq:defthetav} vanish on negative modes of $v$, they only involve a single quadratic differential $Q[\phi^\pm]$ which is holomorphic on $\jodisk$ such that 
\begin{equation}
	\label{eq:integral_repr_lpot}
	\eval{\pdv{t}}_{t = 0}
	\Omega^Z_{1,1}\Big((\disk \diffActing{1} \Phi^{v}) \diffActing{1} \zeta^\pm, \disk\Big)
	= \Re \int_{(1+\varepsilon)S^1} v \; Q[\phi^\pm], \qquad v \in \VectC.
\end{equation}

\subsection{The Schwarzian derivative}
\label{section:diskdisk_schwarzian}

\begin{figure}[t]
\centering
\includestandalone[]{figures/fig_univalent_composition}
\caption{
	The inner part (with index ``$1$'') is the commutative diagram of conformal welding in Proposition~\ref{prop:conformal_welding} with $\phi = \phi_1^+$. 
We introduce $\phi_1^- = \invinv(\phi_1^+)$ as in Equation~\eqref{eq:disk_disk_welding} to make the diagram symmetric.
	The full diagram shows the composition of the univalent functions, $\zeta_3^+ = \zeta_1^+ \circ \zeta_2^+$.
	Instead of composing the welding homeomorphisms, we obtain $\phi^{\pm}_3$ and $\zeta_3^-$ from $\zeta_3^+$ using conformal welding.
	This only works if the normalization is chosen in such a way that $\zeta_2^+(S^1)$ is contained in the unit disk, as illustrated by the hatched domains in the figure.
}
\label{fig:univalent_composition}
\end{figure}

If $\phi^\pm$ is a M\"obius transformation, Lemma~\ref{lemma:crit_pt} implies that $Q[\phi^\pm] = 0$ in Equation~\eqref{eq:integral_repr_lpot}.
Now, we determine $Q^-[\phi^\pm]$ away from M\"obius transformations.
By diffeomorphism invariance, the variation~\eqref{eq:defthetav} may also be expressed in terms of~$\zeta^\mp$,
\begin{equation}
\begin{aligned}
	\label{eq:defthetav2}
	\Re \int_{(1+\varepsilon)S^1} v \; Q[\phi^\pm]
	&=
	\eval{\pdv{t}}_{t = 0}
	\Omega^Z_{1,1}\Big(
		(\disk \diffActing{1} \Phi^{v}) \diffActing{1}
		\big((\inversion \circ \zeta^\mp \circ \inversion) \circ
		\phi^\pm\big),
		\disk
	\Big)
	\\&=
	\eval{\pdv{t}}_{t = 0}
	\Omega^Z_{1,1}\Big(
		(\disk \diffActing{1} \Phi^{v}) \diffActing{1}
		(\inversion \circ \zeta^\mp \circ \inversion),
		\disk \diffActing{1} \phi^\mp
	\Big)
	\\&=
	\eval{\pdv{t}}_{t = 0}
	\Omega^Z_{1,1}\Big(
		(\disk \diffActing{1} \Phi^{v}) \diffActing{1}
		(\inversion \circ \zeta^\mp \circ \inversion),
		\disk \diffActing{1} \zeta^\mp
	\Big).
\end{aligned}
\end{equation}
Let $\phi_1^\pm, \phi_2^\pm \in \Diffpan$ be pairs of diffeomorphisms like $\phi^\pm$ above, and $\zeta_1^\pm, \zeta_2^\pm$ the respective univalent functions in the  conformal welding decompositions.
With the same normalization $(\zeta_1^-)' = (\zeta_2^-)' = a$ as in Item~\ref{prop:conformal_welding_item2} of Proposition~\ref{prop:conformal_welding}, the composition $\zeta^+_1 \circ \zeta^+_2 \in \univalent$ exists\footnote{
	From here on, with the current normalization, only the case where ``$\pm$'' is ``$-$'' and ``$\mp$'' is ``$+$'' holds.
	However, we keep track of both cases, as the other case may be obtained by changing the normalization to $(\zeta_1^+)' = (\zeta_2^+)' = a$, where $|a| > 4$, so that we can compose $\zeta^-_1 \circ \zeta^-_2 \in \univalent$.
	Since Equation~\eqref{eq:integral_repr_lpot} is invariant under M\"obius transformations, the variational formula is independent of the normalization.
}, and we let it define a pair of diffeomorphisms $\phi_3^\pm \in \Diffpan$ via Item~\ref{prop:conformal_welding_item2} of Proposition~\ref{prop:conformal_welding}, as shown in Figure~\ref{fig:univalent_composition}.
In the case of $\phi_3^\pm$, Equation~\eqref{eq:defthetav2} reads
\begin{equation}
	\label{eq:theta_proj_p}
	\Re \int_{(1+\varepsilon)S^1} v \; Q[\phi_3^\pm] =
	\eval{\pdv{t}}_{t = 0}
	\Omega^Z_{1,1}\Big((\disk \diffActing{1} \Phi^{v})
		\diffActing{1} (\inversion \circ \zeta_1^\mp \circ \zeta_2^\mp \circ \inversion), \disk \diffActing{1} (\zeta_1^\mp \circ \zeta_2^\mp)
	\Big).
\end{equation}
We can now apply the disk-deformation-disk cocycle identity~\eqref{eq:cocycle_id_mixed_3} to the surfaces $\Sigma_1 = \disk \diffActing{1} (\Phi^{v} \circ \inversion \circ \zeta_1^\mp \circ \zeta_2^\mp \circ \inversion)$ and $\Sigma_2 = \disk \diffActing{1} \zeta_1^\mp$, and to the complex deformation~$\zeta_2^\mp$, to obtain
\begin{equation}
	\label{eq:cocycle_variation_theta_1}
\begin{aligned}
	\textnormal{\eqref{eq:theta_proj_p}}
	=
	\eval{\pdv{t}}_{t = 0} \Bigg(
		&\!- \Omega^Z_{\disk, 1, 1}\Big(\disk \diffActing{1} \zeta_1^\mp, \zeta_2^\mp\Big)
		\\ &+ \Omega^Z_{1, 1}\Big(\big(\disk \diffActing{1} (\Phi^{v} \circ \inversion \circ \zeta_1^\mp \circ \zeta_2^\mp \circ \inversion)\big) \diffActing{1} \invinv(\zeta_2^\mp), \disk \diffActing{1} \zeta_1^\mp\Big)
		\\ &+ \Omega^Z_{\disk, 1, 1}\Big(\big(\disk \diffActing{1} (\Phi^{v} \circ \inversion \circ \zeta_1^\mp \circ \zeta_2^\mp \circ \inversion)\big), \invinv(\zeta_2^\mp)\Big)
	\Bigg).
\end{aligned}
\end{equation}
Note that the first term is time-independent, and the second term is just Equation~\eqref{eq:defthetav2} for $\phi_1^\pm$.
We will identify the third term in terms of the same cocycle identity applied to~\eqref{eq:defthetav2} for $\phi_1^\pm$ and a vector field $w \in \VectC$.
Indeed, taking now the surfaces $\Sigma_1 = \disk \diffActing{1} (\Phi^{w} \circ \inversion \circ \zeta_2^\mp \circ \inversion)$ and $\Sigma_2 = \disk$ and the complex deformation $\zeta_2^\mp$, we obtain
\begin{equation}
\begin{aligned}
	\Re \int_{(1+\varepsilon)S^1} w \; Q[\phi_2^\pm]
	=
	\eval{\pdv{t}}_{t = 0} \Bigg(
		&- \Omega^Z_{\disk, 1, 1}\Big(\disk, \zeta_2^\mp\Big)
		\\ &+ \Omega^Z_{1, 1}\Big(\big(\disk \diffActing{1} (\Phi^{w} \circ \inversion \circ \zeta_2^\mp \circ \inversion)\big) \diffActing{1} \invinv(\zeta_2^\mp), \disk\Big)
		\\ &+ \Omega^Z_{\disk, 1, 1}\Big(\disk \diffActing{1} (\Phi^{w} \circ \inversion \circ \zeta_2^\mp \circ \inversion), \invinv(\zeta_2^\mp)\Big)
	\Bigg).
\end{aligned}
\end{equation}
Again, the first term is time-independent.
The second term is the variation at the identity $\Omega^Z_{1,1}(\disk \diffActing{1} \Phi^{w}, \disk)$, which vanishes by Lemma~\ref{lemma:crit_pt}.
Finally, we see that the third term agrees with the third term in Equation~\eqref{eq:cocycle_variation_theta_1} by choosing $w = (\inversion \circ \zeta_1^\mp \circ \inversion)^* v$, since the deformation to the left of $\Phi^v$ in $\Phi^{w} = (\inversion \circ \zeta_1^\mp \circ \inversion)^{-1} \circ \Phi^v \circ (\inversion \circ \zeta_1^\mp \circ \inversion)$ can be absorbed by the unit disk using 
Item~\ref{lemma:complex_deformations:conj} of Lemma~\ref{lemma:complex_deformations}.
In summary, we have shown that
\begin{equation}
	\Re \int_{(1+\varepsilon)S^1} v \; Q[\phi_3^\pm]
	= \Re \int_{(1+\varepsilon)S^1} \big((\inversion \circ \zeta_1^\mp \circ \inversion)^* v\big) \; Q[\phi_2^\pm]
	+ \Re \int_{(1+\varepsilon)S^1} v \; Q[\phi_1^\pm].
\end{equation}
For the quadratic differential, this implies that
\begin{equation}
	Q^-[\phi_3^\pm] = (\inversion \circ \zeta_1^\mp \circ \inversion)_* Q^-[\phi_2^\pm] + Q^-[\phi_1^\pm].
\end{equation}
This identity together with the $\mob \cap \DefC$ invariance characterizes $Q^-[\phi^\pm]$ as a scalar multiple of the Schwarzian derivative $\schwarzian{\invinv(\zeta^\mp)}$.
Note that this quadratic differential is holomorphic in a neighborhood of $\jcdisk$. Hence, the integrals in \eqref{eq:integral_repr_lpot} may just be taken over $S^1$, that is, we may set $\varepsilon = 0$.

Now, in order to prove Theorem~\ref{thm:lpot},  what remains is to compute the constant $C \in \C$ such that 
\begin{equation}
	\label{eq:lpot_schwarzian_prop}
	\eval{\pdv{t}}_{t = 0}
	\Omega^Z_{1,1}\Big(\disk \diffActing{1} \gamma^\pm, \disk\Big)
	= \Re \Big( C \int_{S^1} v^\pm \; \schwarzian{\invinv(\zeta^\mp)} \Big).
\end{equation}

\subsection{Identification of the central charge}
\label{section:diskdisk_charge}

To begin, for general $v, w \in \VectC$, the variation of the Schwarzian derivative of the flow~$\Phi^w$ is $\eval{\pdv{t}}_{s = 0} \schwarzian{\Phi^{w}(s, \blank)} = w'''(z)$.
By integrating by parts, we compute
\begin{equation}
\label{eq:schwarzian_cocycle}
\begin{aligned}
	\eval{\pdv{s}}_{s = 0}
	\Re \bigg( C \int_{S^1} v \; \schwarzian{\Phi^{w}(s, \blank)} \bigg)
	&= - \Re \bigg( C \int_{S^1} v'(z) \; w''(z) \; \dd z \bigg)
	\\ &= - 24 \pi \, \Re \Big( C \, \big(\gelfandfuks(v, w) - \rotcocyclealg(v, w)\big) \Big),
\end{aligned}
\end{equation}
in terms of the cocycles in Equation~\eqref{eq:cocycle_difference}.

We compute the second variation of $\Omega^Z_{1,1}(\disk \diffActing{1} \blank, \disk)$ at $\phi^\pm = \id$ in the direction of tangent vectors $[s \mapsto [t \mapsto \gamma^\pm(t, s, \blank)]_\sim]_\sim \in T T \Diffpan$ described by two-parameter curves $\gamma^\pm \in \curves{\curves{\Diffpan}}$ such that $\invinv(\gamma^\pm) = \gamma^\mp$.
Keeping the same conformal welding decomposition as in~\eqref{eq:disk_disk_welding}, now with two-parameter curves $\eta^\pm \in \curves{\curves{\univalent}}$ such that $\eta^\pm(0, 0, \blank) = \zeta^\pm = \id$, we may still represent
\begin{equation}
[t \mapsto \eta^\pm(t, s, \blank)]_\sim = [\Phi^{v^\pm}(t, \blank) \circ \eta^\pm(0, s, \blank)]_\sim 
\end{equation}
for time-independent $v^\pm \in \VectC$ and $s \in \R$.
This leads to the same first variation formula~\eqref{eq:lpot_schwarzian_prop} in the $t$-direction.
Note that at $t = 0$, the curves $\eta^\mp(0, s, \blank) \in \curves{\univalent}$ are rooted at the identity, and define tangent vectors 
\begin{equation}
[s \mapsto \eta^\pm(0, s, \blank)]_\sim = [s \mapsto \Phi^{w^\pm}(s, \blank)]_\sim \in T_{\id} \univalent 
\end{equation}
for vector fields $w^\pm \in \VectC$, holomorphic in a neighborhood of $\cdisk$.
Note that the tangent vector in the Schwarzian derivative in~\eqref{eq:lpot_schwarzian_prop} may be represented by $[\invinv(\Phi^{w^\mp})]_\sim$ = $[\Phi^{-\inversion^* w^\mp}]_\sim$.
With this setup, Equation~\eqref{eq:schwarzian_cocycle} implies that
\begin{equation}
\begin{aligned}
	&\eval{\pdv[2]{}{t}{s}}_{t = s = 0}
	\Omega^Z_{1,1}\Big(\disk \diffActing{1} \gamma^\pm(t, s, \blank), \disk\Big)
	\\ = \; &- 24 \pi \, \Re \Big( C \, \big(\gelfandfuks(v^\pm, -\inversion^* w^\mp) - \rotcocyclealg(v^\pm, -\inversion^* w^\mp)\big) \Big).
\label{eq:second_variation_one}
\end{aligned}
\end{equation}
In particular, for $v^+ = v^- = \ell_n$, this evaluates to
$
-24 \pi \, \Re \big(C \frac{\ii}{12} (n^3 - n)\big)
$.
Since this is a variation of a real-valued function, we can already conclude that $C \in \ii \R$.

By differentiating Equation~\eqref{eq:defthetav2}, we compute the same variation in a different way:
\begin{equation}
\begin{aligned}
	\textnormal{\eqref{eq:second_variation_one}}
	&=	
	\eval{\pdv[2]{}{t}{s}}_{t = s = 0}
	\Omega^Z_{1,1}\Big(\disk \diffActing{1} \big(\Phi^{v^\pm}(t, \blank) \circ \Phi^{\inversion^* w^\mp}(s, \blank)\big), \disk \diffActing{1} \Phi^{w^\mp}(s, \blank)\Big)
	\\ &=
	\eval{\pdv[2]{}{t}{s}}_{t = s = 0}
	\Bigg(
	\Omega^Z_{1,1}\Big(\disk \diffActing{1} \big(\Phi^{v^\pm}(t, \Phi^{\inversion^* w^\mp}(s, \blank))\big), \disk\Big)
	\\ &\phantom{=
	\eval{\pdv{s}}_{s = 0}
	\eval{\pdv{t}}_{t = 0}
	\Bigg(}
	+ \Omega^Z_{1,1}\Big(\disk \diffActing{1} \Phi^{v^\pm}(t, \blank), \disk \diffActing{1} \Phi^{w^\mp}(s, \blank)\Big)
	\Bigg).
\end{aligned}
\label{eq:second_variation_theta}
\end{equation}
We can apply the disk-deformation-deformation-disk cocycle identity obtained from combining equations~\eqref{eq:cocycle_id_mixed_1}--\eqref{eq:cocycle_id_mixed_3} and~\eqref{eq:cocycle_mixed_compat_2} to the term in the last line,
\begin{equation}
\begin{aligned}
	&\mathrlap{\Omega^Z_{1,1}\Big(\disk \diffActing{1} \Phi^{v^\pm}(t, \blank), \disk \diffActing{1} \Phi^{w^\mp}(s, \blank)\Big)} \\
	&=
	&&- \Omega^Z_{\disk,1,1}\Big(\disk, \Phi^{v^\pm}(t, \blank)\Big)
	- \Omega^Z_{\disk,1,1}\Big(\disk, \Phi^{w^\mp}(s, \blank)\Big)
	\\
	&&&+
	\Omega^Z_{\disk,1}\Big(\Phi^{v^\pm}(t, \blank), \Phi^{-\inversion^* w^\mp}(s, \blank)\Big) \\
	&&&+ \Omega^Z_{\disk,1,1}\Big(\disk, \Phi^{v^\pm}(t, \Phi^{-\inversion^* w^\mp}(s, \blank))\Big)
	\\
	&&&+
	\Omega^Z_{1,1}\Big(\disk \diffActing{1} \big(\Phi^{v^\pm}(t, \Phi^{-\inversion^* w^\mp}(s, \blank))\big), \disk\Big).
\end{aligned}
\label{eq:cocycle_dddd_flow}
\end{equation}
The terms of the form $\Omega_{\disk, 1, 1}(\disk, \blank)$ vanish by the identity~\eqref{eq:cocycle_surfaces_equal_vanishing}.
Moreover, because we have $[\Phi^{-w}]_\sim = - [\Phi^{w}]_\sim \in T_{\id} \DefC$, the last term in \eqref{eq:cocycle_dddd_flow} cancels with the term in the second line of~\eqref{eq:second_variation_theta}.
This leaves us with
\begin{equation}
	\label{eq:second_var_cocycle1}
	\eval{\pdv[2]{}{t}{s}}_{t = s = 0}
	\Omega^Z_{1,1}\Big(\disk \diffActing{1} \gamma^\pm(t, s, \blank), \disk\Big)
	=
	\eval{\pdv[2]{}{t}{s}}_{t = s = 0}
	\Omega^Z_{\disk,1}\Big(\Phi^{v^\pm}(t, \blank), \Phi^{-\inversion^* w^\mp}(s, \blank)\Big).
\end{equation}
The right-hand side of~\eqref{eq:second_var_cocycle1} is almost the Lie algebra cocycle associated to $\Omega^Z_{\disk,1}$ (see \cite[Section~3]{Maibach-Peltola:Complex_deformations_of_the_circle}), 
which is cohomologous to $\charge \Im \gelfandfuks$ in $H^2(\VectC; \VectR; \R)$ by~\eqref{eq:cocycle_is_bott_nonrel}.
What is missing, is a factor of $\frac{1}{2}$, and the term with $v^\pm$ and $-\inversion^* w^\mp$ swapped.
That term can be computed using disk-deformation-deformation-disk cocycle identities:
\begin{equation}
\begin{aligned}
	&\mathrlap{\Omega^Z_{\disk,1}\Big(
		\Phi^{-\inversion^* w^\mp}(s, \blank),
		\Phi^{v^\pm}(t, \blank)
	\Big)} \\
	&=
	&&- \Omega^Z_{\disk,1,1}\Big(
		\disk,
		\Phi^{-\inversion^* w^\mp}(s, \blank) \circ \Phi^{v^\pm}(t, \blank)
	\Big)
	\\ &&&- \Omega^Z_{1,1}\Big(
		\disk \diffActing{1} \Phi^{-\inversion^* w^\mp}(s, \blank) \circ \Phi^{v^\pm}(t, \blank), \disk
	\Big)
	\\ &&&+ \Omega^Z_{\disk,1,1}\Big(
		\disk, \Phi^{-\inversion^* w^\mp}(s, \blank)
	\Big)
	\\ &&&+ \Omega^Z_{\disk,1,1}\Big(
		\disk, \Phi^{-\inversion^* v^\pm}(t, \blank)
	\Big)
	\\ &&&+ \Omega^Z_{1,1}\Big(
		\disk \diffActing{1} \Phi^{-\inversion^* w^\mp}(s, \blank),
		\disk \diffActing{1} \Phi^{-\inversion^* v^\pm}(t, \blank)
	\Big).
\end{aligned}
\end{equation}
On the right-hand side, the first, third, and fourth terms vanish by Equation~\eqref{eq:cocycle_surfaces_equal_vanishing}.
In the second and fifth terms, we find that $\disk \diffActing{1} \Phi^{-\inversion^* w^\mp}(s, \blank) = \disk$ by Item~\ref{lemma:complex_deformations:disk_invariant} of Lemma~\ref{lemma:complex_deformations}, so they are independent of $s$, making the derivative vanish.
Thus, we conclude that
\begin{equation}
	- 24 \pi \, \Re \Big( C \, \gelfandfuks(v^\pm, -\inversion^* w^\mp) \Big)
	= 2 \; \charge \Im \Big( \gelfandfuks(v^\pm, -\inversion^* w^\mp) + \beta([v^\pm, -\inversion^* w^\mp]) \Big)
\end{equation}
for some $\beta \in Z^1(\VectC; \VectR; \R)$.
This only holds for vector fields $v^\pm$ and $w^\pm$ which are both holomorphic in a neighborhood of $\cdisk$.
Picking $v^+ = w^- = \ell_n$ such that $-\inversion^* w^- = \ell_{-n}$, we find that
\begin{equation}
	- 12 \pi \, \Im C
	= \; \charge + \frac{12 \beta(\ell_{0})}{n^2} , \qquad n \in \Z.
\end{equation}
By taking the limit $n \to \infty$, we find that $\Im C = - \frac{\charge}{12 \pi}$.
Since $C \in \ii \R$, we deduce that
\begin{equation}
	\eval{\pdv{t}}_{t = 0}
	\Omega^Z_{1,1}\Big(\disk \diffActing{1} \gamma^\pm, \disk\Big)
	= \frac{\charge}{12 \pi} \Im \int_{S^1} v^\pm \; \schwarzian{\invinv(\zeta^\mp)}.
	\label{eq:lpot_schwarzian_prop_final}
\end{equation}
This concludes the proof of Theorem~\ref{thm:lpot}, because Equation~\eqref{eq:lpot_schwarzian_prop_final} agrees with the variation of the Loewner energy in Equation~\eqref{eq:lenergy_variation} up to the multiplicative factor $-\frac{\charge}{24}$.

\section{Isomorphisms of local \texorpdfstring{$\Rp$}{R+}-extensions}
\label{section:iso}
Consider local \gls{fsmooth} $\Rp$-extensions $\RMFE$ and $\RMFD$ with equal central charge $\charge \in \R$.
In this section, we construct an isomorphism between $\RMFE$ and $\RMFD$, proving the main result of this work, Theorem~\ref{thm:universality}.
We begin by detailing the precise notion of isomorphism that we use.
\begin{definition}
	An \emph{isomorphism} $\Psi \colon \RMFE \to \RMFD$ of \gls{fsmooth} $\Rp$-extensions consists of \gls{fsmooth} morphisms of principal $\Rp$-bundles,
    \begin{equation}
        \Psi_{\genus, \boundaries} \colon \RMFE(\moduligb) \longrightarrow \RMFD(\moduligb), \qquad \genus, \boundaries \geq 0,
    \end{equation}
    such that they preserve the sewing isomorphisms, that is,
	\label{def:modular_functor_isomorphism}
	\begin{align}
		\Psi_{\genus_1 + \genus_2, \boundaries_1 + \boundaries_2 - 2}(\blank \sewx{\RMFE}{j}{k} \blank) &= \Psi_{\genus_1, \boundaries_1}(\blank) \sewx{\RMFD}{j}{k} \Psi_{\genus_2, \boundaries_2}(\blank), \\
		\sewxself{\RMFD}{j}{k} \Psi_{\genus, \boundaries}(\blank) &= \Psi_{\genus + 1, \boundaries - 2}(\sewxself{\RMFE}{j}{k} \blank).
	\end{align}
\end{definition}

Our strategy is constructive: We find trivializations $Z$ and $W$, respectively of  $\RMFE$ and $\RMFD$, such that all cocycles defined in Section~\ref{section:cocycles} agree.
Then, the isomorphism is given by
\begin{equation}
\begin{aligned}
	\Psi_{\genus,\boundaries} \colon \RMFE(\moduligb) &\longrightarrow \RMFD(\moduligb), \\
	Z(\Sigma) &\longmapsto W(\Sigma).
\end{aligned}
\label{eq:iso_triv}
\end{equation}
We shall construct the components $(\Psi_{\genus, \boundaries})_{\genus, \boundaries \geq 0}$ one by one.
With each new case, there are new cocycles to be considered, and we need to show that these agree for $Z$ and $W$.

\subsection{Genus zero}
\label{section:genus_zero}

First, we consider the sphere and disks (genus $\genus = 0$ and $\boundaries \in \{0,1\}$).
By Proposition~\ref{prop:reparam_indep}, since $\spheres = \{\hat{\C}\}$ and $\disks = \setsuchthat{\disk \diffActing{1} \phi}{\phi \in \Diffpan}$, we know that reparametrization invariant trivializations over $\spheres$ and $\disks$ are unique up to a constant factor.
Thus, we may normalize them in such a way that
\begin{equation}
	\Omega^Z_{1,1}(\disk, \disk) = \Omega^W_{1,1}(\disk, \disk) = 0.
	\label{eq:disk_disk_normalization}
\end{equation}
We define $\Psi_{0,0}$ and $\Psi_{0,1}$ using Equation~\eqref{eq:iso_triv} for the normalized trivializations.

The only cocycles which come up at this point are the disk-disk cocycles (cf.~Section~\ref{section:disk_disk}).
Since we assume that $\RMFE$ and $\RMFD$ have equal central charge, by Theorem~\ref{thm:lpot}, the disk-disk cocycles agree up to a constant, which is fixed by the normalization~\eqref{eq:disk_disk_normalization}.
This will suffice for establishing an isomorphism for all genus $0$ surfaces via induction on the number $\boundaries \geq 1$ of boundary components. 
To this end, we invoke the $\genus = 0$ case of Propositions~\ref{prop:induction_trivialization} and~\ref{prop:induction_extension} (proven in the next Section~\ref{section:induction_step}), thus proving Theorem~\ref{thm:universality} in genus zero.
\begin{theorem}
	All local \gls{fsmooth} $\Rp$-extensions with equal central charges are isomorphic in genus zero.
	\label{thm:universality_genus_zero}
\end{theorem}

\begin{remark}
	By using the unit disk as the reference surface, the trivializations $Z$ and $W$ on disks define trivializations $Z_{\disk, 1}$ and $W_{\disk, 1}$ of the central extensions $\RMFE(\DefC)$ and $\RMFD(\DefC)$ of complex deformations via Equation~\eqref{eq:local_triv_ext}.
	For a composable pair $\phi, \psi \in \DefC$ of complex deformations, we may compute $\Omega^Z_{\disk, 1}(\phi, \psi)$ by combining the disk-deformation-deformation cocycle identity~\eqref{eq:cocycle_mixed_compat_2} with $\Sigma = \disk = \Sr$,
\begin{equation}
	\Omega^Z_{\disk, 1}(\phi, \psi)
	= \Omega^Z_{\disk, 1,1}(\disk \diffActing{1} \phi, \psi) 
	+ \underbrace{\Omega^Z_{\disk, 1,1}(\disk, \phi)}_{= \; 0 \; \textnormal{by~\eqref{eq:cocycle_surfaces_equal_vanishing}}} 
	- \underbrace{\Omega^Z_{\disk, 1,1}(\disk, \phi \circ \psi)}_{= \; 0 \; \textnormal{by~\eqref{eq:cocycle_surfaces_equal_vanishing}}},
\end{equation}
with the expression for $\Omega^Z_{\disk, 1, 1}(\disk \diffActing{1} \phi, \psi)$ obtained from the disk-deformation-disk cocycle identity~\eqref{eq:cocycle_id_mixed_3} for $\psi$, and $\Sigma_1 = \disk \diffActing{1} \phi$ and $\Sigma_2 =\Sr = \disk$,
\begin{equation}
	\Omega^Z_{\disk, 1, 1}(\disk \diffActing{1} \phi, \psi)
	= \underbrace{\Omega^Z_{\disk, 1, 1}(\disk, \invinv(\psi))}_{= \; 0 \; \textnormal{by~\eqref{eq:cocycle_surfaces_equal_vanishing}}}
	+ \Omega^Z_{1,1}(\disk \diffActing{1} \phi, \disk \diffActing{1} \invinv(\psi))
	- \Omega^Z_{1,1}(\disk \diffActing{1} (\phi \circ \psi), \disk).
\label{eq:Dphi1phi2}
\end{equation}
We find that the cocycle $\Omega^Z_{\disk, 1}(\phi, \psi)$ with respect to the composition law of $\RMFE(\DefC)$ is fully expressed in terms of the disk-disk cocycle: 
\begin{equation}
	\Omega^Z_{\disk, 1}(\phi, \psi) 
	= 
\Omega^Z_{1,1}(\disk \diffActing{1} \phi, \disk \diffActing{1} \invinv(\psi))
	- \Omega^Z_{1,1}(\disk \diffActing{1}( \phi \circ \psi), \disk).
\end{equation}
Now, because the disk-disk cocycle of $\RMFE$ agrees with that of $\RMFD$, we may conclude that $\Omega^Z_{\disk, 1}(\phi, \psi) = \Omega^W_{\disk, 1}(\phi, \psi)$ for all composable $\phi, \psi \in \DefC$, and we obtain an isomorphism
\begin{equation}
\Lambda \colon \RMFE(\DiffC) \longrightarrow \RMFD(\DiffC), \qquad
Z_{\disk, 1}(\phi) \longmapsto W_{\disk, 1}(\phi),
\label{eq:iso_deformations_R}
\end{equation}
which is compatible with the composition laws~\eqref{eq:E_defc_composition} of the central extensions.
Note that this is an $\Rp$-bundle isomorphism covering the identity on $\DefC$, and moreover, since $Z_{\disk, 1}(\id) = Z(\disk) \otimes (Z(\disk))^\vee$ is mapped to $W_{\disk, 1}(\id) = W(\disk) \otimes (W(\disk))^\vee$, the isomorphism $\RMFE(\id) \to \RMFD(\id)$ is given by composing the evaluation $\evaluation \colon \RMFE(\id) \to \Rp$ with the inverse evaluation $\evaluation^{-1} \colon \Rp \to \RMFD(\id)$.
We conclude that Equation~\eqref{eq:iso_deformations_R} defines an isomorphism of central extensions in the sense of an isomorphism of exact sequences~\eqref{eq:defc_ext_sequence}, 
\begin{equation}
\begin{tikzcd}
	\{1\} & \Rp & \RMFE(\DefC) & \DefC & \{1\} \\
	\{1\} & \Rp & \RMFD(\DefC) & \DefC & \{1\} ,
	\arrow[from=1-1, to=1-2]
	\arrow[from=1-2, to=1-3, hook]
	\arrow[from=1-3, to=1-4, two heads]
	\arrow[from=1-4, to=1-5]
	\arrow[from=2-1, to=2-2]
	\arrow[from=2-2, to=2-3, hook]
	\arrow[from=2-3, to=2-4, two heads]
	\arrow[from=2-4, to=2-5]
	\arrow[from=1-1, to=2-1]
	\arrow["{\id}"', from=1-2, to=2-2]
	\arrow["{\Lambda}"', from=1-3, to=2-3]
	\arrow["{\id}"', from=1-4, to=2-4]
	\arrow[from=1-5, to=2-5]
\end{tikzcd}
\end{equation}
with the identity on $\Rp$ and $\DefC$.
\end{remark}

\subsection{Induction step on the number of boundary components}
\label{section:induction_step}

To recursively extend an isomorphism of the form~\eqref{eq:iso_triv} to higher number of boundary components, we proceed by induction on $\boundaries$, namely, keeping the genus fixed, and increasing the number of boundary components by one. 
Note that this induction step requires that the new $\Psi_{\genus, \boundaries}$ has $\boundaries \geq 2$ (so the cases $\boundaries = 0$ and $\boundaries = 1$ for each next higher genus need to be considered separately).
The induction step itself assumes that $\RMFE$ and $\RMFD$ are local.

The construction of the new isomorphism $\Psi_{\genus, \boundaries}$ involves choosing particular trivializations of $\RMFE$ and $\RMFD$.
By Item~\ref{prop:fprojective:repr} of Proposition~\ref{prop:fprojective}, any surface in $\moduligb$ may be obtained from a Fuchsian projective surface in $\moduligbmobf$ by boundary reparametrizations.
Since we are seeking to define a reparametrization invariant trivialization, we let it be determined by pullback of a trivialization over $\moduligbmobf$ under reparametrization.
In this process, we will have to show invariance under M\"obius reparametrizations, that is, under $\phi \in \psl \subset \Diffpan$.
While it does not seem essential how exactly the trivializations over $\moduligbmobf$ are chosen, it is important that they are chosen in a consistent manner for $\RMFE$ and $\RMFD$.
We make this choice by sewing a disk to the first boundary component, as follows.

\begin{proposition}
	Fix $\boundaries \geq 1$ and $\genus \geq 0$.
	Let $Z$ be a reparametrization invariant trivialization of a local \gls{fsmooth} $\Rp$-extension $\RMFE$ over $\disks$ and $\moduli{\genus}{\boundaries - 1}$.
	Then, extending the trivialization to $\Sigma \in \moduligbmobf$ by solving the linear equations
	\begin{equation}
		\label{eq:z_linear}
		Z(\Sigma) \sewx{\RMFE}{1}{1} Z(\disk) = Z(\Sigma \sew{1}{1} \disk)
	\end{equation}
	yields a well-defined reparametrization invariant trivialization over $\moduligb$.
	\label{prop:induction_trivialization}
\end{proposition}
\begin{proof}
	The new trivialization~\eqref{eq:z_linear} over Fuchsian projective surfaces $\Sigma \in \moduligbmobf$ singles out $\partial_1 \Sigma$ by attaching a disk to that boundary component, which reduces the number of boundary components by one.
	We will check that, for $1 \leq j \leq \boundaries$ and $\phi \in \Diffpan$, assuming that $\Sigma \diffActing{j} \phi \in \moduligbmobf$ is again Fuchsian projective, it follows that $\Omega_{\disk, 1, j}^Z(\Sigma, \phi) = 0$ (see Item~\ref{prop:local_equiv_2} of Proposition~\ref{prop:local_equiv}) --- yielding reparametrization invariance.
	By Item~\ref{prop:fprojective:mobius} of Proposition~\ref{prop:fprojective}, it follows that $\phi \in \psl$ is a M\"obius transformation.
	By~\eqref{eq:Ecocycle_pair}, the linear equation~\eqref{eq:z_linear} is equivalent to $\Omega^Z_{1,1}(\Sigma, \disk) = 0$.
	For ${j \neq 1}$, consider the surface-deformation-disk cocycle identity~\eqref{eq:cocycle_id_mixed_2} for $\phi \in \psl$, and ${\Sigma_1 = \Sigma}$ and $\Sr = \Sigma_2 = \disk$, 
	\begin{equation}
	\underbrace{\Omega^Z_{1,1}(\Sigma \diffActing{j} \phi, \disk)}_{= \; 0 \; \textnormal{by~\eqref{eq:z_linear}}}
	\; + \; \Omega^Z_{\disk, 1, j}(\Sigma, \phi)
	\; = \; \Omega^Z_{\disk, 1, j}(\Sigma \sew{1}{1} \disk, \phi)
	\; + \; \underbrace{\Omega^Z_{1,1}(\Sigma, \disk)}_{= \; 0 \; \textnormal{by~\eqref{eq:z_linear}}}.
	\label{eq:PSL_cocycle_ides}
	\end{equation}
	The first term on the right-hand side of~\eqref{eq:PSL_cocycle_ides} vanishes because we assumed reparametrization invariance for any genus $\genus$ surface with $\boundaries - 1$ boundary components, such as $\Sigma \sew{1}{1} \disk \in \moduli{\genus}{\boundaries - 1}$ (see again Item~\ref{prop:local_equiv_2} of Proposition~\ref{prop:local_equiv}).
	Therefore, also the second term of the left-hand side of~\eqref{eq:PSL_cocycle_ides} vanishes: $\Omega_{\disk, 1, j}^Z(\Sigma, \phi) = 0$, as desired.
	For $j = 1$, we can likewise use the surface-deformation-disk cocycle identity~\eqref{eq:cocycle_id_mixed_3} with $\Sigma_1 = \Sigma$ and $\Sr = \Sigma_2 = \disk$, 
	\begin{equation}
		\overbrace{\Omega^Z_{1,1}(\Sigma \diffActing{1} \phi, \disk)}^{= \; 0 \; \textnormal{by~\eqref{eq:z_linear}}}
		\; + \; \Omega^Z_{\disk, 1,1}(\Sigma, \phi)
		\; = \;\; \underbrace{
			\Omega^Z_{1,1}(\Sigma, \disk \diffActing{1} \invinv(\phi))
		}_{
			\mathclap{\begin{subarray}{l}
			= \; \Omega^Z_{1,1}(\Sigma, \disk) \; \textnormal{by Item~\ref{lemma:complex_deformations:disk_invariant} of Lemma~\ref{lemma:complex_deformations}} \\
			= \; 0 \; \textnormal{by~\eqref{eq:z_linear}}
	\end{subarray}}
		}
		\; + \; \overbrace{\Omega^Z_{\disk, 1,1}(\disk, \invinv(\phi))}^{= \; 0 \; \textnormal{by~\eqref{eq:local_surface_def}}}  ,
	\end{equation}
	so $\Omega^Z_{\disk, 1,1}(\Sigma, \phi)=0$, as desired.
	This proves the reparametrization invariance over $\moduligbmobf$.

	By Item~\ref{prop:fprojective:repr} of Proposition~\ref{prop:fprojective}, any other surface $\Sigma \in \moduligb$ may be obtained from a Fuchsian projective surface by boundary reparametrization, thus defining the trivialization 
	\begin{equation}
		Z(\Sigma)
		=
		Z(\check{\Sigma}) \diffActingE{1} Z_{\disk, 1}(\phi_1) \cdots  \diffActingE{\boundaries} Z_{\disk, 1}(\phi_\boundaries),
		\label{eq:triv_non_projective}
	\end{equation}
	where $\check{\Sigma} \in \moduligbmobf$ and $\phi_1, \dots, \phi_\boundaries \in \Diffpan$ are such that $\Sigma = \check{\Sigma} \diffActing{1} \phi_1 \cdots \diffActing{\boundaries} \phi_\boundaries$.
	The representative is unique up to choice of M\"obius transformations $\psi_1, \dots, \psi_\boundaries \in \mob$ resulting in
	\begin{equation}
		Z(\check{\Sigma} \diffActing{1} \psi_1 \cdots \diffActing{\boundaries} \psi_\boundaries) \diffActingE{1} Z_{\disk, 1}(\psi_1^{-1} \circ \phi_1) \cdots  \diffActingE{\boundaries} Z_{\disk, 1}(\psi_{\boundaries}^{-1} \circ \phi_\boundaries).
		\label{eq:triv_non_projective_mob}
	\end{equation}
	Thanks to the reparametrization invariance by M\"obius transformations over $\moduligbmobf$ proven above, we have 
	\begin{equation}
		Z(\check{\Sigma} \diffActing{1} \psi_1 \cdots \diffActing{\boundaries} \psi_\boundaries) = Z(\check{\Sigma}) \diffActingE{1} Z_{\disk, 1}(\psi_1) \cdots \diffActingE{\boundaries} Z_{\disk, 1}(\psi_\boundaries),
	\end{equation}
	and thanks to the assumed reparametrization invariance by $\Diffpan$ over $\disk$, we find that 
	\begin{equation}
		Z_{\disk, 1}(\psi_j^{-1} \circ \phi_j) = Z_{\disk, 1}(\psi_j^{-1}) \diffActingE{j} Z_{\disk, 1}(\phi_j) , \qquad 1 \leq j \leq \boundaries.
	\end{equation}
	By Item~\ref{lemma:complex_deformations:disk_invariant} of Lemma~\ref{lemma:complex_deformations} and Remark~\ref{remark:inverse_E}, we have $Z_{\disk, 1}(\psi_j^{-1}) = (Z_{\disk, 1}(\psi_j))^{-1}$, and therefore, the definitions of $Z(\Sigma)$ by Equations~\eqref{eq:triv_non_projective} and~\eqref{eq:triv_non_projective_mob} are equivalent.
	This proves that the newly defined trivialization $Z$ over~$\moduligb$ is well-defined.
	Finally, $Z$ is reparametrization invariant by Item~\ref{prop:local_equiv_3} of Proposition~\ref{prop:local_equiv}.
\end{proof}

For the induction step, the concrete assumption is that the isomorphism $\Psi$ is defined in terms of trivializations $Z$ and $W$ for all surfaces with any number of boundary components for genus strictly less than $\genus$. 
For genus $\genus$ surfaces, in contrast, we assume that $\Psi$ is only defined for surfaces with strictly less than $\boundaries$ boundary components.
We also assume that the identity $\Omega^Z_{j,k}(\Sigma_1, \Sigma_2) = \Omega^W_{j,k}(\Sigma_1, \Sigma_2)$ holds for all surfaces $\Sigma_1 \in \moduli{\genus_1}{\boundaries_1}$ and $\Sigma_2 \in \moduli{\genus_2}{\boundaries_2}$ 
such that either $\genus_1 + \genus_2 < \genus$ and $\boundaries_1$, $\boundaries_2$ are arbitrary, or $\genus_1 + \genus_2 = \genus$ and $\boundaries_1 + \boundaries_2 - 2 < \boundaries$.
Moreover, we assume for the self-sewing cocycles that $\Omega^Z_{j,k}(\Sigma_1) = \Omega^W_{j,k}(\Sigma_1)$ for $\Sigma_1 \in \moduli{\boundaries_1}{\genus_1}$ such that either $\genus_1 + 1 < \genus$ and $\boundaries_1$ is arbitrary, or $\genus_1 + 1 = \genus$ and $\boundaries_1 - 2 < \boundaries$.

\begin{proposition}
\label{prop:induction_extension}
Let $\Psi \colon \RMFE \to \RMFD$ be an isomorphism of \emph{local} \gls{fsmooth} $\Rp$-extensions up to genus $\genus \geq 0$ and (not including) $\boundaries \geq 2$ boundary components, and let $Z$ and $W$ respectively be reparametrization invariant trivializations of $\RMFE$ and $\RMFD$ up to genus $\genus$ and (not including) $\boundaries \geq 2$, such that $\Psi \circ Z = W$.
Then, the reparametrization invariant extensions of both trivializations $Z$ and $W$ to $\moduligb$  obtained from Proposition~\ref{prop:induction_trivialization} induce an extension of the isomorphism $\Psi$ to $\boundaries$ boundary components, defined by 
\begin{equation}
\Psi_{\genus, \boundaries}(Z(\Sigma)) = W(\Sigma) , \qquad \Sigma \in \moduligb .
\end{equation}
\end{proposition}

Note that this result only assumes that the extension of the sections to $\moduligb$ is defined using Proposition~\ref{prop:induction_trivialization}.

\begin{proof}
It suffices to check that all new cocycles involving surfaces with one additional boundary component agree for $Z$ and $W$.
The first case to consider is $\Omega^Z_{1,1}(D, \Sigma)$ for $D \in \disks$ and $\Sigma \in \moduligb$.
The surface $\Sigma$ is related to a Fuchsian projective surface $\check{\Sigma} \in \moduligbmobf$ by $\Sigma = \hat{\Sigma} \diffActing{1} \phi_1 \cdots \diffActing{\boundaries} \phi_\boundaries$ for diffeomorphisms $\phi_1, \dots, \phi_\boundaries \in \Diffpan$ (see Item~\ref{prop:fprojective:repr} of Proposition~\ref{prop:fprojective}).
By diffeomorphism invariance and symmetry, 
it suffices to consider the case $\Sigma = \check{\Sigma} \diffActing{1} \phi$ for $\phi \in \Diffpan$ and $D = \disk$.
Applying the surface-disk-disk cocycle identity~\eqref{eq:basic_cocycle_idenity} with
$\Sigma_1=\Sigma_3=\disk$, and $\Sigma_2=\check{\Sigma} \diffActing{1} \phi$,
attaching a unit disk at the second boundary component\footnote{This is why we assume $\boundaries \geq 2$ here.}, and using the induction hypothesis (the surface has one less boundary component), we obtain 
\begin{equation}
\begin{aligned}
	\Omega^Z_{1, 1}\big(\disk, \check{\Sigma} \diffActing{1} \phi\big) 
	&+ \overbrace{\Omega^Z_{2, 1}\big(\disk \sew{1}{1} (\check{\Sigma} \diffActing{1} \phi), \disk\big)}^{= \; \Omega^W_{2, 1}\big(\disk \sew{1}{1} (\check{\Sigma} \diffActing{1} \phi), \disk\big)}
	\\ &= \underbrace{\Omega^Z_{1, 1}\big(\disk, (\check{\Sigma} \diffActing{1} \phi) \sew{2}{1} \disk\big)}_{= \; \Omega^W_{1, 1}\big(\disk, (\check{\Sigma} \diffActing{1} \phi) \sew{2}{1} \disk\big)} 
	+ \Omega^Z_{2, 1}\big((\check{\Sigma} \diffActing{1} \phi), \disk\big).
\end{aligned}
\end{equation}
Thus, $\Omega^Z_{1,1}(\disk, \check{\Sigma} \diffActing{1} \phi)$ equals the same cocycle of $W$ if and only if $\Omega^Z_{2,1}(\check{\Sigma} \diffActing{1} \phi, \disk) = \Omega^Z_{2,1}(\check{\Sigma}, \disk)$ does.
To the latter, we apply the same cocycle identity again, now with $\phi = \id$.
Since the trivialization over $\check{\Sigma}$ was constructed using Proposition~\ref{prop:induction_trivialization}, we have $\Omega^Z_{1,1}(\disk, \check{\Sigma}) = 0$, and thus, the result follows.
Note that by changing the index ``$2$'' to any other boundary component ``$j$'' and putting back in the other diffeomorphisms, we obtain
\begin{equation}
\label{eq:cocycles_agree_one_disk}
\Omega^Z_{1,j}(D, \Sigma) = \Omega^W_{1,j}(D,\Sigma), \qquad
D \in \disks, \; \Sigma \in \moduligb, \; 1 \leq j \leq \boundaries.
\end{equation}
Second, any other new cocycles without self-sewing are of the form $\Omega^Z_{j, k}(\Sigma_1, \Sigma_2)$ and $\Omega^W_{j, k}(\Sigma_1, \Sigma_2)$ for surfaces $\Sigma_1 \in \moduli{\genus_1}{\boundaries_1}$ and $\Sigma_2 \in \moduli{\genus_2}{\boundaries_2}$ with $\genus_1 + \genus_2 = \genus$ and $\boundaries_1 + \boundaries_2 - 2 = \boundaries$.
If either of the surfaces is a disk, this case is covered by Equation~\eqref{eq:cocycles_agree_one_disk} above.
Since $\boundaries \geq 2$, either $\Sigma_1$ or $\Sigma_2$ has at least one more boundary component.
Without loss of generality, let us assume that $\Sigma_1$ has another boundary component $l \neq j$.
Consider the surface-disk-disk cocycle identity~\eqref{eq:basic_cocycle_idenity} 
attaching a unit disk $\Sigma_3=\disk$ to this boundary component: 
\begin{equation}
\Omega^Z_{j, k}(\Sigma_1, \Sigma_2) 
\; + \; \Omega^Z_{l, 1}(\Sigma_1 \sew{j}{k} \Sigma_2, \disk)
\; = \; \Omega^Z_{j, k}(\Sigma_1, \Sigma_2 \sew{l}{1} \disk) 
\; + \; \Omega^Z_{l, 1}(\Sigma_2, \disk).
\end{equation}
Note that all the terms except the first one involve surfaces with at least one fewer boundary components in total.
Therefore, the induction hypothesis implies that the cocycles agree.
Finally, the self-sewing cocycles $\Omega^Z_{j,k}(\Sigma)$ to consider involve $\Sigma \in \moduli{\genus - 1}{\boundaries + 2}$, since ${\sewself{j}{k} \Sigma \in \moduligb}$.
Let $1 \leq l \leq \boundaries + 2$ be a boundary component of $\Sigma$ other than $j$ or $k$.
Thanks to the cocycle identity~\eqref{eq:basic_self_cocycle_idenity} attaching a unit disk $\Sigma_2=\disk$ to this boundary component of $\Sigma_1=\Sigma$,
using the induction hypothesis and Equation~\eqref{eq:cocycles_agree_one_disk}, we see that
\begin{equation}
\Omega^Z_{j, k}(\Sigma)
\; + \; \underbrace{\Omega^Z_{l, 1}(\sewself{j}{k} \Sigma, \disk)}_{= \; \Omega^W_{l, 1}(\sewself{j}{k} \Sigma, \disk)}
\; = \;\; \underbrace{\Omega^Z_{j, k}(\Sigma \sew{l}{1} \disk)}_{= \; \Omega^W_{j, k}(\Sigma \sew{l}{1} \disk)} 
\; + \; \underbrace{\Omega^Z_{l, 1}(\Sigma, \disk)}_{= \; \Omega^W_{l, 1}(\Sigma, \disk)} .
\end{equation}
This shows that
$\Omega^Z_{j,k}(\Sigma) = \Omega^W_{j,k}(\Sigma)$, and we have covered all new cocycles.
\end{proof}

\subsection{Modular and crossing invariance}
\label{section:modular_invariance_proof}

Combining the genus $0$ case of Theorem~\ref{thm:universality} (i.e., Theorem~\ref{thm:universality_genus_zero}), and Theorem~\ref{thm:detrc_prop} concerning the real determinant line bundle, we obtain the following additional properties.

\begin{proposition}
Any local \gls{fsmooth} $\Rp$-extension $\RMFE$ is {flatly modular invariant}, {crossing invariant}, and {hyperbolically modular invariant}.
\label{prop:E_prop}
\end{proposition}
\begin{proof}
	By Theorem~\ref{thm:universality_genus_zero}, $\RMFE$ is isomorphic to $\RMFD = \Detrc$ in genus zero for some $\charge \in \R$, and $\RMFD$ is crossing invariant by Theorem~\ref{thm:detrc_prop}.
	Since crossing invariance (Definition~\ref{def:crossing}) depends only on the genus zero sewing operations, $\RMFE$ is crossing invariant.

	Let $W(\Sigma) = [\gcc(\Sigma)]$ denote the trivialization~\eqref{eq:detrc_triv} of $\RMFD$ by constant curvature metrics, and let $Z$ be a reparametrization invariant trivialization of $\RMFE$ such that the cocycles of $W$ and $Z$ agree in genus $0$.
	In particular, then $Z$ and $W$ are the cylindrical trivializations on~$\annuli$ (Proposition~\ref{prop:cylindrical}).
	Consider the differences of (self-)sewing cocycles
	\begin{equation}
		\Xi_{j, k}^{Z,W}(\Sigma_1, \Sigma_2) = \Omega^Z_{j, k}(\Sigma_1, \Sigma_2) - \Omega^W_{j, k}(\Sigma_1, \Sigma_2), \qquad
		\Xi_{j, k}^{Z,W}(\Sigma) = \Omega^Z_{j, k}(\Sigma) - \Omega^W_{j, k}(\Sigma),
	\end{equation}
	which vanish whenever $\Sigma_1$, $\Sigma_2$, and $\Sigma_1 \sew{j}{k} \Sigma_2$ ($\Sigma$ and $\sewself{j}{k} \Sigma$) are of genus $0$.
	
	The cocycle identity associated to the decomposition of a torus into a pair of pants $P \in \pants$ and a disk $D \in \disks$, sewing them as $\sewself{1}{2} (P \sew{3}{1} D) \in \tori$, is
	\begin{equation}
		\underbrace{\Xi_{3,1}^{Z,W}(P, D)}_{\textnormal{$= 0$ (genus zero)}} + \Xi_{1,2}^{Z,W}(P \sew{3}{1} D) = \Xi_{1,2}^{Z,W}(P) + \Xi_{1,1}^{Z,W}(\sewself{1}{2} P, D).
		\label{eq:difference_cocycle_torus}
	\end{equation}
	We now perform the following construction (see Figure~\ref{fig:moduli_construction}): 
		\begin{enumerate}
		\item
			Find a geodesic $\vartheta$ in the homotopy class of the seam of $\sewself{1}{2}(P \sew{3}{1} D)$ in a flat metric.
		\item
			Unravel $\vartheta$ to obtain an annulus isomorphic to $\A_\tau \diffActing{1} \rotation_\alpha$ for some $\tau > 0$ and $\alpha \in [0, 2\pi)$.
		\item
			$\vartheta$ is unique up to translation and rotation, but every choice yields the same $\tau$ and $\alpha$.
	\end{enumerate}
		By Remark~\ref{remark:geodesic_annuli} and flat modular invariance of $\RMFD$, modular invariance of the function 
	\begin{equation}
		\rho(\tau, \alpha) = \Xi^{Z,W}_{1,2}(\A_\tau \diffActing{1} \rotation_\alpha), \qquad \tau > 0, \alpha \in [0, 2\pi),
	\end{equation}
	is equivalent to flat modular invariance of $\RMFE$.
	Observe that $\Xi_{1,2}^{Z,W}(P \sew{3}{1} D) = \rho(\tau, \alpha)$. 
		
	For a \gls{fsmooth} curve $D(t) \in \curves{\disks}$ rooted at $D(0) = D$, let $\tau(t)$ and $\alpha(t)$ denote the smooth functions resulting from the construction above.
	By \gls{fsmooth}ness of $Z$ and $W$, we may differentiate Equation~\eqref{eq:difference_cocycle_torus}, to obtain	
	\begin{equation}
\partial_t \, \rho(\tau(t), \alpha(t)) 
=		\partial_t \,  \Xi_{1,2}^{Z,W}(P \sew{3}{1} D(t)) 
		= \partial_t \,  \Xi_{1,1}^{Z,W}(\sewself{1}{2} P, D(t)) ,
		\label{eq:derivative_rho}
	\end{equation}
	which 
	is modular invariant, because $\sewself{1}{2} P$ is independent of the seam.
	
	Note that any pair $(\tau, \alpha) \in \Rp \times [0, 2\pi)$ may be obtained from our construction by cutting a disk from $\A_\tau \diffActing{1} \rotation_\alpha$. 
	(See also Figure~\ref{fig:moduli_construction}.)
	Moreover, the deformations $D(t)$ include Schiffer variations of the Teichm\"uller space of tori covering $\tori$ given by the marked seam; see~\cite{Nag:Complex_analytic_theory_of_Teichmuller_spaces} and~\cite[Example~4.13]{Maibach-Peltola:Complex_deformations_of_the_circle}.
	Since the Teichm\"uller space is one-dimensional $\C$, 
	Schiffer variation in a single point is sufficient to provide a local coordinate on the pair $(\tau, \alpha)$.
	Thus, we may integrate Equation~\eqref{eq:derivative_rho}, concluding that $\rho(\tau, \alpha)$ is modular invariant up to a constant, that is, $\rho(F(\tau, \alpha)) = \rho(\tau, \alpha) + \chi(F)$ for all ${F \in \mobZ}$.
	By applying successive modular transformations, we find that ${\chi \colon \mobZ \to \R}$ is a group homomorphism.
	However, the only such group homomorphism is $\chi = 0$, 
	so $\rho(\tau, \alpha)$ is modular invariant.
	This implies flat modular invariance of $\RMFE$.

\begin{figure}[t]
\centering
\begin{tikzpicture}[
	midarrow/.style={
		decoration={
			markings,
			mark=at position 0.5 with {\arrow[xshift=1.25pt]{>[length=2.5pt]}}
		},
		postaction={decorate}
	},
	mmidarrow/.style={
		decoration={
			markings,
			mark=at position 0.5 with {\arrow[xshift=2.5pt]{>[length=2.5pt]>[length=2.5pt]}}
		},
		postaction={decorate}
	}
]

\def\A{(1.9919, 3.0789).. controls (1.6871, 2.9711) and (1.3211, 3.3438) .. (1.0299, 3.2032).. controls (0.8739, 3.1279) and (0.8473, 2.9088) .. (0.7616, 2.7582).. controls (0.5767, 2.4334) and (0.0647, 2.1045) .. (0.2315, 1.77).. controls (0.3783, 1.4757) and (1.0268, 1.8949) .. (1.2066, 1.6195).. controls (1.3868, 1.3434) and (0.7671, 0.9373) .. (0.9514, 0.664).. controls (1.1996, 0.2959) and (1.861, 0.3357) .. (2.2733, 0.5004).. controls (2.4935, 0.5884) and (2.5011, 0.9717) .. (2.7249, 1.0501).. controls (2.9627, 1.1335) and (3.328, 0.6812) .. (3.4644, 0.8931).. controls (3.7007, 1.2603) and (2.5232, 1.5413) .. (2.6856, 1.9467).. controls (2.8182, 2.2776) and (3.617, 1.7951) .. (3.7393, 2.13).. controls (3.8614, 2.4645) and (3.3004, 2.6928) .. (3.0521, 2.948).. controls (2.8527, 3.1529) and (2.6964, 3.5291) .. (2.4108, 3.5174).. controls (2.2088, 3.509) and (2.1825, 3.1463) .. (1.9919, 3.0789) -- cycle}
	

	\begin{scope}[scale = 3]

	\node[circle] (pic_0_0) at (0, 0) {};
    \begin{scope}[shift=(pic_0_0)]	
		\path[draw, pattern=north west lines]
			(0,0) circle (0.2);
    \end{scope}

	\node[circle] (pic_1_0) at (0.6, 0) {};
    \begin{scope}[shift=(pic_1_0)]	
		\path[draw]
			(0,0) circle (0.07 and 0.2);
		\begin{scope}[shift={(0.6,0.2)}, rotate=45]
			\path[draw, densely dashed]
				(0,0) circle (0.16 and 0.02);
			\coordinate (PUA) at (0.16, 0);
			\coordinate (PUB) at (-0.16, 0);
		\end{scope}
		\begin{scope}[shift={(0.6,-0.2)}, rotate=-45]
			\path[draw, densely dashed]
				(0,0) circle (0.16 and 0.02);
			\coordinate (PBA) at (0.16, 0);
			\coordinate (PBB) at (-0.16, 0);
		\end{scope}
		\path[draw]
			(0, 0.2) to[in = 135, out = 0] (PUA)
			(0, -0.2) to[in = 225, out = 0] (PBA)
			(PBB) to[out=225, in=-90] (0.3, 0) to[in = 135, out = 90] (PUB);

    \end{scope}

	\node[circle] (pic_2_0) at (2, -0.55) {};
    \begin{scope}[shift=(pic_2_0)]	
	
	\coordinate (P1) at (0, 0);   
	\coordinate (P2) at (1, 0);   
	\coordinate (P3) at (1.2, 1.1);   
	\coordinate (P4) at (0.2, 1.1);   

	\draw[midarrow]
		(P1) -- (P2);
	\draw[mmidarrow]
		(P2) -- (P3);
	\draw[midarrow]
		(P4) -- (P3);
	\draw[mmidarrow]
		(P1) -- (P4);

	\draw[|-|]
		(0,-0.05) -- node[below] {$1$} (1, -0.05);
	\draw[|-|]
		(0,1.15) -- node[above] {$\alpha$} (0.2, 1.15);
	\draw[|-|]
		(1.25, 0) -- node[right] {$\tau$} (1.25, 1.1);

	\begin{scope}[rotate = 90, shift = {(-0.55, -1.048)}]
	\draw[densely dashed]
		(0.8, 0) to[in = 270, out = 90] (0.6, 0.1) to[in = 270, out = 90] (0.9, 0.15) to[in = 270, out = 90] (0.7, 0.4) to[in = 270, out = 90] (1, 0.25) to[in = 270, out = 90] (0.85, 0.6) to[in = 295, out = 90] (0.7, 0.6) to[in = 225, out = 105] (0.7, 0.8) to[in = 180, out = 45] (1.05, 0.7) to[in = 270, out = 0] (0.8, 1.0);
	\end{scope}
	
	\begin{scope}[shift={(0.4, 0.45)}, scale=0.15]
	\path[draw, pattern=north west lines]
		\A;
	\end{scope}

	\end{scope}

	\node[]
		() at (0, -0.3) {$D$};
	\node[]
		() at (0.35, 0) {$\sew{1}{3}$};
	\node[]
		() at (0.8, 0) {$P$};
	\node[]
		() at (1.25, 0) {$\sew{1}{2}$};

	\node[]
		() at (1.7, 0) {$\longrightarrow$};

	\end{scope}
\end{tikzpicture}
\caption{
	A disk $D \in \disks$ and a pair of pants $P \in \pants$ are sewn to form a torus $T = \sewself{1}{2} (D \sew{1}{3} P) \in \tori$.
	The right-hand side depicts the torus as a quotient of the upper half-plane such that the seam (dashed) winds once horizontally.
This determines the modulus $\tau > 0$ of an annulus $\A_\tau$ and a rotation angle $\alpha \in [0, 2\pi)$ such that ${T = \sewself{1}{2} (\A_\tau \diffActing{1} \rotation_\alpha)}$.
	Note that deforming the disk (hatched) may change $\tau$ and $\alpha$.
}
\label{fig:moduli_construction}
\end{figure}

	Coming back to Equation~\eqref{eq:difference_cocycle_torus}, we also find that modular invariance of $\rho(\tau, \alpha)$ and $\Xi_{1,1}^{Z,W}(\sewself{1}{2}P, D)$ implies modular invariance of $\Xi_{1,2}^{Z,W}(P)$, which is equivalent to hyperbolic modular invariance of $\RMFE$ (Definition~\ref{def:modular_invariant_hyperbolic}).
	This concludes the proof.
\end{proof}

\subsection{Tori and handles}

To extend our proof of Theorem~\ref{thm:universality} to the genus $1$ case, we use the flat modular invariance of the local \gls{fsmooth} $\Rp$-extensions $\RMFE$ and $\RMFD$ proven in Proposition~\ref{prop:E_prop}.
From the genus $0$ case (Theorem~\ref{thm:universality_genus_zero}), we already know that there exists an isomorphism $\Psi_{0, \boundaries}$ for $\boundaries \geq 0$.
Let $X$ be a reparametrization invariant trivialization of $\RMFE$ in genus $0$, which is the cylindrical trivialization over $\annuli$ (Proposition~\ref{prop:cylindrical}).
Since being cylindrical only involves the annulus-annulus cocycle with respect to $X$, and $\Psi$ is already an isomorphism in genus $0$, the trivialization $Y = \Psi \circ X$ on $\annuli$ is cylindrical, too.

For tori $T \in \tori$, we define
\begin{equation}
	X(T) = \sewxself{\RMFE}{1}{2} X(A) \qquad \textnormal{and} \qquad Y(T) = \sewxself{\RMFE}{1}{2} Y(A),
	\label{eq:tori_triv_geod}
\end{equation}
using any annulus $A \in \annuli$ such that $\sewself{1}{2} A = T$ and 
the seam in $T$ is a geodesic in a conformal flat metric.
The independence of the choice of $A$ is precisely the property of flat modular invariance of $\RMFE$ and $\RMFD$ (Definition~\ref{eq:flat_modular_invariance}).
The resulting cocycles $\Omega^X_{1,2}(A)$ and $\Omega^Y_{1,2}(A)$ for $A \in \annuli$ have the property that $\Omega^X_{1,2}(A) = \Omega^Y_{1,2}(A) = 0$ if the seam of $A$ in $\sewself{1}{2} A$ is geodesic.
We next show that $\Omega^X_{1,2}(A) = \Omega^Y_{1,2}$(A) for general annuli $A \in \annuli$.
\begin{lemma}
\label{prop:tori_E_D}
Let $\Psi_{1,0}(X(T)) = Y(T)$ for $T \in \tori$.
The following diagram commutes:
\label{lem:isom_tori}
\begin{equation}
\label{diag:annuli_tori}
\begin{tikzcd}
\RMFE(\annuli) & \RMFD(\annuli) \\
\RMFE(\tori) & \RMFD(\tori).
\arrow["{\Psi_{1,0}}", from=2-1, to=2-2]
\arrow["{\Psi_{0,2}}", from=1-1, to=1-2]
\arrow["{\sewxself{\RMFE}{1}{2}}", from=1-1, to=2-1]
\arrow["{\sewxself{\RMFD}{1}{2}}", from=1-2, to=2-2]
\end{tikzcd}
\end{equation}
\end{lemma}
\begin{proof}
Consider two annuli $A, B \in \annuli$ such that $\sewself{1}{2} A = \sewself{1}{2} B$ are equivalent tori.
First, let us assume that the seams in $\sewself{1}{2} A = \sewself{1}{2} B$ are homotopic and disjoint. 
In this case, we can decompose $A$ and $B$ in such a way that $A = \tilde{A} \sew{1}{2} \tilde{B}$ and $B = \tilde{B} \sew{1}{2} \tilde{A}$ for annuli $\tilde{A}, \tilde{B} \in \annuli$ defined by unraveling (cutting) the torus at both seams.
By the cocycle identity~\eqref{eq:cycle_sewing} 
in the case of two annuli $\Sigma_1=\tilde{A}$ and $\Sigma_2=\tilde{B}$,
we have
\begin{equation}
	\Omega^X_{1, 2}(B) + \Omega^X_{2,1}(\tilde{A}, \tilde{B})
	= \Omega^X_{2, 1}(A) + \Omega^X_{1,2}(\tilde{A}, \tilde{B}),
\end{equation} 
and analogously for $Y$.
Because the annulus-annulus cocycles of $X$ and $Y$ agree, taking the difference between the identities with $X$ and $Y$ and using the symmetries $\Omega^X_{2, 1}(A)=\Omega^X_{1,2}(A)$ and $\Omega^Y_{2, 1}(A)=\Omega^Y_{1,2}(A)$ results in
\begin{equation}
\Omega^X_{1, 2}(B) - \Omega^Y_{1, 2}(B)
= \Omega^X_{1,2}(A) - \Omega^Y_{1,2}(A) .
\label{eq:annuli_tori_compat}
\end{equation}
The left and right-hand sides of Equation~\eqref{eq:annuli_tori_compat} are precisely the factors picked up by mapping from the fiber over the torus 
$\RMFE(\sewself{1}{2} A) = \RMFE(\sewself{1}{2} B)$ to $\RMFD(\sewself{1}{2} A) = \RMFD(\sewself{1}{2} B)$ in Diagram~\eqref{diag:annuli_tori}.
This makes the diagram invariant under changing a lift from $\RMFE(\tori)$ to $\RMFE(\annuli)$ by homotopic and disjoint seams.

\begin{figure}[t]
\centering
\begin{tikzpicture}[
	midarrow/.style={
		decoration={
			markings,
			mark=at position 0.5 with {\arrow[xshift=1.25pt]{>[length=2.5pt]}}
		},
		postaction={decorate}
	},
	mmidarrow/.style={
		decoration={
			markings,
			mark=at position 0.5 with {\arrow[xshift=2.5pt]{>[length=2.5pt]>[length=2.5pt]}}
		},
		postaction={decorate}
	}
]


	\begin{scope}[scale = 3]
	
	\coordinate (P1) at (0, 0);   
	\coordinate (P2) at (1, 0);   
	\coordinate (P3) at (1.2, 1);   
	\coordinate (P4) at (0.2, 1);   

	\draw[midarrow]
		(P1) -- (P2);
	\draw[mmidarrow]
		(P2) -- (P3);
	\draw[midarrow]
		(P4) -- (P3);
	\draw[mmidarrow]
		(P1) -- (P4);

	\coordinate (PA) at (0.4, 0);
	\coordinate (PB) at (0.6, 1);
	
	\coordinate (PLA) at (0.06, 0.3);
	\coordinate (PRA) at (1.06, 0.3);
	
	\coordinate (PLB) at (0.13, 0.65);
	\coordinate (PRB) at (1.13, 0.65);

	\coordinate (PM) at (0.63, 0.45);

	\draw[]
		(PA) to[in = 180, out = 90] (PRA)
		(PLA) to[in = 270, out = 0] (PM)
		(PM) to[in = 0, out = 90] (PLB)
		(PRB) to[in = 270, out = 180] (PB);

	\coordinate (QA) at (0.35, 0);
	\coordinate (QB) at (0.55, 1);

	\coordinate (QLA) at (0.07, 0.35);
	\coordinate (QRA) at (1.07, 0.35);
	
	\coordinate (QLB) at (0.12, 0.6);
	\coordinate (QRB) at (1.12, 0.6);

	\coordinate (QM) at (0.25, 0.45);

	\draw[densely dashed]
		(QA) to[in = 180, out = 90] (QRA)
		(QLA) to[in = 270, out = 0] (QM)
		(QM) to[in = 0, out = 90] (QLB)
		(QRB) to[in = 270, out = 180] (QB);

	\coordinate (A) at (0.28, 0);
	\coordinate (B) at (0.48, 1);

	\draw[densely dotted]
		(A) -- (B);

	\end{scope}
\end{tikzpicture}
\caption{
	The self-sewing cocycles of annuli which form the same torus are related by a cocycle identity if the seams are disjoint.
	Passing to a geodesic seam might take multiple steps.
	In this figure, the initial seam (solid) is disjoint from the intermediate seam (dashed), which, in turn, is disjoint from the geodesic (dotted).
}
\label{fig:torus_geodesic}
\end{figure}

In order to show that Equation~\eqref{eq:annuli_tori_compat} holds for any annuli such that $\sewself{1}{2} A = \sewself{1}{2} B$, 
using the already known case, for any $A \in \annuli$, we can find a finite sequence of annuli $A = A_1, A_2, \dots, A_n \in \annuli$ in which we apply~\eqref{eq:annuli_tori_compat} to the pairs $A_j, A_{j + 1}$ ending with $A_n$ having geodesic seam in $\sewself{1}{2} A$; see Figure~\ref{fig:torus_geodesic}.
With this method, we can reduce the proof of~\eqref{eq:annuli_tori_compat} for arbitrary $A$ and $B$ to the case of annuli with geodesic seam in the torus, where the cocycles in~\eqref{eq:annuli_tori_compat} vanish by flat modular invariance of $\RMFE$ and $\RMFD$.
Therefore, Equation~\eqref{eq:annuli_tori_compat} holds trivially for these annuli, and thus, Diagram~\eqref{diag:annuli_tori} commutes.
\end{proof}

While the trivializations $X$ and $Y$ over $\annuli$ are so far constrained to be cylindrical, the trivializations over $\pants$ may be assumed to be defined by Proposition~\ref{prop:induction_trivialization} without loss of generality\footnote{
In Section~\ref{section:iso_higher_genus}, we replace the trivializations over $\pants$ by hyperbolic modular invariant trivializations.}.
Moreover, Equation~\eqref{eq:tori_triv_geod} defines $X$ and $Y$ over $\tori$.
Then, $X$ and $Y$ extend to the case of handles in $\moduli{1}{1}$ via Proposition~\ref{prop:induction_trivialization}, which applies to the $\boundaries = 1$ case.
We next show that this particular extension to handles also extends the isomorphism $\Psi$, that is, we show that all the cocycles involving handles agree for $X$ and $Y$.

\begin{proposition}
	Assume that $X$ and $Y$ are reparametrization invariant trivializations of respectively $\RMFE$ and $\RMFD$ in genus $0$ such that on $\annuli$ they are
	cylindrical, on $\pants$ they are defined by Proposition~\ref{prop:induction_trivialization}, and on $\tori$ they are defined by Equation~\eqref{eq:tori_triv_geod}.
	We extend $X$ and $Y$ to $\handles$ using Proposition~\ref{prop:induction_trivialization}.
	Then, defining
	\begin{equation}		
		\Psi_{1,1}(X(H)) = Y(H) , \qquad H \in \handles ,
	\end{equation}	
	extends the isomorphism $\Psi$ to $\handles$.
	\label{prop:isom_handles}
\end{proposition}
\begin{proof}
	We prove the equality of the cocycles for $X$ and $Y$, beginning with the self-sewing cocycles $\Omega^X_{j,k}(P)$ and $\Omega^Y_{j,k}(P)$ for pairs of pants $P \in \pants$.
	First, consider the case where $P \in \hpants$ is a pair of pants such that two boundary components, labeled $j$ and $k$, are of equal length in the uniform metric.
	Sewing a disk $P \sew{l}{1} \disk$ to the remaining boundary component $l \neq j, k$ yields an annulus.
	By appropriately scaling the round metric on the disk, the annulus comes equipped with a smooth conformal metric, which is the union of the uniform metric on the pair of pants and the round metric on the disk.
	Moreover, since the $j$th and~$k$th boundary components have equal length, this metric also induces a smooth conformal metric on the torus $\sewself{j}{k} (P \sew {j}{1} \disk)$, and the seam of these boundary components is still a geodesic in the torus.
	Since geodesics on surfaces are preserved under conformal transformations up to reparametrization, this seam may be reparametrized such that it is a geodesic in the flat conformal metric on the torus.
	By reparametrization invariance, the cocycle $\Omega^X_{j, k}(P \sew{j}{1} \disk)$ is unchanged after reparametrization the seam to become a geodesic, and with the geodesic seam, the cocycle vanishes by Equation~\eqref{eq:tori_triv_geod}.

Associated with the above sewing operations, the cocycle identity~\eqref{eq:basic_self_cocycle_idenity} with
$\Sigma_1=P$ and $\Sigma_2=\disk$
gives 
\begin{equation}
\Omega^X_{j, k}(P) 
+ \underbrace{\Omega^X_{l, 1}(\sewself{j}{k} P, \disk)}_{= \; 0 \, \textnormal{by~\eqref{eq:z_linear}}}
= \underbrace{\Omega^X_{j, k}(P \sew{l}{1} \disk)}_{= \; 0 \, \textnormal{by~\eqref{eq:tori_triv_geod}}} 
+ \Omega^X_{l, 1}(P, \disk),
\end{equation}
and analogously for $Y$.
Note that~\eqref{eq:z_linear} Proposition~\ref{prop:induction_trivialization} applies here: After the sewing of the $j$th and $k$th boundary components of equal lengths, the surface is still hyperbolic, hence also Fuchsian projective (by Item~\ref{prop:fprojective:hyperbolic} of Proposition~\ref{prop:fprojective}), and moreover, the $l$th boundary component becomes the first (and only) boundary component.
For the final term, we have $\Omega^X_{l, 1}(P, \disk) = \Omega^Y_{l, 1}(P, \disk)$ by the genus $0$ case, and thus, we see that
\begin{equation}
	\Omega^X_{j, k}(P) = \Omega^Y_{j, k}(P), \qquad P \in \hpants.
	\label{eq:cocycles_agree_hpants}
\end{equation}

Next, we consider any pair of pants $P \in \pants$ and an annulus $A \in \annuli$. 
Consider the cocycle identity~\eqref{eq:cycle_sewing} of sewing the annulus between two legs $j$ and $k$ of the pants,
with $\Sigma_1=P$ and $\Sigma_2=A$,
\begin{equation}
\Omega^X_{j, 2}(P \sew{k}{1} A) + \Omega^X_{k,1}(P, A)
= \Omega^X_{k, 1}(P \sew{j}{2} A) + \Omega^X_{j,2}(P, A),
\end{equation}
and analogously for $Y$.
Note that the pants-annulus cocycles agree for $X$ and $Y$ since $\Psi$ is an isomorphism in genus zero. 
Taking the difference between the identities with $X$ and $Y$ shows that the cocycles for $P \sew{l}{1} A$ agree if and only if they agree for $P \sew{j}{2} A$:
\begin{equation}
\Omega^X_{j,2}(P \sew{l}{1} A)
- \Omega^Y_{j,2}(P \sew{l}{1} A)
= \Omega^X_{l,1}(P \sew{j}{2} A)
- \Omega^Y_{l,1}(P \sew{j}{2} A).
\end{equation}
	Similarly as in the proof of Proposition~\ref{prop:tori_E_D}, we may consider the self-sewing cocycle $\Omega^X_{j,k}(P)$ of a pair of pants $P \in \pants$ and change the seam a finite number of times to obtain a seam which is geodesic in the handle $\sewself{j}{k} \check{P} \diffActing{l} \phi = \sewself{j}{k} P$ where $\check{P} \in \hpants$ has two boundary components of equal length in the  
	uniform metric, and  $\phi \in \Diffpan$.
	Then, by reparametrization invariance and 
	the first result~\eqref{eq:cocycles_agree_hpants} derived above, 
	we have $\Omega^X_{j,k}(P) = \Omega^X_{j,k}(\check{P}) = \Omega^Y_{j,k}(\check{P}) = \Omega^Y_{j,k}(P)$.

Lastly, the other cocycles involve sewing a disk $D \in \disks$ or an annulus $A \in \annuli$ to a handle $H = \sewself{2}{3} P \in \moduli{1}{1}$ for $P \in \pants$.
Consider the cocycle identity~\eqref{eq:basic_self_cocycle_idenity} 
with $\Sigma_1=P$ and $\Sigma_2=D$,
\begin{equation}
\Omega^X_{2,3}(P) + \Omega^X_{1, 1}(\sewself{2}{3} P, D)
= \Omega^X_{2,3}(P \sew{1}{1} D) + \Omega^X_{1, 1}(P, D),
\end{equation}
and with $\Sigma_1=P$ and $\Sigma_2=A$,
\begin{equation}
\Omega^X_{2,3}(P) + \Omega^X_{1, 1}(\sewself{2}{3} P, A)
= \Omega^X_{2,3}(P \sew{1}{1} A) + \Omega^X_{1, 1}(P, A),
\end{equation}
which combine these sewing operations with self-sewing of the form $H = \sewself{2}{3} P$. 
We thus express the new cocycles $\Omega^X_{1, 1}(H, D) = \Omega^X_{1, 1}(\sewself{2}{3} P, D)$
and $\Omega^X_{1, 1}(H, A) = \Omega^X_{1, 1}(\sewself{2}{3} P, A)$ 
in terms of those which are already known to agree for $X$ and $Y$.
Doing the same for $Y$ implies that they agree as well:
$\Omega^X_{1, 1}(H, D) = \Omega^Y_{1, 1}(H, D)$ and $\Omega^X_{1, 1}(H, A) = \Omega^Y_{1, 1}(H, A)$. 
\end{proof}

Now we have shown that $\RMFE$ and $\RMFD$ are isomorphic in genus zero (Theorem~\ref{thm:universality_genus_zero}), and in genus one for $\boundaries = 0$ (Lemma~\ref{lem:isom_tori}) and for $\boundaries = 1$ (Proposition~\ref{prop:isom_handles}).
Thus, we can apply the induction step (Proposition~\ref{prop:induction_extension}) to conclude the proof of Theorem~\ref{thm:universality} up to genus one:
\begin{theorem}
	All {local} \gls{fsmooth} $\Rp$-extensions with equal central charges are isomorphic up to genus one.
	\label{thm:universality_genus_one}
\end{theorem}

\subsection{Higher genus}
\label{section:iso_higher_genus}

Our construction of the isomorphism for genus $\genus \geq 2$ uses crossing invariance (Definition~\ref{def:crossing}) 
and hyperbolic modular invariance (Definition~\ref{def:modular_invariant_hyperbolic}) as proven in Proposition~\ref{prop:E_prop}.
Given the isomorphism $\Psi$ in genus zero and one (Theorem~\ref{thm:universality_genus_one}), we use crossing invariant trivializations on pairs of pants instead of using Proposition~\ref{prop:induction_trivialization}.
In short, since crossing invariance and hyperbolic modular invariance respectively correspond to the invariance under A- and S-moves, as defined by Hatcher and Thurston~\cite{Hatcher:Pants_decompositions_of_surfaces, Hatcher-Thurston:A_presentation_for_the_mapping_class_group_of_a_closed_orientable_surface}, between pants decompositions of hyperbolic surfaces, the new trivializations extend canonically to any hyperbolic surface via pants decomposition (see Equation~\eqref{eq:triv_pants_hyp}).
The equality of cocycles then follows by reduction to the $\genus \in \{0,1\}$ cases.
This extends the isomorphisms $\Psi$ to any genus and number of boundary components, proving Theorem~\ref{thm:universality}.

\begin{remark}
The trivializations on $\pants$ defined by Proposition~\ref{prop:induction_trivialization} are most likely not crossing invariant, since they are defined to be invariant under sewing of a unit disk at the first boundary component.
Aside from this being an asymmetric definition with respect to permutations of the boundary labels, we may also consider what this condition means for the real determinant line bundle.
There, we equip the unit disk with the flat or round metric, the annulus with the flat metric, and the pair of pants with the uniform metric.
Then, sewing the unit disk to a pair of pants, we do not obtain the flat metric on the annulus, which would contradict the triviality of the disk-pants cocycle.
\end{remark}

So far, the isomorphism $\Psi$ is defined in genus $0$ and $1$ for any number of boundary components.
Let $X$ be a reparametrization invariant trivialization of $\RMFE$ over $\spheres$, $\disks$, $\annuli$, and $\tori$, that is, all the surfaces with non-positive Euler characteristic.
On pairs of pants in $\pants$, we take $X$ to be a crossing invariant as in Definition~\ref{def:crossing}.
The isomorphism $\Psi$ defines the reparametrization invariant trivialization  $Y(\Sigma) = \Psi(X(\Sigma))$ of $\RMFD$ for these moduli spaces.
Since $\Psi_{0,3}$ is compatible with the genus zero operation of sewing two pairs of pants, the trivialization $Y$ of $\RMFD(\pants)$ is also crossing invariant.
Let $\Sigma \in \hmoduligb$ with $\chi(\Sigma) = 2 - 2 \genus - \boundaries < 0$ be a hyperbolic surface and $\Sigma = \sewall \underline{P}$ a pants decomposition obtained by unraveling hyperbolic geodesics.
We extend the trivializations by
\begin{equation}
\label{eq:triv_pants_hyp}
X(\Sigma) = \sewxall{\RMFE} X(\underline{P}) \qquad \textnormal{and} \qquad
Y(\Sigma) = \sewxall{\RMFD} Y(\underline{P}).
\end{equation}
This is independent of the pants decomposition, since it is possible to transition between pants decompositions using a finite number of elementary moves
--- namely, the S- and A-moves. 
Respectively, hyperbolic modular invariance (Definition~\ref{def:modular_invariant_hyperbolic}) and crossing invariance (Definition~\ref{def:crossing}) precisely say that~\eqref{eq:triv_pants_hyp} is invariant under these moves.
For surfaces with general boundary parametrizations, $X$ and $Y$ are defined by reparametrization invariance.
Now, since $Y(P) = \Psi_{0,3}(X(P))$ holds by definition, we have
\begin{align*}
\Psi(X(\Sigma)) = \sewxall{\RMFD} \Psi(X(\underline{P})) = \sewxall{\RMFD} Y(\underline{P}) = Y(\Sigma),
\qquad \Sigma \in \moduligb ,
\end{align*}
for surfaces of genus $\genus = 0$ and $\boundaries \geq 2$, or of genus $\genus = 1$ and $\boundaries \geq 1$: 
the new trivialization $X$ is mapped to the trivialization $Y$ where the isomorphism $\Psi$ was already defined.

To prove the theorem, we have to verify that a number of cocycles agree.
For hyperbolic surfaces $\Sigma_1$ and $\Sigma_2$ such that $\partial_j \Sigma_1$ and $\partial_k \Sigma_2$ have the same boundary length, we clearly have $\Omega^X_{j,k}(\Sigma_1, \Sigma_2) = \Omega^Y_{j,k}(\Sigma_1, \Sigma_2) = 0$ and also for self-sewing.
By reparametrization invariance, the same holds for general boundary parametrizations away from the seam.
If there is another parametrization at the seam, however, some more arguments are needed.
Let us proceed by induction on the genus $\genus \geq 2$, relying on the base cases $\genus = 0$ and $\genus = 1$.

Given any two surfaces $\Sigma_1 \in \moduli{\genus_1}{\boundaries_1}$ and $\Sigma_2 \in \moduli{\genus_2}{\boundaries_2}$, in the cocycle $\Omega^X_{j,k}(\Sigma_1, \Sigma_2)$, we open up a seam in one of the surfaces, say $\Sigma_1 = \sewself{l}{m} \smash{\hat \Sigma}$, such that $\smash{\hat \Sigma}$ has genus one less than $\Sigma_1$.
Note that, as we begin the induction knowing genera $0$ and $1$, we can assume without loss of generality that $\genus_1 + \genus_2 \geq 2$ and $\genus_1 \geq 1$.
Thus, the sewn surface $\Sigma_1 \sew{j}{k} \Sigma_2$ has a uniform metric, so we can assume that the new seam is a hyperbolic geodesic in the total surface $\Sigma_1 \sew{j}{k} \Sigma_2$.
Considering the cocycle identity~\eqref{eq:basic_self_cocycle_idenity} with $\smash{\hat \Sigma}$ and $\Sigma_2$,
\begin{equation}
\Omega^X_{l, m}(\hat \Sigma) 
+ \underbrace{\Omega^X_{j, k}(\sewself{l}{m} \hat \Sigma, \Sigma_2)}_{= \; \Omega^X_{j, k}(\Sigma_1, \Sigma_2)}
= \underbrace{\Omega^X_{l, m}(\hat \Sigma \sew{j}{k} \Sigma_2)}_{= \; 0} 
+ \Omega^X_{j, k}(\hat \Sigma, \Sigma_2),
\end{equation}
and analogously for $Y$, we see that the third term vanishes, since we assumed the seam to be a hyperbolic geodesic.
We also note that the first and last cocycles involve surfaces of a total genus at most $\genus_1 + \genus_2 - 1$, for which the cocycles of $X$ and $Y$ agree by the induction hypothesis.
We conclude that $\Omega^X_{j,k}(\Sigma_1, \Sigma_2) = \Omega^Y_{j,k}(\Sigma_1, \Sigma_2)$.

Finally, let $\Sigma \in \moduligb$ be any surface of genus $\genus \geq 1$ and number $\boundaries \geq 2$ of boundary components.
We consider the self-sewing cocycle $\Omega^X_{j,k}(\Sigma)$.
Let $\Sigma = \sewself{l}{m} \smash{\hat \Sigma}$ such that the seam is a hyperbolic geodesic in the sewn surface $\sewself{j}{k} \Sigma$ of genus $\genus + 1 \geq 2$.
Such a seam must be non-separating in $\Sigma$ such that $\smash{\hat \Sigma} \in \moduli{\genus - 1}{\boundaries}$ is of genus one lower.
Consider the cocycle identity~\eqref{eq:cycle_selfsewing} with $\smash{\hat \Sigma}$, 
\begin{equation}
\Omega^X_{l, m}(\hat \Sigma) 
+ \underbrace{\Omega^X_{j, k}(\sewself{l}{m} \hat \Sigma)}_{= \; \Omega^X_{j, k}(\Sigma)}
= \underbrace{\Omega^X_{l, m}(\sewself{j}{k} \hat \Sigma)}_{= \; 0} 
+ \Omega^X_{j, k}(\hat \Sigma) ,
\end{equation}
where the third term again vanishes, since the seam is a hyperbolic geodesic.
Moreover, the first and third terms are cocycles for one genus lower, so by the induction hypothesis, we know that they agree for $X$ and $Y$. 
It again follows that $\Omega^X_{j,k}(\Sigma) = \Omega^Y_{j,k}(\Sigma)$.

\noindent
This completes the proof of Theorem~\ref{thm:universality}.

\section{Restriction covariant loop measures}
\label{section:loop_measures}
A decade ago, S.~Benoist posed the question~\cite{Benoist:Classifying_conformally_invariant_loop_measures} 
which ``restriction functions'' can consistently appear in the Radon--Nikodym derivative~\eqref{eq:restriction_blm} for the restriction of Malliavin--Kontsevich--Suhov (MKS) loop measures to domains $D \subset \C$, replacing the normalized Brownian loop measure (BLM) term $\Lambda^*(\loopzero, \partial D)$ by a generic function $f(\loopzero, D, \C)$ with natural properties (see Items~\ref{def:restriction_function:b}~\&~\ref{def:restriction_function:c} in Definition~\ref{def:restriction_function} below).
This generalized setup also extends very naturally to the setting of loop measures on Riemann surfaces~\cite{Kontsevich:CFT_SLE_and_phase_boundaries, Kontsevich-Suhov:On_Malliavin_measures_SLE_and_CFT, Werner:The_conformally_invariant_measure_on_self-avoiding_loops, LeJan:Markov_paths_loops_and_fields, LeJan:Brownian_loops_topology, Zhan:SLE_loop_measures}.

The purpose of this section is to prove a classification result answering Benoist's question (Theorem~\ref{thm:restriction_function}). 
In the proof, we relate restriction functions to local \gls{fsmooth} $\Rp$-extensions (Proposition~\ref{prop:restriction_function_to_modular_functor}), and apply the universal property (Theorem~\ref{thm:universality}).
Our result holds under a further natural symmetry assumption on the restriction functions (Item~\ref{def:restriction_function:d} in Definition~\ref{def:restriction_function}), and a \gls{fsmooth}ness assumption on analytic configurations (Definition~\ref{def:restriction_function_frsmooth}), which can both be justified by the properties of BLM, and which are essential to prove the link to $\Rp$-extensions.

\subsection{Conformally invariant restriction functions}

Consider a compact connected Riemann surface $\Sigma$ (without boundary paramerization), and denote by $\cloops{\Sigma}$ the space of \emph{(unparametrized) simple continuous loops} in $\Sigma$, that is, the set of images of continuous injective maps from $S^1$ to $\Sigma$, equipped with the topology induced by the Hausdorff distance with respect to any conformal Riemannian metric\footnote{
	The Hausdorff distance is defined as $d_g(K_1, K_2) = \inf \setsuchthat{r > 0}{K_1 \subset B_{g}(K_2,r) \textnormal{ and } K_2 \subset B_{g}(K_1,r)}$, where $K_1, K_2 \subset \Sigma$ are non-empty compact subsets and $B_{g}(K,r) = \cup_{x \in K} B_{g}(x,r)$ is the $r$-neighborhood of $K$ in the metric $g$. 
	On compact surfaces, all conformal metrics are bi-Lipschitz equivalent, which implies that the Hausdorff topology induced by $d_g$ does not depend on the choice of conformal metric $g$.
}. 

\begin{definition}
	\leavevmode
	\begin{enumerate}
	\item \label{item:restriction_function}
		A \emph{restriction function} is a family of continuous functions $f(\blank , \Sigma_1, \Sigma_2) \colon \cloops{\Sigma_1} \to \R$, indexed by Riemann surfaces $\Sigma_1 \subset \Sigma_2$, and satisfying the following properties:
		\begin{enumerate}
			\item \emph{Conformal invariance.}
				We have
				\label{def:restriction_function:b}
				\begin{equation}
					f(F(\loopzero), F(\Sigma_1), F(\Sigma_2)) = f(\loopzero, \Sigma_1, \Sigma_2), \qquad \loopzero \in \cloops{\Sigma_1},
					\label{eq:restriction_function_invariance}
				\end{equation}
				for any biholomorphism $F \colon \Sigma_2 \to \tilde{\Sigma}_2$ to another Riemann surface $\tilde{\Sigma}_2$.
			\item \emph{Cocycle property.}
				For three Riemann surfaces $\Sigma_1 \subset \Sigma_2 \subset \Sigma_3$, we have
				\label{def:restriction_function:c}
				\begin{equation}
					f(\loopzero, \Sigma_1, \Sigma_2) + f(\loopzero, \Sigma_2, \Sigma_3) = f(\loopzero, \Sigma_1, \Sigma_3), \qquad \loopzero \in \cloops{\Sigma_1}.
					\label{eq:restriction_function_cocycle}
				\end{equation}
			\item \emph{Symmetry.}
				For two disjoint loops $\loopone, \looptwo \in \cloops{\Sigma}$, we have
				\label{def:restriction_function:d}
				\begin{equation}
					f(\loopone, \Sigma \setminus \looptwo, \Sigma)
					= f(\looptwo, \Sigma \setminus \loopone, \Sigma).
					\label{eq:restriction_function_symmetry}
				\end{equation}
In the case of a separating loop, the restriction function depends only on the connected component containing the other loop. 
		\end{enumerate}
	\item \label{item:restriction_function_equiv}
		Restriction functions $f_1$ and $f_2$ are \emph{equivalent} if there exists a conformally invariant family of continuous functions $g(\blank, \Sigma) \colon \cloops{\Sigma} \to \R$, indexed by Riemann surfaces $\Sigma$, such that
		\begin{equation}
			f_2(\loopzero, \Sigma_1, \Sigma_2) = f_1(\loopzero, \Sigma_1, \Sigma_2) + g(\loopzero, \Sigma_2) - g(\loopzero, \Sigma_1), \qquad \loopzero \in \cloops{\Sigma_1}.
			\label{eq:restrictoin_function_equiv}
		\end{equation}
		Such $g(\blank, \Sigma)$ are automatically symmetric in the sense that
		\begin{equation}
			g(\loopone, \Sigma \setminus \looptwo) - g(\loopone, \Sigma) = g(\looptwo, \Sigma \setminus \loopone) - g(\looptwo, \Sigma).
		\end{equation}
	\end{enumerate}
	\label{def:restriction_function}
\end{definition}

Compared to Benoist's setup~\cite[Section~2.2]{Benoist:Classifying_conformally_invariant_loop_measures}, we assume the additional symmetry~\ref{def:restriction_function:d}, which appears to be essential in the proof of our results.
In particular, in the proof of Proposition~\ref{prop:restriction_function_to_modular_functor}, the symmetry~\ref{def:restriction_function:d} leads to the cocycles~\eqref{eq:MKS_cocycles} of the $\Rp$-extension induced by the restriction function. 
It may be motivated by the normalized BLM, $f(\loopzero, D, \C) = \Lambda^*(\loopzero, \partial D)$, in a domain $D \subset \C$.
Namely,
the measure $\Lambda^*(\loopone, \partial D)$ for the domain $D$ bounded by $\looptwo$ and containing $\loopone$ (such that $\partial D = \looptwo$) is the normalized measure of Brownian loops hitting both loops $\loopone$ and $\looptwo$, which is symmetric under exchange of the loops.
Moreover, Brownian motion started anywhere on $\loopone$ stays in the same connected component as $\loopone$. 
Heuristically speaking, the symmetry property also guarantees the ``resampling property'' for multiple SLE measures (see, for instance,~\cite{Lawler:Partition_functions_loop_measure_and_versions_of_SLE,Lawler:Defining_SLE_multiply_connected} for the idea, 
the discussion in~\cite{HPW:Multiradial_SLE_with_spiral}, and more references therein; and~\cite[Section~2]{Luo-Maibach:Two-loop_Loewner_potentials} for the two-loop case).

\begin{remark}
	Expanding on the planar setup, the main example of a restriction measure on arbitrary Riemann surfaces should be the measure of Brownian loops in~$\Sigma_2$, which hit~$\loopzero$ and exit~$\Sigma_1$, with normalization accounting for the infinite mass due to increasingly long loops.
	While the definition of BLM immediately extends to Riemann surfaces 
	(see, e.g.,~\cite{Lawler-Werner:The_Brownian_loop_soup,
	Field-Lawler:Reversed_radial_SLE_and_the_Brownian_loop_measure, LeJan:Brownian_loops_topology, Xue-Wang:The_Brownian_loop_measure_on_Riemann_surfaces_and_applications_to_length_spectra, Hou-Wu:The_total_mass_of_Brownian_loop_measure_of_Riemann_surfaces_for_large_genus}), it is not clear that the normalization works as in the plane; see also \cite[Conjecture~5.2]{Zhan:SLE_loop_measures} and Remark~\ref{remark:BLM_normalization}. 
	\label{remark:BLM_conjecture}
\end{remark}

\begin{definition}	
	Given a restriction function $f$, an \emph{$f$-restriction covariant loop measure} is a family of $\sigma$-finite Borel measures $\nu_\Sigma$ on $\cloops{\Sigma}$, indexed by Riemann surfaces $\Sigma$, such that\label{def:f_loop_meas}
	\begin{enumerate}
		\item
			For a biholomorphism $F \colon \Sigma_1 \to \Sigma_2$ we have $F^* \nu_{\Sigma_2} = \nu_{\Sigma_1}$.
		\item
			 For subsurface $\Sigma_1 \subset \Sigma_2$, the measure $\nu_{\Sigma_1}$ is absolutely continuous with respect to $\nu_{\Sigma_2}$, and the Radon--Nikodym derivative is
			\begin{equation}
				\frac{\dd \nu_{\Sigma_1}}{\dd \nu_{\Sigma_2}}(\loopzero) = e^{f(\loopzero, \Sigma_1, \Sigma_2)} \boldone_{\loopzero \subset \Sigma_1}, \qquad \loopzero \in \cloops{\Sigma_2}.
			\end{equation}
	\end{enumerate}
\end{definition}

\subsection{Analytic loops and Segal moduli space}

Let us now explain how Segal moduli spaces $\moduligb$ (Section~\ref{section:setup}) relate to the setting of loop measures.
By considering the seam created by sewing $\Sigma_1 \sew{j}{k} \Sigma_2$, we find a simple analytic loop, which is parametrized by either one of the boundary parametrizations coming from $\Sigma_1$ or $\Sigma_2$.
A cocycle $\Omega^Z_{j,k}(\Sigma_1, \Sigma_2)$ of a $\Rp$-extension may be regarded as a function of the seam.
In the case of a reparametrization invariant trivialization (Proposition~\ref{prop:local_equiv}), the cocycle is a function on the seam up to reparametrization.
Thus, we consider the space $\anloops{\Sigma}$ of \emph{(unparametrized) simple analytic loops} in a Riemann surface.

We equip $\anloops{\Sigma}$ with a topology as follows.
First, for $n \geq 1$, consider the sets $\mathcal{O}(U_n, \Sigma)$ of holomorphic maps $U_n \to \Sigma$ with domains $U_n = \setsuchthat{z \in \C}{ 1 - 1/n < |z| < 1 + 1/n}$.
We equip $\mathcal{O}(U_n, \Sigma)$ with the topology of uniform convergence on compact subsets of $U_n$.
Next, consider the set $\mathcal{O}(S^1, \Sigma)$ of real-analytic functions $S^1 \to \Sigma$, with the inductive limit topology induced by the restrictions $\mathcal{O}(U_n, \Sigma) \to \mathcal{O}(S^1, \Sigma)$ for $n \geq 1$.
The space of parametrized simple analytic loops is a subset of $\mathcal{O}(S^1, \Sigma)$.
Hence, we can equip $\anloops{\Sigma}$ with the subspace topology quotiented by real-analytic reparametrizations of the loop.

\begin{lemma}
	The inclusion $\anloops{\Sigma} \subset \cloops{\Sigma}$ is continuous with dense image.
	\label{lemma:an_c_inclusion}
\end{lemma}
\begin{proof}
The inductive limit topology on $\anloops{\Sigma}$ is defined by quotienting by reparametrization.
	Thus, without loss of generality, we may work with para\-metrized loops.
	The inductive limit $\mathcal{O}(S^1, \Sigma)$ is regular~\cite[Theorem~8.4]{Kriegl-Michor:Convenient_setting_of_global_analysis}, whence for a convergent sequence $(\vartheta_n)_{n \in \mathbb{N}} \in \mathcal{O}(S^1, \Sigma)$, we can find $m \in \mathbb{N}$ such that $\vartheta_n$ extends to a holomorphic map on $U_m$ for all $n \geq m$.
	In particular, they converge in the Hausdorff distance. 
	To see that the inclusion is dense, note that any continuous simple loop (or simple closed curve) may be approximated by smooth curves~\cite[Section~1.2.2]{Farb-Margalit:Primer_on_mapping_class_groups}, and 
	the set of analytic embeddings of $S^1$ in $\Sigma$ is dense in the set of smooth embeddings by~\cite[Theorems~1.4~\&~5.1]{Hirsch:Differential_topology}.
\end{proof}
\noindent
The following result is a slight generalization of~\cite[Proposition~3.6]{Benoist:Classifying_conformally_invariant_loop_measures}\footnote{
	The remark on continuity by Benoist concerns continuity on $\cloops{\Sigma}$:
	If $\gonlyan$ extends continuously to $\cloops{\Sigma}$, it follows that $f$ is equivalent to the trivial restriction function by setting $g = \gonlyan$ Definition~\ref{def:restriction_function}.
}.

\begin{proposition}
	Given a restriction function $f$, there exists a family of continuous functions $\gonlyan(\blank, \Sigma) \colon \anloops{\Sigma} \to \R$, indexed by Riemann surfaces $\Sigma$, such that
	\label{prop:an_coboundary}
	\begin{equation}
		f(\loopzero, \Sigma_1, \Sigma_2) = \gonlyan(\loopzero, \Sigma_2) - \gonlyan(\loopzero, \Sigma_1), \qquad \loopzero \in \anloops{\Sigma_1}.
		\label{eq:an_coboundary}
	\end{equation}
\end{proposition}
\begin{proof}
	Let $\vartheta \colon S^1 \to \loopzero$ be a real-analytic parametrization of $\loopzero \in \anloops{\Sigma}$.
	There exists an annular neighborhood $S^1 \subset A \subset \hat{\C}$ such that $\vartheta$ is univalent on $A$.
	The function
	\begin{equation}
		\gonlyan(\loopzero, \Sigma) = f(\loopzero, \vartheta(A), \Sigma) - f(S^1, A, \hat{\C}) 
		\label{eq:def_g}
	\end{equation}
	is independent of $A$: for another annular neighborhood $B$ such that $S^1 \subset A \subset B \subset \hat{\C}$,
	using the cocycle property~\ref{def:restriction_function:c} and conformal invariance~\ref{def:restriction_function:b} of restriction functions (Definition~\ref{def:restriction_function}), we have
	\begin{equation}
	\begin{aligned}
		&\Big(f(\loopzero, \vartheta(A), \Sigma) - f(S^1, A, \hat{\C})\Big)
		- \Big(f(\loopzero, \vartheta(B), \Sigma) - f(S^1, B, \hat{\C})\Big)
		\\=\; & f(\loopzero, \vartheta(A), \vartheta(B)) - f(S^1, A, B) = 0.
	\end{aligned}
	\end{equation}
	Next, we show that $\gonlyan(\loopzero, \Sigma)$ is independent of the parametrization $\vartheta$. To this end, 
	we reparametrize $\vartheta$ by $\phi \in \Diffpan$.
	Take an annular neighborhood $S^1 \subset A \subset \hat{\C}$ such that $\phi$ and $\vartheta \circ \phi$ are univalent on $A$, and let $B = \phi(A)$.
	Then, we have
	\begin{equation}
	\begin{aligned}
		&\Big(f(\loopzero, \vartheta(\phi(A)), \Sigma) - f(S^1, \phi(A), \hat{\C})\Big)
		- \Big(f(\loopzero, \vartheta(B), \Sigma) - f(S^1, B, \hat{\C})\Big)
		\\=\; & f(\loopzero, \vartheta(\phi(A)), \vartheta(B)) - f(S^1, \phi(A), B) = 0.
	\end{aligned}
	\end{equation}
	Using the cocycle property~\ref{def:restriction_function:c} in Definition~\ref{def:restriction_function}, we see that the asserted Equation~\eqref{eq:an_coboundary} holds.
	Since the inclusion $\anloops{\Sigma} \subset \cloops{\Sigma}$ is continuous (Lemma~\ref{lemma:an_c_inclusion}), and in the topology of $\anloops{\Sigma}$ we can pick the annulus $A$ locally uniformly in $\loopzero$, it follows that $\gonlyan(\blank, \Sigma)$ is a continuous function on $\anloops{\Sigma}$.
\end{proof}

\begin{remark}
	One may think of Equation~\eqref{eq:def_g} as a renormalization of the measure of Brownian loops in $\Sigma$ hitting $\loopzero$, decomposed as loops exiting $\vartheta(A)$ plus those loops staying inside $\vartheta(A)$.
	By conformal invariance of BLM, the latter may be expressed as the difference between all Brownian loops in $\hat{\C}$ hitting $S^1$ and those which also exit $A$, that is, $f(S^1, A, \hat{\C})$.
	The renormalization amounts to subtracting the (infinite) measure of Brownian loops in $\hat{\C}$ hitting $S^1$, which then results in Equation~\eqref{eq:def_g}.
	Note that this renormalization only works for analytic loops $\loopzero \in \anloops{\Sigma}$ since it relies on the existence of the annular neighborhood $A$.
	A proper normalization would also have to deal with the increasingly long loops in $\Sigma$.
	\label{remark:BLM_normalization}
\end{remark}

Let us now turn to smoothness criteria for the restriction functions.
As usual for Fr\"olicher smoothness (Section~\ref{section:setup}), we formulate it in terms of \gls{fsmooth} curves in Segal moduli spaces, as follows. 

\begin{definition}	
	A restriction function $f$ is \emph{\gls{fsmooth} on analytic configurations} if the~map $t \mapsto f(\loopzero(t), \Sigma_1(t), \Sigma_2(t))$ is a smooth function for any configuration of the following form: 
	\begin{enumerate}[leftmargin=*]
		\item
			Let $S_1(t) \in \curves{\moduli{\genus_1}{\boundaries_1}}$ and $S_2(t) \in \curves{\moduli{\genus_2}{\boundaries_2}}$ be \gls{fsmooth} curves in Segal moduli spaces.
			Consider the surface $\Sigma_1 = S_1(t) \sew{1}{1} S_2(t)$.
		\item
			For times $t$, the curves $S_1(t)$ and $S_2(t)$ have representatives of the form~\cite[Equation~(4.25)]{Maibach-Peltola:Complex_deformations_of_the_circle}.
			Let $\zeta_1(t, z)$ denote the boundary parametrization of $\partial_1 S_1(t)$ in such a representative.
			Then, the image $\loopzero(t) = \zeta_1(t, S^1)$ an element of $\anloops{\Sigma_1}$.
		\item
			Let $\underline{Q} = (Q_1(1), \dots, Q_n(t))$ be 
			$n \geq 0$ additional \gls{fsmooth} curves in Segal moduli spaces, and let $\Sigma_2(t) = \Sigma_1(t) \sewall \underline{Q}(t)$ denote the surface obtained from sewing and self-sewing operations applied to $\Sigma_1, Q_1, \dots, Q_n$ in any way so that the result is connected.
	\end{enumerate}
	\label{def:restriction_function_frsmooth}
\end{definition}

In particular, we expect the \gls{fsmooth}ness in Definition~\ref{def:restriction_function_frsmooth} to hold for the BLM restriction function $\Lambda^*(\loopzero, D)$ because of its relation to zeta-regularized determinants of the Laplacian (for smooth $\loopzero$ and $D$ with smooth boundary)~\cite[Section~2]{Dubedat:SLE_and_free_field}, which, in turn, we expect to smoothly depend on the domain.
Any further regularization extending $\Lambda^*$ to higher genus should then also respect the smoothness.

\begin{remark}
	The natural setting for Definition~\ref{def:restriction_function_frsmooth} would involve ``comparison moduli spaces''~\cite{Schippers-Staubach:Comparison_moduli_spaces_of_Riemann_surfaces} of configurations $\loopzero \subset \Sigma_1 \subset \Sigma_2$, where one considers the triple up to biholomorphisms on the outermost surface $\Sigma_2$ preserving the additional marked structures $\loopzero$ and $\Sigma_1$ set-wise.
	Then, the subspace of the comparison moduli space where $\loopzero$ and $\Sigma_1$ embed analytically may be equipped with a Fr\"olicher structure similar to that of $\moduligb$, by using the action of complex deformations (see~\cite[Section~4.4]{Maibach-Peltola:Complex_deformations_of_the_circle}).
\end{remark}

\subsection{Local \texorpdfstring{$\Rp$}{R+}-extensions from restriction functions}

We will now show how restriction functions yield $\Rp$-extensions (Definition~\ref{def:modular_functor}).
In particular, the next result defines the notion of \emph{central charge} of a restriction function as the central charge of the associated $\Rp$-extension (Definition~\ref{def:central_charge}).
The main result of this section (Theorem~\ref{thm:restriction_function}) is proven in the end, using the classification of local \gls{fsmooth} $\Rp$-extensions (Theorem~\ref{thm:universality}).

\begin{proposition}
	Let $f$ be a restriction function which is {\gls{fsmooth} on analytic configurations}.
	Then, there exists a {local} \gls{fsmooth} $\Rp$-extension 
	$E$ with reparametrization invariant trivialization $Z$ such that
	\begin{equation}
		\Omega_{j, k}^Z(\Sigma_1, \Sigma_2) = \gonlyan(\loopzero, \Sigma_1 \sew{j}{k} \Sigma_2), \qquad
		\Omega_{j, k}^Z(\Sigma) = \gonlyan(\loopzero, \Sigma).
		\label{eq:MKS_cocycles}
	\end{equation}
	where $\gonlyan$ is given by Proposition~\ref{prop:an_coboundary} and $\loopzero$ is the seam  $\Sigma_1 \sew{j}{k} \Sigma_2$ or $\sewself{j}{k} \Sigma$ respectively.
	\label{prop:restriction_function_to_modular_functor}
\end{proposition}
\begin{proof}
	Let $\RMFE(\moduligb) = \Rp \times \moduligb$ be the trivial principal $\Rp$-bundle, and $Z(\Sigma) = (1, \Sigma)$.
	Then, the sewing isomorphisms defined by $Z(\Sigma_1) \sew{j}{k} Z(\Sigma_2) = e^{\gonlyan(\loopzero, \Sigma_1 \sew{j}{k} \Sigma_2)} Z(\Sigma_1 \sew{j}{k} \Sigma_2)$ yield the sewing cocycles in~\eqref{eq:MKS_cocycles}, 
	and similarly for the self-sewing cocycles.
	However, we need to check that the sewing isomorphisms satisfy the associativity identities in Definition~\ref{def:modular_functor}.
	Equivalently, we may check that the cocycle identities~(\ref{eq:basic_cocycle_idenity},~\ref{eq:basic_self_cocycle_idenity},~\ref{eq:cycle_sewing},~\ref{eq:cycle_selfsewing}) in Section~\ref{section:cocycles} hold.
	These, in turn, follow from the symmetry property~\ref{def:restriction_function:d} in Definition~\ref{def:restriction_function}, where the loops are the seams of respective sewing operations.
	For instance, the cocycle identity~\eqref{eq:basic_cocycle_idenity} is obtained as follows:
	\begin{equation}
	\begin{aligned}
		&\gonlyan\Big(\loopone, \Sigma_1 \sew{j}{k} \Sigma_2\Big)
		- \gonlyan\Big(\loopone, (\Sigma_1 \sew{j}{k} \Sigma_2) \sew{l}{m} \Sigma_3\Big)
		\\ = \; &
		f\Big(\loopone, (\Sigma_1 \sew{j}{k} \Sigma_2), (\Sigma_1 \sew{j}{k} \Sigma_2) \sew{l}{m} \Sigma_3\Big) 
		 && \textnormal{[by~\eqref{eq:an_coboundary}]}
		\\ = \; &
		f\Big(\looptwo, (\Sigma_2 \sew{l}{m} \Sigma_3), \Sigma_1 \sew{j}{k} (\Sigma_2 \sew{l}{m} \Sigma_3)\Big) 
		&& \textnormal{[by~\eqref{eq:restriction_function_symmetry}]}
		\\ = \; &
		\gonlyan\Big(\looptwo, \Sigma_2 \sew{l}{m} \Sigma_3\Big)
		- \gonlyan\Big(\looptwo, \Sigma_1 \sew{j}{k} (\Sigma_2 \sew{l}{m} \Sigma_3)\Big), 
		&& \textnormal{[by~\eqref{eq:an_coboundary}]}
	\end{aligned}
	\end{equation}
	where $\loopone$ is the seam between $\Sigma_1$ and $\Sigma_2$, and $\looptwo$ is the seam between $\Sigma_2$ and $\Sigma_3$.
	Lastly, the assumption of \gls{fsmooth}ness on analytic configurations (Definition~\ref{def:restriction_function_frsmooth}) implies that the cocycles, and thus $\RMFE$ and $Z$, are \gls{fsmooth}.
\end{proof}

\begin{proof}[Proof of Theorem~\ref{thm:restriction_function}]
	By Proposition~\ref{prop:restriction_function_to_modular_functor}, there are local \gls{fsmooth} $\Rp$-extensions $\RMFE$ and $\RMFD$ with reparametrization invariant trivializations $Z$ and $W$ associated to $f_1$ and $f_2$, respectively.
	Since the central charge agrees, Theorem~\ref{thm:universality} yields an isomorphism $\Psi$ from~$\RMFE$ to~$\RMFD$.
	Set $\Psi(Z(\Sigma)) = e^{B(\Sigma)} W(\Sigma)$ for functions $B \in \functions{\moduligb}$, which by Proposition~\ref{prop:reparam_indep} are independent of the boundary parametrizations.

	The restriction functions on analytic loops, determined by functions $\gonlyan_1$ and $\gonlyan_2$ from Equation~\eqref{eq:an_coboundary} in Proposition~\ref{prop:an_coboundary}, may be reconstructed from the cocycles by choosing an analytic parametrization $\vartheta$ of $\loopzero \in \anloops{\Sigma}$, and by equipping $\Sigma$ with arbitrary analytic boundary parametrizations.
	Then, in terms of the unraveling (cutting) operation defined in~\cite[Section~4.2]{Maibach-Peltola:Complex_deformations_of_the_circle}, the identity
	\begin{equation}
		\gonlyan_1(\loopzero, \Sigma)
		= \begin{cases}
			\Omega_{1, 1}^Z(\unravelplus{\vartheta} \Sigma, \unravelminus{\vartheta} \Sigma), & \textnormal{$\loopzero$ a separating loop,} \\
			\Omega_{1, 2}^Z(\unravel{\vartheta} \Sigma), & \textnormal{$\loopzero$ a non-separating loop,}
		\end{cases}
		\label{eq:cocycle_gonlyan}
	\end{equation}
	and the analogous formula for $\gonlyan_2$, are independent of the choice of $\vartheta$ and of the boundary parametrizations, by reparametrization invariance of $Z$ and $W$.

	The cohomology relations of the cocycles are
	\begin{equation}
	\begin{aligned}
		\Omega^W_{j,k}(\Sigma_1, \Sigma_2) 
		&= \Omega^Z_{j,k}(\Sigma_1, \Sigma_2) + B(\Sigma_1) + B(\Sigma_2) - B(\Sigma_1 \sew{j}{k} \Sigma_2), \\
		\Omega^W_{j,k}(\Sigma) 
		&= \Omega^Z_{j,k}(\Sigma) + B(\Sigma) - B(\sewself{j}{k} \Sigma).
	\end{aligned}
	\end{equation}
	These define a function $g$ relating $f_1$ and $f_2$ as in Equation~\eqref{eq:restrictoin_function_equiv} on analytic loops:
	\begin{equation}
	\begin{aligned}
		g(\loopzero, \Sigma) 
		&= \gonlyan_1(\loopzero, \Sigma) - \gonlyan_2(\loopzero, \Sigma) \\
		&= \begin{cases}
			B(\unravelplus{\vartheta} \Sigma) + B(\unravelminus{\vartheta} \Sigma) - B(\Sigma) , & \textnormal{$\loopzero$ separating,} \\
			B(\unravel{\vartheta} \Sigma) - B(\Sigma) , & \textnormal{$\loopzero$ non-separating.}
		\end{cases}
	\end{aligned}
	\label{eq:f1_f2_equiv}
	\end{equation}
	This shows that $f_1$ and $f_2$ are equivalent on analytic loops.
	Moreover, since $B$ does not depend on boundary parametrizations, the function $g$ is well-defined on $\cloops{\Sigma}$.
	By continuity of $f_1$ and $f_2$, as $\anloops{\Sigma}$ is dense in $\cloops{\Sigma}$ (Lemma~\ref{lemma:an_c_inclusion}), the relation~\eqref{eq:restrictoin_function_equiv} extends to the whole of~$\cloops{\Sigma}$, proving the equivalence of $f_1$ and $f_2$.
This concludes the proof.
\end{proof}

\bibliographystyle{annotate}
\newcommand{\etalchar}[1]{$^{#1}$}

\end{document}